\documentclass[twoside]{IEEEtran}

\usepackage[
  pdfusetitle,
  colorlinks,
  linkcolor = blue,
  citecolor = blue,
  urlcolor = blue]{hyperref}

\usepackage{amsmath}
\usepackage{amsfonts}
\usepackage{amsthm}
\usepackage{amssymb}
\usepackage[capitalize]{cleveref}
\usepackage[normalem]{ulem}
\usepackage{graphicx}
\usepackage{pbox}
\usepackage[T1]{fontenc}

\pdfoutput=1

\newtheorem{theorem}{Theorem}
\newtheorem{proposition}[theorem]{Proposition}
\newtheorem{lemma}[theorem]{Lemma}
\newtheorem{corollary}[theorem]{Corollary}

\crefname{conjecture}{conjecture}{conjectures}
\Crefname{conjecture}{Conjecture}{Conjectures}
\newtheorem{definition}[theorem]{Definition}

\Crefname{observation}{Observation}{Observations}

\crefname{secinapp}{appendix}{appendices}
\Crefname{secinapp}{Appendix}{Appendices}

\newcommand{\abs}[1]{\ensuremath{{\lvert#1\rvert}}}
\newcommand{\bra}[1]{\ensuremath{\left\langle{#1}\right|}}
\newcommand{\ket}[1]{\ensuremath{\left|{#1}\right\rangle}}
\newcommand{\braket}[2]{\ensuremath{\left\langle{#1}\middle|{#2}\right\rangle}}
\newcommand{\ip}[1]{\left< #1 \right>}
\newcommand{\braopket}[3]{\ensuremath{\left\langle{#1}\middle|{#2}\middle|{#3}\right\rangle}}
\newcommand{\braopketnolr}[3]{\langle #1 | #2 | #3 \rangle}
\newcommand{\rank}[0]{\ensuremath{\textnormal{rank}}}

\newcommand{\linspan}[0]{\ensuremath{\textnormal{span}}}
\newcommand{\diag}[0]{\ensuremath{\textnormal{diag}}}
\newcommand{\Tr}[0]{\ensuremath{\textnormal{Tr}}}
\newcommand{\linop}[1]{\ensuremath{\mathcal{L}(#1)}}
\newcommand{\proj}[1]{{\ket{#1}\bra{#1}}}
\newcommand{\conj}[1]{\overline{#1}}

\newcommand{\homm}[0]{\to}
\newcommand{\home}[0]{\stackrel{*}{\to}}

\newcommand{\djp}[0]{*}

\newcommand{\cart}{\mathbin{\square}}

\newcommand{\chan}[1]{\ensuremath{\mathcal{#1}}}

\newcommand{\chanadj}[1]{\ensuremath{\mathcal{#1}^*}}

\newcommand{\chanconj}[1]{\ensuremath{\conj{\mathcal{#1}}}}
\newcommand{\chanadjconj}[1]{\ensuremath{\conj{\mathcal{#1}^*}}}
\newcommand{\cN}[0]{\chan{N}}

\newcommand{\gcmpl}[1]{\overline{#1}}
\newcommand{\ot}[0]{\otimes}
\newcommand{\Gc}[0]{{\gcmpl{G}}}
\newcommand{\Hc}[0]{{\gcmpl{H}}}

\newcommand{\thm}[0]{\vartheta'}
\newcommand{\thp}[0]{\vartheta^+}
\newcommand{\thbar}[0]{\overline{\vartheta}}
\newcommand{\thmbar}[0]{\overline{\vartheta}'}
\newcommand{\thpbar}[0]{\overline{\vartheta}^+}
\newcommand{\qth}[0]{\tilde{\vartheta}}
\newcommand{\qthperp}[0]{\overline{\vartheta}}
\newcommand{\qthpperp}[1]{\overline{\vartheta}^+_{#1}}
\newcommand{\qthmperp}[1]{\overline{\vartheta}'_{#1}}

\newcommand{\opnorm}[1]{\ensuremath{\left\lVert#1\right\rVert}}

\newcommand{\bvec}[0]{\mathbf{b}}
\newcommand{\cvec}[0]{\mathbf{c}}
\newcommand{\xvec}[0]{\mathbf{x}}
\newcommand{\yvec}[0]{\mathbf{y}}

\newcommand{\rot}[0]{\mathcal{R}}
\newcommand{\sep}[0]{\textnormal{SEP}}
\newcommand{\cone}[0]{\mathcal{C}}

\newcommand{\psdcone}[0]{\mathcal{S}^+}
\newcommand{\nonneg}[0]{\mathcal{N}}
\newcommand{\ppt}[0]{\textnormal{PPT}}

\newcommand{\cliquec}[0]{\omega}
\newcommand{\cliqueq}[0]{\omega_q}
\newcommand{\cliquece}[0]{\omega_*}
\newcommand{\cliqueqe}[0]{\omega_{q*}}
\newcommand{\indepc}[0]{\alpha}
\newcommand{\indepq}[0]{\alpha_q}
\newcommand{\indepce}[0]{\alpha_*}
\newcommand{\indepqe}[0]{\alpha_{q*}}
\newcommand{\chromc}[0]{\chi}
\newcommand{\chromq}[0]{\chi_q}
\newcommand{\chromce}[0]{\chi_*}
\newcommand{\chromqe}[0]{\chi_{q*}}
\newcommand{\orthrank}[0]{\xi}

\begin{document}

\title{Quantum zero-error source-channel coding and non-commutative graph theory}
\author{Dan~Stahlke
    \thanks{
        Dan Stahlke was with the Department of Physics, Carnegie Mellon University,
        Pittsburgh, Pennsylvania 15213, USA
        (e-mail: dan@stahlke.org)
    }
    \thanks{
        Copyright (c) 2014 IEEE. Personal use of this material is permitted.  However,
        permission to use this material for any other purposes must be obtained from the
        IEEE by sending a request to pubs-permissions@ieee.org.
    }
}
\date{\today}
\maketitle

\begin{abstract}
    Alice and Bob receive a bipartite state (possibly entangled) from some finite
    collection or from some subspace.  Alice sends a message to Bob through a noisy
    quantum channel such that Bob may determine the initial state, with zero chance
    of error.
    This framework encompasses, for example, teleportation, dense coding, entanglement
    assisted quantum channel capacity, and one-way communication complexity of function
    evaluation.

    With classical sources and channels, this problem can be analyzed using graph
    homomorphisms.  We show this quantum version can be analyzed using homomorphisms
    on non-commutative graphs (an operator space generalization of graphs).
    Previously the Lov{\'a}sz $\vartheta$ number has been generalized to non-commutative
    graphs; we show this to be a homomorphism monotone, thus providing bounds on quantum
    source-channel coding.  We generalize the Schrijver and Szegedy numbers, and show
    these to be monotones as well.  As an application we construct a quantum channel
    whose entanglement assisted zero-error one-shot capacity can only be unlocked by using
    a non-maximally entangled state.

    These homomorphisms allow definition of a chromatic number for
    non-commutative graphs.  Many open questions are presented regarding the possibility
    of a more fully developed theory.
\end{abstract}


\IEEEpeerreviewmaketitle

\section{Introduction}

We investigate a quantum version of zero-error source-channel coding (communication over a
noisy channel with side information).  This includes such problems as zero-error
quantum channel capacity (with or without entanglement
assistance)~\cite{dsw2013,arxiv:0906.2527,arxiv:1301.1166,arxiv:0709.2090},
dense coding~\cite{PhysRevLett.69.2881}, teleportation~\cite{PhysRevLett.70.1895},
function evaluation using one-way (classical or quantum)
communication~\cite{witsen76,dewolfphd}, and measurement of bipartite states
using local operations and one-way communication (LOCC-1)~\cite{PhysRevA.88.062316}.
Unless otherwise mentioned all discussion is in the context of zero-error
information theory---absolutely no error is allowed.

The problem we consider is as follows.  Alice and Bob each receive half of a bipartite state
$\ket{\psi_i}$ from some finite collection that has been agreed to in advance (the
\emph{source}).  Alice sends
a message through a noisy quantum channel, and Bob must determine $i$ using Alice's noisy
message and his half of the input $\ket{\psi_i}$.
The goal is to determine whether such a protocol is possible
for a given collection of input states and a given noisy channel.
One may also ask how many channel uses are needed per input state if
several different input states arrive in parallel and are coded using a block code.
This is known as the \emph{cost rate}.
We also consider a variation in which the discrete index $i$ is replaced by a quantum register.

For classical inputs and a classical channel, source-channel coding is possible if and only if
there is a graph homomorphism between two suitably defined graphs.
Since the Lov{\'a}sz $\vartheta$ number of a graph is a homomorphism monotone, it provides a
lower bound on the cost rate~\cite{1705019}.
This bound also applies if Alice and Bob can make use of an entanglement
resource~\cite{6994835,6880319}.
We extend the notion of graph homomorphism to non-commutative graphs and show the
generalized Lov{\'a}sz $\vartheta$ number of~\cite{dsw2013} to be monotone under these
homomorphisms, providing a lower bound on cost rate for quantum source-channel coding.

Schrijver's $\thm$ and Szegedy's $\thp$, which are variations on Lov{\'a}sz's $\vartheta$, are
also homomorphism monotones.
We generalize these for non-commutative graphs, providing stronger bounds on
one-shot quantum channel capacity in particular and on quantum source-channel coding in
general.
Although $\thm$ and $\thp$ provide only mildly stronger bounds as compared to $\vartheta$
for classical graphs, with non-commutative graphs the differences are often dramatic.
For classical graphs $\thm$ and $\thp$ are monotone under entanglement assisted
homomorphisms~\cite{6880319}, but oddly this is not the case for non-commutative graphs.
As a consequence, these quantities can be used to study the power of entanglement assistance.
We construct a channel with large one-shot entanglement assisted capacity but no one-shot
capacity when assisted by a maximally entangled state.

In \cref{sec:classical} we review graph theory and (slightly generalized) classical
source-channel coding.
In \cref{sec:ncgraph} we review the theory of non-commutative graphs and define a
homomorphism for these graphs.
In \cref{sec:qsrcchan} we build the theory of quantum source-channel coding and provide a
few basic examples.
In \cref{sec:qthmon} we prove that $\vartheta$ is monotone under entanglement assisted
homomorphisms of non-commutative graphs.
In \cref{sec:repetitions} we consider block coding and define various products on
non-commutative graphs.
In \cref{sec:schrijver} we define Schrijver $\thm$ and Szegedy $\thp$ numbers for
non-commutative graphs;  we then revisit some examples from the
literature and also show that one-shot entanglement assisted capacity for a quantum
channel can require a non-maximally entangled state.
We conclude with a list of many open questions in \cref{sec:conclusion}.

\subsection{Relation to prior work}

Zero-error source-channel coding in a quantum context was first considered
in~\cite{6994835} then in~\cite{6880319}.  There the sources
and channels are classical but an entanglement resource is available.
Zero-error entanglement assisted capacity of quantum channels was
considered in~\cite{dsw2013}, but without sources.
Measurement of bipartite states using one-way classical communication was considered
in~\cite{PhysRevA.88.062316}; however, this was not in the context of source-channel
coding.
We consider for the first time (in a zero-error context) quantum sources, and consider
their transmission using quantum channels.
To this end we apply the concept of a non-commutative graph, first conceived
in~\cite{dsw2013}, to characterize a quantum source.
This is new, as previously only classical sources were considered so only classical graphs
were needed.
This gives novel perspective even for classical sources with entanglement assistance:
using a non-commutative graph allows to consider the entanglement as part of the source.
In~\cite{dsw2013,6994835,6880319} entanglement was considered
separate from the source, with no framework available for investigating the type of entanglement needed.

Graph homomorphisms are central to classical source-channel coding~\cite{1705019}.  This
concept has been extended to the entanglement-assisted case but still with classical
graphs~\cite{6994835,6880319}.
We define graph homomorphisms for non-commutative graphs, potentially opening a path for a
more developed theory.  Already this leads to a chromatic number for non-commutative
graphs; previously only independence number was defined~\cite{dsw2013}.
Subsequent to first submission of the present paper, an alternative definition has been
provided for the chromatic number of a non-commutative graph~\cite{arxiv:1411.7666}.

The Lov{\'a}sz, Schrijver, and Szegedy numbers were known to provide bounds on classical
source-channel coding~\cite{1705019,de2013optimization}.
These bounds were recently shown to hold also when entanglement assistance is
allowed~\cite{6994835,6880319}.
A Lov{\'a}sz number has been defined for non-commutative graphs, providing a bound on
zero-error entanglement assisted capacity of a quantum channel~\cite{dsw2013}.
We show this generalized Lov{\'a}sz number also provides a bound on quantum source-channel
coding.
Inspired by~\cite{dsw2013}, we provide analogous generalizations for the Schrijver and
Szegedy numbers.  Such a generalization is non-obvious as it involves a basis-independent
reformulation of entrywise positivity constraints on a matrix.
We show these generalized Schrijver and Szegedy numbers provide bounds on quantum
source-channel coding, effectively providing a fully quantum generalization
of~\cite{6994835,6880319}.
Interestingly, these generalized quantities become sensitive to the nature of the
entanglement resource.  This leads to a counterintuitive result: a quantum
channel whose one-shot zero-error capacity can only be unlocked by using a non-maximally
entangled state.  The existence of a classical channel with such a property is still an
open question.

These generalized Lov{\'a}sz, Schrijver, and Szegedy numbers are used to reproduce well
known bounds on dense coding and teleportation, as well as results
from~\cite{arxiv:0906.2527,PhysRevA.88.062316}.  Although our techniques yield less
direct proofs than those previously known, it is notable that such diverse results can be
reproduced using a single technique.

Parallel to entanglement assisted communication runs the subfield of
quantum non-locality games.  In the context of entanglement assisted
communication one has the entanglement assisted independence and chromatic numbers and
entanglement assisted homomorphisms~\cite{6994835,6880319}.  In the
context of quantum non-locality one has quantum independence and chromatic numbers and
quantum homomorphisms~\cite{david2006quantum,cameron2007quantum,arxiv:1212.1724}.
These concepts are mathematically similar, and indeed it is an open question whether they
are identical.
Given this similarity, it seems feasible that the work of the present paper could have an
analogue in non-locality games.  This is beyond the present scope and is left as a potential
direction for further research.

\section{Classical source-channel coding}
\label{sec:classical}

We will make use of the following graph theory terminology.
A graph $G$ consists of a finite set of \emph{vertices} $V(G)$ along with a symmetric binary
relation $x \sim_G y$ (the \textrm{edges} of $G$).
The absence of an edge is denoted $x \not\sim_G y$.
The subscript will be omitted when the graph can be inferred from context.
We allow loops on vertices.  That is to say, we allow $x \sim x$ for some of the $x \in V(G)$.
Typically we will be dealing with graphs that do not have loops (\emph{simple graphs}),
but allow the possibility due to the utility and insight that loops will afford.
We will note the subtleties that this causes, as they arise.
We denote by $\Gc$ the \emph{complement} of $G$, having vertices $V(G)$ and edges
$x \sim_\Gc y \iff x \not\sim_G y, x \ne y$.
For graphs with loops it is also common to use as the complement the graph with edges
$x \sim_\Gc y \iff x \not\sim_G y$.  Fortunately, we will only consider the complement of
loop graphs that have loops on all vertices, and in this case the two definitions
coincide.
A \emph{clique} is a set of vertices $C \subseteq V(G)$ such that $x \sim y$ for all $x, y
\in C, x \ne y$.  An \emph{independent set} is a clique of $\Gc$, equivalently a set
$C \subseteq V(G)$ such that $x \not\sim y$ for all $x, y \in C, x \ne y$.
The \emph{clique number} $\omega(G)$ is the size of the largest clique, and the
\emph{independence number} $\alpha(G)$ is the size of the largest independent set.
A \emph{proper coloring} of $G$ is a map $f : G \to \{1,\cdots,n\}$ (an assignment of
\emph{colors} to the vertices of $G$) such that $f(x) \ne f(y)$
whenever $x \sim y$ (note that this is only possible for graphs with no loops).
The \emph{chromatic number} $\chi(G)$ is the smallest possible number of colors needed.
If no proper coloring exists (i.e.\ if $G$ has loops) then $\chi(G) = \infty$.
The \emph{complete graph} $K_n$ has vertices $\{ 1,\cdots,n \}$ and edges
$x \sim y \iff x \ne y$ (note in particular that $K_n$ does not have loops).
$G$ is a \textit{subgraph} of $H$ if $V(G) \subseteq V(H)$ and
$x \sim_G y \implies x \sim_H y$.

Suppose Alice wishes to send a message to Bob through a noisy classical channel $\cN : S \to V$
such that Bob can decode Alice's message with zero chance of error.
How big of a message can be sent?
Denote by $\cN(v|s)$ the probability that sending $s \in S$ through $\cN$ will result in Bob
receiving $v \in V$, and define the graph $H$ with vertex set $S$ and with edges
\begin{align}
    \label{eq:distinguishgraph}
    s \sim_H t \iff \cN(v|s) \cN(v|t) = 0 \textrm{ for all } v \in V.
\end{align}
Two codewords $s$ and $t$ can be distinguished with certainty by Bob if they are
never mapped to the same $v$.
Therefore, the largest set of distinguishable codewords corresponds to the largest clique in
$H$, and the number of such codewords is the clique number $\omega(H)$.
We will call $H$ the \emph{distinguishability graph} of the channel $\cN$.
It is traditional to instead deal with the \emph{confusability graph}, of
which~\eqref{eq:distinguishgraph} is the complement.
We choose to break with this tradition as this will lead to cleaner notation.
Also the distinguishability graph has the advantage of not having loops, making it more
natural from a graph-theoretic perspective.
In order to facilitate comparison to prior results we will sometimes speak of
$\alpha(\Hc)$ rather than $\omega(H)$ (note that these are equal).

If Bob already has some side information regarding the message Alice wishes to send, the
communication task becomes easier: the number of codewords is no longer limited to $\omega(H)$.
This situation is known as \emph{source-channel coding}.
We will use a slightly generalized version of source-channel coding, as this will aid in the
quantum generalization in \cref{sec:qsrcchan}.
Suppose Charlie chooses a value $i$ and sends a value $x$ to Alice and $u$ to Bob with
probability $P(x,u|i)$.
Alice sends Bob a message through a noisy channel.
Bob uses Alice's noisy message, along
with his side information $u$, to deduce Charlie's input $i$ (\cref{fig:coding}).
This reduces to standard source-channel coding if $P(x,u|i) \ne 0$ only when $x=i$.
In other words, the standard scenario has no Charlie, $x$ and $u$ come in with probability $P(x,u)$,
and Bob is supposed to produce $x$.

\begin{figure}
    \centering
    \includegraphics{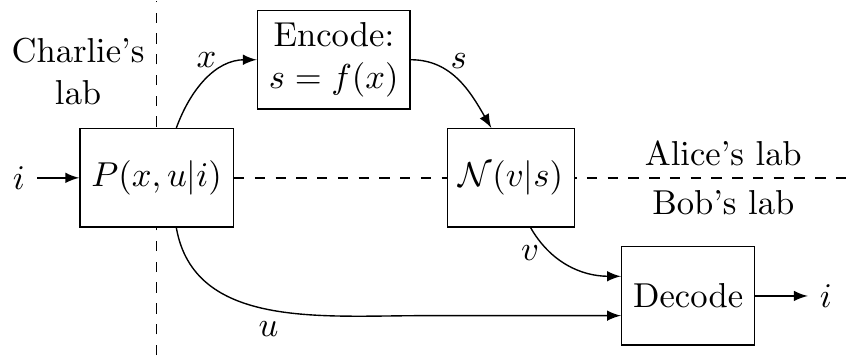}
    \caption{
        Zero-error source-channel coding.
    }
    \label{fig:coding}
\end{figure}

There are a number of reasons one might wish to consider such a scenario.
For instance, suppose that $x=i$ always.  The side information $u$ might have originated from a
previous noisy transmission of $x$ from Alice to Bob.  The goal is to resend using channel
$\chan{N}$ in order to fill in the missing information.
Or, the communication complexity of bipartite function evaluation fits into this model.
Suppose that Alice and Bob receive $x$ and $u$, respectively, from a referee Charlie.  Alice
must send a message to Bob such that Bob may evaluate some function $g(x,u)$.
To fit this into the model of \cref{fig:coding}, imagine that Charlie first chooses a value $i$
for $g$, then sends Alice and Bob some $x,u$ pair such that $g(x,u)=i$.
From the perspective of Alice and Bob, determining $i$ is equivalent to evaluating $g(x,u)$.
One may ask how many bits Alice needs to send to Bob to accomplish this.

In general, Alice's strategy is to encode her input $x$ using some function $f : X \to S$
before sending it through the channel (a randomized strategy never helps when zero-error
is required).
As before, Bob receives a value $v$ with probability $\cN(v|s)$.
The values $u$ and $v$ must be sufficient for Bob to compute $i$.
For a given $u$, Bob knows Alice's input comes from the set
$\{x : \exists i \textrm{ such that } P(x,u|i) \ne 0\}$.
Bob only needs to distinguish between the values of $x$ corresponding to different $i$, since
his goal is to determine $i$.
Define a graph $G$ with vertices $V(G)=X$ and with edges between Alice inputs that Bob
sometimes needs to distinguish:
\begin{align}
    \label{eq:chargraph}
    x \sim_G y \iff \exists u, \exists i \ne j \textrm{ s.t. }
        P(x,u|i) P(y,u|j) \ne 0.
\end{align}
This is the \emph{characteristic graph} of the source $P$.
If Bob must sometimes distinguish $x$ from $y$ then Alice's encoding must ensure that $x$ and
$y$ never get mapped to the same output by the noisy channel.
In other words, her encoding must satisfy $f(x) \sim_H f(y)$ whenever $x \sim_G y$.
By definition, this is possible precisely when $G$ is homomorphic to $H$.

\begin{definition}
    \label{def:homc}
    Let $G$ and $H$ be graphs without loops.
    $G$ \emph{is homomorphic to} $H$, written $G \homm H$, if there is
    a function $f : V(G) \to V(H)$ such that $x \sim y \implies f(x) \sim f(y)$.
    The function $f$ is said to be a \emph{homomorphism from $G$ to $H$}.
\end{definition}

Graph homomorphisms are examined in great detail in~\cite{hahn1997graph,HellNesetril200409}.
We state here some basic facts that can be immediately verified.
\begin{proposition}
    \label{thm:homfacts}
    Let $F,G,H$ be graphs without loops.
    \begin{enumerate}
        \item If $F \homm G$ and $G \homm H$ then $F \homm H$.
        \item If $G$ is a subgraph of $H$ then $G \homm H$.
        \item The clique number $\omega(H)$
            is the largest $n$ such that $K_n \homm H$.
        \item The chromatic number $\chi(G)$ is the smallest $n$ such
            that $G \homm K_n$.
    \end{enumerate}
\end{proposition}

The above arguments can be summarized as follows.

\begin{proposition}
    \label{thm:src_chan_hom}
    There exists a zero-error source-channel coding protocol for source $P(x,u|i)$ and channel
    $\chan{N}(v|s)$ if and only if $G \homm H$ where
    $G$ is the characteristic graph of the source,~\eqref{eq:chargraph}, and
    $H$ is the distinguishability graph of the channel,~\eqref{eq:distinguishgraph}.
\end{proposition}

As required by \cref{def:homc}, neither $G$ nor $H$ have loops.
More precisely, $G$ has a loop if and only if there is an $x,u$ that can occur for two
different inputs by Charlie.  In this case it is impossible for Alice and Bob to recover
Charlie's input, no matter how much communication is allowed.

We emphasize that, although we refer to source-channel coding and use the associated
terminology, we are actually considering something
a bit more general since we use a source $P(x,u|i)$, with Bob answering $i$, rather than a
source $P(x,u)$, with Bob answering $x$.
Standard source-channel coding, which can be recovered by setting
$P(x,u|i) \ne 0 \iff x=i$, was characterized in terms of graph homomorphisms in~\cite{1705019}.
Our generalization does not substantially change the theory,\footnote{
    Although, for our generalization extra care needs to be taken when considering
    block coding.  This will be discussed in \cref{sec:repetitions}.
} and will allow a smoother transition to the quantum version (in the next section).

The Lov{\'a}sz number of the complementary graph, $\thbar(G) = \vartheta(\Gc)$, is given
by the following dual (and equivalent) semidefinite
programs:~\cite{lovasz79,lovaszsemidef}\footnote{
    The first of these follows from theorem 6 of~\cite{lovasz79} by setting
    $T=A/\abs{\lambda_n(A)}$ (note that in~\cite{lovasz79} vertices are considered
    adjacent to themselves).
    The second comes from page 167 of~\cite{lovaszsemidef}, or from
    theorem 3 of~\cite{lovasz79} by taking $Z=\lambda I-A+J$ with $\lambda$ being
    the maximum eigenvalue of $A$.
}
\begin{align}
    \thbar(G) &= \max\{ \opnorm{I + T} : I + T \succeq 0, \notag
        \\ &\hphantom{= \max\{\;}
        T_{ij}=0 \textnormal{ for } i \not\sim j \},
    \\ \thbar(G) &= \min\{ \lambda : \exists Z \succeq J, Z_{ii} = \lambda, \notag
        \\ &\hphantom{= \max\{\;}
        Z_{ij} = 0 \textnormal{ for } i \sim j \},
\end{align}
where we assume that $G$ has no loops.
The norm here is the operator norm (equal to the largest singular value),
$J$ is the matrix with every entry equal to 1, and
$Z \succeq J$ means that $Z-J$ is positive semidefinite.
This quantity is a homomorphism monotone in the sense
that~\cite{de2013optimization}
\begin{align}
    \label{eq:thcmon}
    G \homm H \implies \thbar(G) \le \thbar(H).
\end{align}
Consequently (see \cref{thm:homfacts}) we have the
Lov{\'a}sz sandwich theorem
\begin{align}
    \label{eq:sandwich}
    \omega(G) \le \thbar(G) \le \chi(G).
\end{align}
Since source-channel coding is only possible when $G \homm H$, it follows that
$\thbar(G) \le \thbar(H)$ is a necessary condition.
Two related quantities, Schrijver's $\thmbar$ and Szegedy's $\thpbar$, which will be
defined in \cref{sec:schrijver}, have similar monotonicity
properties~\cite{de2013optimization} so they provide similar bounds.

\begin{proposition}
    \label{thm:coding_mon_classical}
    One-shot source-channel coding is possible only if
    $\thbar(G) \le \thbar(H)$, $\thmbar(G) \le \thmbar(H)$, and
    $\thpbar(G) \le \thpbar(H)$, where $G$ is the characteristic graph of the
    source,~\eqref{eq:chargraph}, and $H$ is the distinguishability graph of the
    channel,~\eqref{eq:distinguishgraph}.
\end{proposition}

Traditionally, source-channel coding has been studied in the case where
$P(x,u|i) \ne 0$ only when $x=i$.
In this case, the following bound holds~\cite{1705019}:\footnote{
    Actually,~\cite{1705019} seems to have stopped just short of stating such a bound,
    although they lay all the necessary foundation.
}
\begin{proposition}
    \label{thm:costrate_classical}
    Suppose $P(x,u|i) \ne 0$ only when $x=i$ and let graphs $G$ and $H$ be given
    by~\eqref{eq:chargraph} and~\eqref{eq:distinguishgraph}.
    Then $m$ parallel instances of the source can be sent using $n$ parallel instances of
    the channel only if
    \begin{align*}
        \frac{n}{m} \ge \frac{\log \thbar(G)}{\log \thbar(H)}.
    \end{align*}
\end{proposition}

We will always take logarithms to be base 2.
The infimum of $n/m$ (equivalent to the limit as $m \to \infty$) is known as the
\emph{cost rate}; \cref{thm:costrate_classical} can be interpreted as an upper bound on
the cost rate.
This bound relies on the fact that $\thbar$ is multiplicative under various graph
products, a property not shared by $\thmbar$ or $\thpbar$.
\Cref{thm:coding_mon_classical,thm:costrate_classical}
apply also to the case of entanglement assisted source-channel coding,
still with classical inputs and a classical channel~\cite{6880319}.
We will later show (\cref{thm:generalized_costrate}) that the condition $P(x,u|i) \ne 0$
only when $x=i$ is not necessary in \cref{thm:costrate_classical}.

With some interesting caveats, these two theorems in fact also apply to a generalization
of source-channel coding in which the source produces bipartite entangled states and in
which the channel is quantum.
The rest of this paper is devoted to development of this theory.

\section{Non-commutative graph theory}
\label{sec:ncgraph}

Given a graph $G$ on vertices $V(G) = \{1,\dots,n\}$ we may define the operator space
\begin{align}
    \label{eq:SfromG}
    S &= \linspan\{ \ket{x}\bra{y} : x \sim y \} \subseteq \linop{\mathbb{C}^n}
\end{align}
where $\ket{x}$ and $\ket{y}$ are basis vectors from the standard basis.
Because we consider symmetric rather than directed graphs, this space is Hermitian:
$A \in S \iff A^\dag \in S$ (more succinctly, $S = S^\dag$).
If $G$ has no loops, $S$ is trace-free (it consists only of trace-free operators).
If $G$ has loops on all vertices, $S$ contains the identity.

Concepts from graph theory can be rephrased in terms of such operator spaces.
For example, for trace-free spaces the clique number can be defined as the size of the largest
set of nonzero vectors
$\{\ket{\psi_i}\}$ such that $\ket{\psi_i}\bra{\psi_j} \in S$ for all $i \ne j$.
Note that since $S$ is trace-free, these vectors must be orthogonal.
Although not immediately obvious, this is indeed equivalent to $\omega(G)$ when $S$ is defined
as in~\eqref{eq:SfromG}.

Having defined clique number in terms of operator spaces, one can drop the requirement
that $S$ be of the form~\eqref{eq:SfromG} and can speak of the clique number of an arbitrary
Hermitian subspace.
Such subspaces, thought of in this way, are called
\emph{non-commutative graphs}~\cite{dsw2013}.
Note that~\cite{dsw2013} requires $S$ to contain the identity, but we drop this requirement and
insist only that $S=S^\dag$.
Such a generalization is analogous to allowing the vertices of a graph to not have loops.
Dropping also the condition $S = S^\dag$ would give structures analogous to directed graphs,
however we will not have occasion to consider this.

To draw clear distinction between non-commutative graphs and the traditional kind, we will
often refer to the latter as \emph{classical graphs}.
We will say $S$ \emph{derives from a classical graph} if $S$ is of the
form~\eqref{eq:SfromG}.

The distinguishability graph of a quantum channel $\chan{N} : \linop{A} \to
\linop{B}$ with Kraus operators $\{N_i\}$ can be defined as
\begin{align}
    \label{eq:SNiNj}
    T &= (\linspan\{ N_i^\dag N_j : \forall i,j \})^\perp \subseteq \linop{A}
\end{align}
where $\perp$ denotes the perpendicular subspace under the Hilbert--Schmidt inner product
$\ip{X,Y} = \Tr(X^\dag Y)$.
For a classical channel this is equal to~\eqref{eq:SfromG} with $G$ given
by~\eqref{eq:distinguishgraph}.
The space $\linspan\{ N_i^\dag N_j \}$ (the confusability graph) was considered in
\cite{arxiv:0709.2090,arxiv:0906.2527,arxiv:1301.1166,dsw2013}; however, we consider
the perpendicular space for the same reason that we considered the distinguishability rather
than the confusability graph in \cref{sec:classical}: it leads to simpler notation
especially when discussing homomorphisms.
It will be convenient to use the notation
\begin{align*}
    N := \linspan\{ N_i \},
\end{align*}
and likewise for other sets
of Kraus operators so that~\eqref{eq:SNiNj} becomes simply
\begin{align}
    \label{eq:qm_nonconf}
    T = (N^\dag N)^\perp,
\end{align}
with the multiplication of two operator spaces defined to be the linear span of the
products of operators from the two spaces.
Note that the closure condition for Kraus operators gives
$\sum_i N_i^\dag N_i = I \implies I \in N^\dag N \implies I \perp T$.  Therefore $T$
is trace-free.

In~\cite{dsw2013} a generalization of the Lov{\'a}sz $\vartheta(G)$ number was provided for
non-commutative graphs, which they called $\tilde{\vartheta}(S)$.
We present the definition in terms of $\qthperp(S) := \qth(S^\perp)$, which should be
thought of as a generalization of $\thbar(G) = \vartheta(\Gc)$.
\begin{definition}[\cite{dsw2013}]
    \label{def:qth}
    Let $S \subseteq \linop{A}$ be a trace-free non-commutative graph.
    Let $A'$ be an ancillary system of the same dimension as $A$, and define the vector
    $\ket{\Phi} = \sum_i \ket{i}_A \ot \ket{i}_{A'}$.
    Then $\qthperp(S)$ is defined by the following dual (and equivalent) programs:
    \begin{align}
        \label{eq:qth_primal}
        \qthperp(S) &= \max\{ \opnorm{I+X} : X \in S \ot \linop{A'}, I+X \succeq 0 \},
        \\
        \label{eq:qth_dual}
        \qthperp(S) &= \min\{ \opnorm{\Tr_A Y} : Y \in S^\perp \ot \linop{A'},
            Y \succeq \ket{\Phi}\bra{\Phi} \}.
    \end{align}
    We will use the notation $\qth(S^\perp)=\qthperp(S)$.
\end{definition}
When $S$ derives from loop-free graph $G$ via~\eqref{eq:SfromG}, this reduces to the
standard Lov{\'a}sz number: $\qthperp(S) = \thbar(G)$.
Similarly, when $S$ derives from a graph $G$ having loops on all vertices,
$\qth(S) = \vartheta(G)$.
Analogous to the classical case, $\qthperp(S)$ gives an upper bound on the zero-error
capacity of a quantum channel.  In fact, it even gives an upper bound on the zero-error
entanglement assisted capacity~\cite{dsw2013}.

Independence number for non-commutative graphs has been investigated
in~\cite{arxiv:0709.2090,arxiv:0906.2527,arxiv:1301.1166,dsw2013}, and
in~\cite{dsw2013} the authors posed the question of whether further concepts from graph theory
can be generalized as well.
We carry out this program by generalizing graph homomorphisms, which will in turn lead to a
chromatic number for non-commutative graphs.
These generalized graph homomorphisms will characterize quantum source-channel
coding in analogy to \cref{thm:src_chan_hom}.
In fact, one could \emph{define} non-commutative graph homomorphisms as being the relation that
gives a generalization of \cref{thm:src_chan_hom}, but we choose instead to provide more direct
justification for our definition.

We begin by describing ordinary graph homomorphisms in terms of operator spaces of the
form~\eqref{eq:SfromG}; this will lead to a natural generalization to non-commutative graphs.
Suppose that $S \subseteq \linop{A}$ and $T \subseteq \linop{B}$ are derived from graphs
$G$ and $H$ via~\eqref{eq:SfromG}, and consider a function $f : V(G) \to V(H)$.
In terms of $S$ and $T$, the homomorphism condition $x \sim_G y \implies f(x) \sim_H f(y)$
becomes
\begin{align}
    \label{eq:qhom1}
    \ket{x}\bra{y} \in S \implies \ket{f(x)}\bra{f(y)} \in T,
\end{align}
where $\ket{x}$ and $\ket{y}$ are vectors from the standard basis.
Consider the classical channel that maps $x \to f(x)$.
Viewed as a quantum channel, this can be written as the superoperator
$\chan{E} : \linop{A} \to \linop{B}$ with the action
$\chan{E}(\ket{x}\bra{x}) = \ket{f(x)}\bra{f(x)}$.
The Kraus operators of this channel are $E_x = \ket{f(x)}\bra{x}$.
Again using the notation $E = \linspan\{ E_i \}$,~\eqref{eq:qhom1} can be written
$E S E^\dag \subseteq T$.
The generalization to non-commutative graphs is obtained by dropping the condition that
$\chan{E}$ be a classical channel, allowing instead arbitrary completely positive trace preserving
(CPTP) maps.

\begin{definition}
    \label{def:homq}
    Let $S \subseteq \linop{A}$ and $T \subseteq \linop{B}$ be trace-free non-commutative
    graphs.
    We write $S \homm T$ if there exists a completely positive trace preserving (CPTP)
    map $\chan{E} : \linop{A} \to \linop{B}$
    with Kraus operators $\{E_i\}$ such that
    \begin{align}
        E S E^\dag &\subseteq T \textrm{ or, equivalently,}
        \label{eq:homq_def1}
        \\ E^\dag T^\perp E &\subseteq S^\perp.
        \label{eq:homq_def2}
    \end{align}
    Equivalently, $S \to T$ if and only if there is a Hilbert space $C$ and an isometry
    $J : A \to B \ot C$ such that
    \begin{align}
        &JSJ^\dag \subseteq T \ot \linop{C} \textrm{ or, equivalently,}
        \label{eq:homq_def3}
        \\ &J^\dag (T^\perp \ot \linop{C}) J \subseteq S^\perp.
        \label{eq:homq_def4}
    \end{align}
    We will say that the subspace $E$, or the Kraus operators $\{E_i\}$, or the isometry
    $J$, is a \emph{homomorphism from $S$ to $T$}.
\end{definition}

That~\eqref{eq:homq_def1}-\eqref{eq:homq_def4} are equivalent can be seen as follows.
$\eqref{eq:homq_def1} \iff (\Tr\{ese'^\dag t'\}=0 \; \forall e,e'\in E, s\in S, t'\in T^\perp)
\iff \eqref{eq:homq_def2}$.
Similar reasoning shows $\eqref{eq:homq_def3} \iff \eqref{eq:homq_def4}$, using
$(T \ot \linop{C})^\perp = T^\perp \ot \linop{C}$.
Equivalence of~\eqref{eq:homq_def2} and~\eqref{eq:homq_def4} follows from the fact that
$E=\linspan_{\ket{\phi}}\{ (I \ot \bra{\phi})J \}$ where $J$ is related to $\chan{E}$ by
Stinespring's dilation theorem.

When $S$ and $T$ derive from classical graphs \cref{def:homq} is equivalent to
\cref{def:homc}, as we will now show.

\begin{theorem}
    \label{thm:hom_cq_coincide}
    For non-commutative graphs that derive from classical graphs,
    \cref{def:homc,def:homq} coincide.
    In other words, if $S$ and $T$ derive from graphs $G$ and $H$ according to the
    recipe~\eqref{eq:SfromG} then $G \homm H \iff S \homm T$.
\end{theorem}
\begin{proof}
    Let $S$ and $T$ be non-commutative graphs deriving from classical graphs $G$ and $H$.

    ($\implies$)
    Suppose $G \homm H$.
    By \cref{def:homc} there is an $f : G \to H$ such that
    $x \sim_G y \implies f(x) \sim_H f(y)$.
    Consider the set of Kraus operators $E_x = \ket{f(x)}\bra{x}$.
    Then,
    \begin{align*}
        E S E^\dag
        &= \linspan\{ E_i \ket{x}\bra{y} E_j^\dag : i,j, x \sim_G y \}
        \\ &= \linspan\{ \ket{f(x)}\bra{f(y)} : x \sim_G y \}
        \\ &\subseteq T.
    \end{align*}

    ($\impliedby$)
    Suppose $S \homm T$.
    By \cref{def:homq} there is a channel $\chan{E} : \linop{A} \to \linop{B}$
    such that $E S E^\dag \subseteq T$.
    For each vertex $x$ of $G$, there is an $i(x)$ such that $E_{i(x)} \ket{x}$ does not vanish.
    Pick an arbitrary nonvanishing index of the vector $E_{i(x)} \ket{x}$ and call this
    $f(x)$ so that $\bra{f(x)} E_{i(x)} \ket{x} \ne 0$.

    Now consider any edge $x \sim_G y$.
    We have
    \begin{align*}
        \ket{x}\bra{y} \in S
        &\implies E \ket{x}\bra{y} E^\dag \in T
        \\ &\implies E_{i(x)} \ket{x}\bra{y} E^\dag_{i(y)} \in T.
    \end{align*}
    Define $\tau := E_{i(x)} \ket{x}\bra{y} E^\dag_{i(y)}$.
    Then $\tau \in T$ and
    \begin{align*}
        \braopket{f(x)}{\tau}{f(y)} \ne 0
        &\implies \Tr\{ \tau \ket{f(y)}\bra{f(x)} \} \ne 0
        \\ &\implies \ket{f(x)}\bra{f(y)} \not\in T^\perp
        \\ &\implies \ket{f(x)}\bra{f(y)} \in T
        \\ &\implies f(x) \sim_H f(y).
    \end{align*}
    Therefore $x \sim_G y \implies f(x) \sim_H f(y)$.
\end{proof}

\Cref{def:homq} could be loosened to require only that $\sum_i E_i^\dag E_i$ be
invertible (equivalently $E \ket{\psi} \ne \{0\}$ for all $\ket{\psi}$, equivalently
$J^\dag J$ invertible)
rather than $\chan{E}$ being trace preserving.
\Cref{thm:hom_cq_coincide} would still hold; however, \cref{def:homq} as currently stated
has an operational interpretation in terms of quantum source-channel coding (which we will
introduce in \cref{sec:qsrcchan}) and satisfies the monotonicity relation
$S \homm T \implies \qthperp(S) \le \qthperp(T)$ (which we will show in \cref{sec:qthmon}).
Hilbert space structure seems to be important for non-commutative graphs, so it is
reasonable to require that $J$ preserve this structure (i.e.\ $J$ should be an isometry).

As a guide to the intuition, one should not think of $ESE^\dag$ in~\eqref{eq:homq_def1} as
density operators $\rho \in S$ going into a channel, like $\sum_i E_i \rho E_i^\dag$, but
rather as a mechanism for
comparing the action of the channel on two different states, something like
$\{ E_i \ket{\psi}\bra{\phi} E_j^\dag : \forall i,j \}$ with
$\ket{\psi}\bra{\phi} \in S$.
But this is only a rough intuition, as $S$ might not necessary be composed of dyads
$\ket{\psi}\bra{\phi}$.
The two copies of $E$ here are analogous to the two Kraus operators appearing in the
Knill--Laflamme condition, which we will explore in \cref{sec:qsrcchan}.
Note that $E \ket{\psi}$ is equal to the support of
$\chan{E}(\ket{\psi}\bra{\psi})$.

The non-commutative graph homomorphism of \cref{def:homq} satisfies properties analogous to
those of \cref{thm:homfacts}.

\begin{proposition}
    Let $R, S, T$ be trace-free non-commutative graphs.
    \begin{enumerate}
        \item \label{thm:qhom_facts_trans}
            If $R \homm S$ and $S \homm T$ then $R \homm T$.
        \item \label{thm:qhom_facts_subg}
            If $S \subseteq T$ then $S \homm T$.
            More generally, if $J$ is an isometry and $JSJ^\dag \subseteq T$ then
            $S \homm T$.
    \end{enumerate}
\end{proposition}
\begin{proof}
    \Cref{thm:qhom_facts_trans} follows from considering the composition of channels
    associated with the homomorphisms $R \homm S$ and $S \homm T$.
    \Cref{thm:qhom_facts_subg} follows trivially from~\eqref{eq:homq_def3}, taking space
    $C$ to be trivial (one dimensional).
\end{proof}

(The condition that appears above, $JSJ^\dag \subseteq T$ with $J$ an isometry,
seems to be a reasonable
generalization of the notion of subgraphs for non-commutative graphs, although we won't be
making use of this concept.
Note that~\cite{dsw2013} defined \textit{induced subgraphs} as $J^\dag S J$.
It appears that these two definitions are somewhat incompatible.)

For classical graphs the clique number is the greatest $n$ such that
$K_n \homm G$ and the chromatic number is the least $n$ such that $G \homm K_n$.
We use this to extend these concepts to non-commutative graphs.
In the previous section, the complete graph $K_n$ was defined to have no loops.
The corresponding non-commutative graph, defined via~\eqref{eq:SfromG}, is
$\linspan\{\ket{x}\bra{y} : x \ne y\}$, the space of matrices with zeros on the diagonal.
However, it is reasonable to also consider $(\mathbb{C} I)^\perp$, the space of trace-free
operators.
We consider both.

\begin{definition}
    For $n \ge 1$ define the \emph{classical and quantum complete graphs}
    \begin{align*}
        K_n &= \linspan\{\ket{x}\bra{y} : x \ne y \} \subseteq \linop{\mathbb{C}^n},
        \\
        Q_n &= (\mathbb{C} I)^\perp \subseteq \linop{\mathbb{C}^n}.
    \end{align*}
\end{definition}

One can think of $K_n$ as consisting of the operators orthogonal to the ``classical
loops'' $\proj{x}$ and $Q_n$ as consisting of the operators orthogonal to the ``coherent
loop'' $I$.
We use these to define clique, independence, and chromatic numbers for
non-commutative graphs.
In \cref{sec:qsrcchan} we will see that all of these quantities have operational
interpretations in the context of communication problems.
These quantities, and others, are summarized in \cref{tab:quantities_summary}.

\begin{definition}
    \label{def:omega_xi}
    Let $S$ be a trace-free non-commutative graph.
    We define the following quantities.
    \begin{enumerate}
        \item $\cliquec(S)$ is the greatest $n$ such that $K_n \homm S$
        \item $\cliqueq(S)$ is the greatest $n$ such that $Q_n \homm S$
        \item $\indepc(S^\perp) = \cliquec(S)$ and
            $\indepq(S^\perp) = \cliqueq(S)$.
            Note that $I \in S^\perp$.
        \item $\chromc(S)$ is the least $n$ such that $S \homm K_n$, or $\infty$ if
            $S \not\homm K_n$ for all $n$
        \item $\chromq(S)$ is the least $n$ such that $S \homm Q_n$
    \end{enumerate}
    The quantities $\cliqueq$ and $\chromq$ are not to be confused with the quantities of
    similar name that are discussed in the context of Bell-like nonlocal
    games~\cite{david2006quantum,cameron2007quantum,arxiv:1212.1724}.
\end{definition}

When $S$ derives from a classical graph $G$, our $\cliquec$ and $\chromc$ correspond to the
ordinary definitions of clique number and chromatic number and our $\chromq$
corresponds to the orthogonal rank $\orthrank(G)$.\footnote{
The \emph{orthogonal rank} of a graph is the smallest dimension of a vector space such
that each vertex may be assigned a nonzero vector, with the vectors of adjacent vertices
being orthogonal.
}
This will be proved shortly.
For non-commutative graphs with $I \in S$,
our definition of $\indepc(S)$ and $\indepq(S)$ corresponds to
that of~\cite{dsw2013,arxiv:0906.2527,arxiv:1301.1166,arxiv:0709.2090}, as we will
show in \cref{thm:alpha_indep}.
In other words, when $S=N^\dag N$ is the confusability graph of a channel $\chan{N}$,
$\indepc(S)$ and $\indepq(S)$ correspond to the one-shot classical and quantum capacities;
when $S=(N^\dag N)^\perp$ the same can be said for
$\cliquec(S)$ and $\cliqueq(S)$.

\begin{theorem}
    \label{thm:alpha_xi_reduce}
    Let $S$ be the non-commutative graph associated with a classical loop-free graph $G$.
    Then $\cliquec(S)=\omega(G)$, $\chromc(S)=\chi(G)$, $\chromq(S)=\orthrank(G)$,
    and $\cliqueq(S)=1$.
\end{theorem}
\begin{proof}
    $\cliquec(S)=\omega(G)$ and $\chromc(S)=\chi(G)$ follow directly from
    \cref{def:omega_xi} and \cref{thm:homfacts,thm:hom_cq_coincide}.

    An \emph{orthogonal representation} of $G$ is a map from vertices to nonzero vectors
    such that adjacent vertices correspond to orthogonal vectors.
    The \emph{orthogonal rank}
    $\orthrank(G)$ is defined to be the smallest possible dimension of an orthogonal
    representation.
    Let $\{\ket{\psi_x}\}_{x \in V(G)} \subseteq \linop{\mathcal{C}^n}$ be an
    orthogonal representation of $G$.  Without loss of generality assume these vectors to
    be normalized.  The Kraus operators $E_x = \ket{\psi_x}\bra{i}$
    provide a homomorphism $S \homm Q_n$.  So $\chromq(S) \le \orthrank(G)$.

    Conversely, suppose a set of Kraus operators $\{E_i\}$ provides a homomorphism
    $S \homm Q_n$ with $n = \chromq(S)$.
    Because $\sum_i E_i^\dag E_i = I$, for each $x \in G$ there is an $i(x)$ such that
    $E_{i(x)} \ket{x}$ does not vanish.  Define $\ket{\psi_x} = E_{i(x)} \ket{x}$.
    For any edge $x \sim y$ of $G$ we have
    \begin{align*}
        \ket{x}\bra{y} \in S
        &\implies E \ket{x}\bra{y} E^\dag \in Q_n
        \\ &\implies \ket{\psi_x}\bra{\psi_y} \in Q_n
        \\ &\implies \braket{\psi_x}{\psi_y} = 0.
    \end{align*}
    So $\{\ket{\psi_x}\}_{x \in V(G)}$ is an orthogonal representation of $G$ of dimension
    $n$, giving $\orthrank(G) \le \chromq(S)$.

    $\cliqueq(S)=1$ because it is not possible to have $Q_n \homm S$ if $n > 1$.
    For, suppose that such a homomorphism $E$ existed.
    There must be some $x \in V(G)$ and some $i$ such that $\bra{x}E_i \ne 0$.
    Since $G$ is loop free, $\proj{x} \in S^\perp$ so
    $E_i^\dag \proj{x} E_i \in E^\dag S^\perp E$.
    But $Q_n^\perp = \mathbb{C}I$ contains no rank-1 operators so
    $E^\dag S^\perp E \not\subseteq Q_n^\perp$ and $E$ cannot be a homomorphism from $Q_n$
    to $S$.
\end{proof}

\begin{theorem}
    \label{thm:alpha_indep}
    Let $S \subseteq \linop{A}$ be a non-commutative graph with $I \in S$.
    Then our $\indepc(S)$ and $\indepq(S)$ are equivalent to the independence number and
    quantum independence number
    of~\cite{dsw2013,arxiv:0906.2527,arxiv:1301.1166,arxiv:0709.2090}.
\end{theorem}
\begin{proof}
    This is a consequence of the operational interpretation of non-commutative graph
    homomorphisms which we will prove in \cref{sec:qsrcchan}; however, we give here a
    direct proof.
    The independence number of~\cite{dsw2013} is the largest number of nonzero vectors
    $\{\ket{\psi_i}\}_i$ such that
    \begin{align}
        \label{eq:dsw_alpha}
        \ket{\psi_i}\bra{\psi_j} \in S^\perp \textrm{ when } i \ne j.
    \end{align}
    Given such a collection of $n$ vectors one can define $E_i : \mathbb{C}^n \to A$
    as $E_i = \ket{\psi_i}\bra{i}$.  Since $I \in S$,~\eqref{eq:dsw_alpha} requires
    orthogonal vectors; thus $\sum_i E_i^\dag E_i = I$ so these $\{E_i\}$ are indeed Kraus
    operators.  Now,
    \begin{align*}
        E K_n E^\dag &= \linspan\{ E_{i'} \ket{i}\bra{j} E_{j'} : i \ne j \}
        \\ &= \linspan\{ \ket{\psi_i}\bra{\psi_j} : i \ne j \}
        \subseteq S^\perp,
    \end{align*}
    giving $K_n \to S^\perp$, or $\indepc(S) \ge n$.

    Conversely, take $n=\indepc(S)$.  By the definition of $\indepc(S)$, we have
    $K_n \to S^\perp$.
    Let $\{E_i\}$ be the Kraus operators that satisfy $E K_n E^\dag \subseteq S^\perp$,
    as per \cref{def:homq}.
    Since $\sum_k E_k^\dag E_k = I$, for each $i \in \{1,\dots,n\}$ there must be
    some $k(i)$ such that $E_{k(i)} \ket{i} \ne 0$.
    Define $\ket{\psi_i} = E_{k(i)} \ket{i}$.  Then for $i \ne j$,
    $E K_n E^\dag \subseteq S^\perp \implies \ket{\psi_i}\bra{\psi_j} \in S^\perp$.

    The quantum independence number is the largest rank projector $P$ such that
    $PSP = \mathbb{C}P$.
    Suppose we have such a projector.  Let $n=\rank(P)$ and let $J:\mathbb{C}^n \to A$
    be an isometry such that $J J^\dag = P$.
    Then $J^\dag S J = J^\dag PSP J = \mathbb{C} J^\dag P J = \mathbb{C} I = Q_n^\perp$.
    By~\eqref{eq:homq_def4}, taking $C$ to be the trivial (one-dimensional) space,
    this gives $Q_n \to S^\perp$, or $\indepq(S) \ge n$.

    Conversely, take $n=\indepq(S)$.  Since $Q_n \to S^\perp$, there are Kraus operators
    $\{E_i\}$ such that $E^\dag S E \subseteq Q_n^\perp = \mathbb{C}I$, as
    per~\eqref{eq:homq_def2}.  At least one of these Kraus operators, call it $E_0$, must
    satisfy $E_0^\dag E_0 \ne 0$.  Since $I \in S$, $E_0^\dag I E_0 \in E^\dag S E
    \subseteq \mathbb{C}I$, so $E_0^\dag E_0 = \alpha I$ with $\alpha \ne 0$.
    Then $J := E_0 / \sqrt{\alpha}$ is an isometry and $P:=J J^\dag$ is a rank $n$
    projector.  Furthermore,
    $PSP = J E_0^\dag S E_0 J^\dag \subseteq \mathbb{C} JIJ^\dag = \mathbb{C} P$.
\end{proof}

\section{Quantum source-channel coding}
\label{sec:qsrcchan}

We construct a quantum version of source-channel coding, as depicted in \cref{fig:fgame}.
The channel $\chan{N}$ from Alice to Bob is now a quantum channel.
Instead of classical inputs $x$ and $u$, Alice and Bob receive a bipartite quantum state.
One may imagine that a referee Charlie chooses a bipartite mixed state
$\rho_i \in \linop{A} \ot \linop{B}$ from some finite collection and sends the $A$
subsystem to Alice and the $B$ subsystem to Bob.
The details of the collection $\{ \rho_i \}$ are known ahead of time to Alice and Bob.
Bob must determine $i$, with zero chance of error, using Alice's message and his share of
$\rho_i$.
We call this \textit{discrete quantum source-channel coding} (discrete QSCC).
Here ``discrete'' refers to $i$; we will later quantize even this.
Discrete QSCC reduces to classical source-channel coding (\cref{sec:classical}) by taking
$\chan{N}$ to be a classical channel and the source to be of the form
$\rho_i = \sum_{xu} P(x,u|i) \proj{x} \ot \proj{u}$.

\begin{figure}
    \centering
    \includegraphics{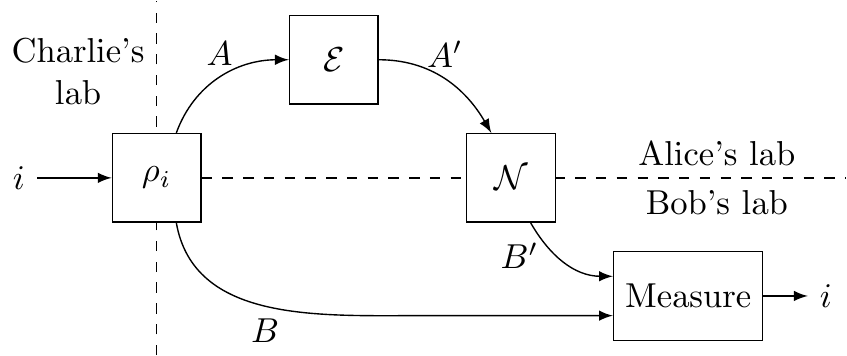}
    \caption{
        Discrete quantum source-channel coding (discrete QSCC).
    }
    \label{fig:fgame}
\end{figure}

The most general strategy is for Alice to encode her portion of $\rho_i$ using some
quantum operation (some CPTP map) $\chan{E} : \linop{A} \to \linop{A'}$ before sending it
through $\cN$ to Bob, and for Bob to perform a POVM measurement on the joint state
consisting of his portion of $\rho_i$ and the message received from Alice.
After receiving Alice's message, Bob is in possession of the mixed state
\begin{align}
    \sigma_i &= \cN( \chan{E}(\rho_i) ) \notag
    \\ &=
    \sum_{jk} (N_k E_j \ot I) \rho_i (E_j^\dag N_k^\dag \ot I) \notag
    \\ &=
    \sum_{jkl} (N_k E_j \ot I) \ket{\psi_{il}}\bra{\psi_{il}} (E_j^\dag N_k^\dag \ot I),
    \label{eq:sigma_NEpsi_sum}
\end{align}
where the unnormalized vectors $\ket{\psi_{il}}$ are defined according to
$\rho_i = \sum_l \ket{\psi_{il}}\bra{\psi_{il}}$.
There is a measurement that can produce the value $i$ with zero error if and only if the
states $\sigma_i$ and $\sigma_{i'}$ are orthogonal whenever $i \ne i'$.
Since each term of~\eqref{eq:sigma_NEpsi_sum} is positive semidefinite we have, with
$\left< \cdot, \cdot \right>$
denoting the Hilbert--Schmidt inner product,
\begin{align*}
    \left< \sigma_i, \sigma_{i'} \right> = 0
    &\iff
    \braopket{\psi_{il}}{(E_j^\dag N_k^\dag \ot I) (N_{k'} E_{j'} \ot I)}{\psi_{i'l'}} = 0
    \\ &\hspace{10em} \forall j,j',k,k',l,l'
    \\ &\iff
    \ip{E_{j'} \Tr_B\{\ket{\psi_{i'l'}}\bra{\psi_{il}}\} E_j^\dag, N_{k'}^\dag N_k} = 0
    \\ &\hspace{10em} \forall j,j',k,k',l,l'
    \\ &\iff E \cdot \Tr_B\{\ket{\psi_{i'l'}}\bra{\psi_{il}}\} \cdot E^\dag \perp N^\dag N
    \\ &\hspace{10em} \forall l,l'.
\end{align*}
By \cref{def:homq}, such an encoding $\chan{E}$ exists if and only if
$\linspan\{\Tr_B\{\ket{\psi_{i'l'}}\bra{\psi_{il}}\} : \forall i \ne i', \forall l,l'\} \to
(N^\dag N)^\perp$.
This immediately leads to the following theorem.
\begin{theorem}
    \label{thm:fgame}

    Consider discrete QSCC (\cref{fig:fgame}) with $i \in \{1,\dots,n\}$.
    For each $i$, let $\ket{\psi_i} \in A\ot B\ot C$ be a purification of
    $\rho_i \in \linop{A} \ot \linop{B}$.
    Define the isometry $J : R \to A\ot B\ot C$ with $R = \mathbb{C}^n$ as
    $J = \sum_{i=1}^n \ket{\psi_i}\bra{i}$.
    There is a zero-error strategy if and only if $S \homm T$ where $T$ is the
    distinguishability graph of $\cN$, given by~\eqref{eq:qm_nonconf}, and $S$
    is the characteristic graph of the source, given by
    \begin{align}
        S = \Tr_{BC}\{ \linop{C} J K_n J^\dag \}.
        \label{eq:thm_fgame}
    \end{align}
\end{theorem}
Suppose that Alice and Bob also share an entanglement resource
$\ket{\lambda} \in A'' \ot B''$.  This can be
absorbed into the source, considering the source to be $\rho_i \ot \proj{\lambda}$.
Then~\eqref{eq:thm_fgame} becomes
$\Tr_{BC}\{ \linop{C} J K_n J^\dag \} \ot \Lambda$
where $\Lambda = \Tr_{B''}\{ \proj{\lambda} \}$.
This motivates the following definition:
\begin{definition}
    \label{def:homq_star}
    Let $S$ and $T$ be trace-free non-commutative graphs.
    We say there is an \emph{entanglement assisted homomorphism} $S \home T$ if there
    exists an operator $\Lambda \succ 0$ such that $S \ot \Lambda \homm T$.
    The entanglement assisted quantities
    $\indepce(S)$, $\indepqe(S)$, $\cliquece(S)$, $\cliqueqe(S)$, $\chromce(S)$, and
    $\chromqe(S)$ are defined by using $\home$ rather than $\homm$ in \cref{def:omega_xi}.
\end{definition}
If $S$ and $T$ are induced by classical graphs $G$ and $H$ then $S \home T$ if and only if
$G \home H$ as defined in~\cite{6994835,6880319}.
This equivalence follows from the fact that $S \home T$ and $G \home H$ have identical
operational interpretation in terms of entanglement assisted source-channel coding.
Our $\indepce(S)$ corresponds to the entanglement assisted independence
number of~\cite{dsw2013} and if $S$ derives from a classical graph our $\chromce(S)$
corresponds to the entangled chromatic number of~\cite{6994835,6880319}.
These quantities, and others, are summarized in \cref{tab:quantities_summary}.

We give some examples.
\begin{itemize}
    \item Dense coding.
        Let $\rho_i = \proj{i}_{A_1} \ot \proj{\lambda}_{A_2 B}$ where $i \in \{1,\dots,m\}$
        represents the codeword to be transmitted and $\proj{\lambda}_{A_2 B}$ is an
        entanglement resource shared by Alice and Bob.
        Take $\cN$ to be a noiseless quantum channel of dimension $n$ (i.e.\ a channel of
        $\log n$ qubits).
        By \cref{thm:fgame}, dense coding is possible if and only if
        $K_m \ot \Tr_B\{\proj{\lambda}\} \homm Q_n$.
        The well known bound $m \le n^2$ for dense coding gives
        $(K_m \home Q_n \iff m \le n^2)$.
        In other words, $\cliquece(Q_n) = n^2$ and $\chromqe(K_{n^2}) = n$.

    \item Entanglement assisted zero-error communication of $n$ different codewords through a
        noisy channel $\chan{N}$ is possible if and only if $K_n \home (N^\dag N)^\perp$.
        So the one-shot entanglement assisted classical capacity is $\log \indepce(N^\dag N)$.

    \item
        Classical or quantum one-way communication complexity of a function.
        Suppose the referee sends Alice a classical message $x$ and sends Bob a classical
        message $y$, with $(x,y) \in R$.
        How large of a message must Alice send to Bob such that Bob may compute some function
        $f(x,y)$?
        The set $R$ and function $f$ are known ahead of time to all parties.

        Take $\rho_i = \sum_{(x,y) \in R \cap f^{-1}(i)} \proj{x} \ot \proj{y}$.
        Let $S = \Tr_{BC}\{ \linop{C} J K_n J^\dag \}$ from~\eqref{eq:thm_fgame}.
        Then $S$ derives (via~\eqref{eq:SfromG}) from the graph $G$ with edges
        $x \sim x' \iff \exists y \textrm{ s.t. } f(x,y) \ne f(x',y)$.
        A classical channel of size $n$ suffices iff $S \homm K_n$, and a quantum channel
        suffices iff $S \homm Q_n$.
        So the smallest sufficient $n$ for a classical channel is $\chromc(S)$
        and for a quantum channel is $\chromq(S)$.
        Since $S$ derives from a classical graph, $\chromc(S)$ and $\chromq(S)$ are just the
        chromatic number and orthogonal rank of $G$.
        This reproduces the result of~\cite{witsen76} and theorem 8.5.2 of~\cite{dewolfphd}.

        If Alice and Bob can share an entangled state the condition becomes
        $S \home K_n$ or $S \home Q_n$ and the smallest $n$ is $\chromce(S)$ or
        $\chromqe(S)$.

    \item One-way communication complexity of nonlocal measurement.
        Alice and Bob each receive half of a bipartite state
        $\ket{\psi_i} \in \linop{A} \ot \linop{B}$ drawn from some finite collection
        agreed to ahead of time.
        What is the smallest message that must be sent from Alice to Bob so that
        Bob can determine $i$?
        Defining $S = \linspan\{\Tr_B\{\ket{\psi_i}\bra{\psi_j}\} : i \ne j \}$,
        a quantum message of dimension $n$ suffices if and only if
        $S \homm Q_n$.
        So the message from Alice to Bob must be at least $\log \chromq(S)$ qubits or
        $\log \chromc(S)$ bits.
        If the states $\{ \ket{\psi_i} \}$ are not distinguishable via one-way local
        operations and classical communication (LOCC-1) then $\chromc(S)=\infty$.
\end{itemize}

We further generalize by replacing the index $i$ with a quantum state.
Instead of the referee sending $\rho_i$, we imagine an isometry
$J : R \to A\ot B\ot C$ into which the referee passes a quantum state
$\ket{\psi} \in R$.
Alice receives subsystem $A$, Bob receives $B$, and $C$ is dumped to the environment.
One may think of $J$ as the Stinespring isometry for a channel
$\mathcal{J} : \linop{R} \to \linop{A \ot B}$.
We call this \textit{coherent QSCC}; the setup is depicted in \cref{fig:fgame_coh}.
The goal is for Bob to reproduce the state $\ket{\psi}$, with perfect fidelity.
Discrete QSCC is recovered by taking
$J = \sum_i \ket{\psi_i}\bra{i}$ where $\ket{\psi_i} \in A\ot B\ot C$ is
a purification of $\rho_i \in \linop{A} \ot \linop{B}$, and requiring that the input
state be a basis state.

\begin{figure}
    \centering
    \includegraphics{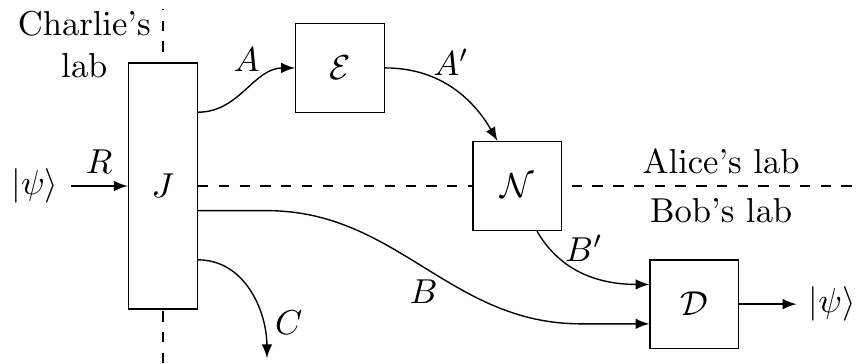}
    \caption{
        Coherent quantum source-channel coding (coherent QSCC).
    }
    \label{fig:fgame_coh}
\end{figure}

After Alice's transmission, Bob is in possession of the state
$\mathcal{N}(\mathcal{E}(\Tr_C(J\proj{\psi}J^\dag)))$.
In order to recover $\ket{\psi}$, Bob must perform some operation that converts the
channel $\rho \to \mathcal{N}(\mathcal{E}(\Tr_C(J \rho J^\dag)))$ into the identity
channel.
The Kraus operators of this channel are
$\{(N_k E_j \ot I_B \ot \bra{l}_C) J\}_{jkl} \subseteq \linop{R \to B \ot B'}$.
By the Knill--Laflamme error correction condition~\cite{PhysRevA.55.900},
recovery of $\ket{\psi}$ is possible if and only if, $\forall j,j',k,k',l,l'$,
\begin{align}
    \label{eq:knill_JEN}
    J^\dag \left( E_j^\dag N_k^\dag N_{k'} E_{j'} \ot I_B \ot \ket{l}\bra{l'}_C \right) J
    \in \mathbb{C} I.
\end{align}
An operator is proportional to $I$ if and only if it is orthogonal to all trace free operators,
so this becomes
\begin{align*}
    &\phantom{\iff}
    \left< J^\dag \left( E_j^\dag N_k^\dag N_{k'} E_{j'} \ot I_B
        \ot \ket{l}\bra{l'}_C \right) J, X \right> = 0
        \\ &\hspace{10em} \forall j,j',k,k',l,l', \forall X \in Q_n
    \\ &\iff
    \left< E_{j'} \Tr_{B}\{ \bra{l'}_C J X J^\dag \ket{l}_C \} E_j^\dag,
        N_{k'}^\dag N_k \right> = 0
        \\ &\hspace{10em} \forall j,j',k,k',l,l', \forall X \in Q_n
    \\ &\iff
    E \cdot \Tr_{BC}\{ \linop{C} J Q_n J^\dag \} \cdot E^\dag \subseteq (N^\dag N)^\perp.
\end{align*}
Or, using the terminology of homomorphisms,
\begin{theorem}
    \label{thm:fgame_coh}
    There is a zero-error strategy for coherent QSCC (\cref{fig:fgame_coh}) if and only if
    $S \homm T$ where $T$ is the
    distinguishability graph of $\cN$, given by~\eqref{eq:qm_nonconf}, and $S$
    is the characteristic graph of the source, given by
    \begin{align}
        S = \Tr_{BC}\{ \linop{C} J Q_n J^\dag \}
    \end{align}
    where $n = \dim(R)$.
\end{theorem}
This differs from \cref{thm:fgame} only in the replacement of $K_n$ by $Q_n$.
As before, if Alice and Bob are allowed to make use of an entanglement resource
the condition becomes $S \home T$ rather than $S \to T$.

We give some examples.
\begin{itemize}
    \item Teleportation.  Take $J = I_{A_1} \ot \ket{\lambda}_{A_2 B}$ where $I$ is the
        identity operator (i.e.\ the referee directly gives $\ket{\psi}$ to Alice) and
        $\ket{\lambda}$ is an entanglement resource.
        Take $\cN$ to be a perfect classical channel.
        By \cref{thm:fgame_coh} teleportation is possible if and only if
        $Q_m \ot \Tr_B\{\proj{\lambda}\} \homm K_n$ where $m$ is the dimension of the state to
        be teleported and $n$ is the dimension of the classical channel.
        The well known bound $m^2 \le n$ for teleportation gives
        $(Q_m \home K_n \iff m^2 \le n)$.
        In other words, $\cliqueqe(K_{m^2}) = m$ and $\chromce(Q_m) = m^2$.

    \item Zero-error one-shot quantum communication capacity.  Take $\cN$ to be a noisy
        channel, and take $J : R \to A$ to be the identity operator (i.e.\ the referee
        gives $\ket{\psi}$ directly to Alice, and Bob gets no input).
        It is possible to send $\log m$ error-free qubits though $\cN$ if and only if
        $Q_m \homm (N^\dag N)^\perp$.
        By definition, $m \le \indepq(N^\dag N)$.
        If Alice and Bob can use an entangled state, these conditions become
        $Q_m \home (N^\dag N)^\perp$ and $m \le \indepqe(N^\dag N)$.

    \item Suppose Alice and Bob each have a share of a quantum state that has been
        cloned in the standard basis.  That is to say, suppose
        $J = \sum_{x=1}^n \ket{xx}_{AB} \bra{x}_R$.
        Can Alice send a classical message to Bob such that Bob may reconstruct the
        original quantum state?
        The characteristic graph of this source (call it $S$) is the space of trace-free diagonal
        matrices.
        Conjugating by the Fourier matrix yields a subspace of $K_n$.
        So the Fourier transform is a homomorphism $S \homm K_n$; indeed a classical
        message does suffice.

    \item Imagine that Alice tries to send a quantum message to Bob, but part of the
        signal bounces back.
        This can be modeled by a channel $\chan{J} : \linop{R} \to \linop{A} \ot \linop{B}$.
        Alice must now send a second message though a second channel $\cN$ in order to
        allow Bob to reconstruct the original message.
        This is exactly the setup depicted in \cref{fig:fgame_coh},
        with Charlie being Alice and $J$ being the Stinespring isometry of $\chan{J}$.

    \item Correction of algebras.
        Suppose instead of transmitting $\ket{\psi}$ perfectly, one needs only
        that some $C^*$-algebra of observables $\mathcal{A}$ be preserved (i.e.\ the
        receiver can do any POVM measurement with elements from $\mathcal{A}$).
        This reduces to discrete QSCC when $\mathcal{A}$ consists of the diagonal operators.
        By theorem 2 of~\cite{PhysRevLett.98.100502}, this problem is analyzed via a
        straightforward modification of the Knill--Laflamme condition: $\mathbb{C}I$
        in~\eqref{eq:knill_JEN} should be replaced by the space of operators that commute
        with everything in $\mathcal{A}$ (the \textit{commutant} of $\mathcal{A}$);
        \cref{thm:fgame_coh} is modified by replacing $Q_n$ with the space
        perpendicular to the commutant of $\mathcal{A}$.
        \Cref{thm:fgame} is recovered by taking $\mathcal{A}$ to consist of the diagonal
        operators.

    \item Consider discrete QSCC with the inputs $\rho_1,\rho_2,\rho_3,\rho_4$ being the
        four Bell states (or even three of the four).
        The characteristic graph is $Q_2$.
        This is the same as the graph for coherent QSCC
        with the goal being for Alice to transmit an arbitrary qubit to Bob
        ($J : R \to A$ is the identity operator).
        Since the characteristic graphs are the same for the two problems, they
        require the same communication resources.
\end{itemize}

\begin{table}
    \centering
    \setlength{\tabcolsep}{2mm}
    \begin{tabular}{p{3.6cm}p{4.5cm}}
        Quantity & Interpretation
        \\ \hline
        $K_n = \{ M \in \linop{\mathbb{C}^n} : M_{ii}=0 \}$ &
        Classical complete graph.
        The set of $n \times n$ matrices with zeros down the diagonal.
        \\[1mm]
        $Q_n = (\mathbb{C} I_n)^\perp$ &
        Quantum complete graph.
        The set of trace-free $n \times n$ matrices.
        \\[1mm]
        $N = \linspan\{ N_i \}$ &
        Span of Kraus operators for channel $\cN$.
        \\[1mm]
        $N^\dag N = \linspan\{ N_i^\dag N_j \}$ &
        Confusability graph of channel $\cN$.
        \\[1mm]
        $(N^\dag N)^\perp$ &
        Distinguishability graph of channel $\cN$.
        \\[1mm]
        $S \homm T \iff ESE^\dag \subseteq T$ with $E$ span of Kraus operators &
        Graph homomorphism.  Source with characteristic graph $S$ can be transmitted
        using channel with distinguishability graph $T$.
        \\[1mm]
        \pbox{4.5cm}{
        $S \home T \iff (\exists \Lambda \succ 0$ \\ s.t.\ $S \ot \Lambda \homm T)$
        } &
        Entanglement assisted homomorphism.
        As before, but sender and receiver share an entanglement resource.
        \\[1mm]
        $\cliquec(S) = \max\{ n : K_n \homm S \}$ &
        Clique number.  One-shot classical capacity of channel with distinguishability
        graph $S$ is $\log \cliquec(S)$.
        \\[1mm]
        $\cliqueq(S) = \max\{ n : Q_n \homm S \}$ &
        Quantum clique number.  One-shot quantum capacity of channel with
        distinguishability
        graph $S$ is $\log \cliquec(S)$.
        \\[1mm]
        $\indepc(S) = \cliquec(S^\perp)$ &
        Independence number.  One-shot classical capacity of channel with confusability
        graph $S$ is $\log \indepc(S)$.
        \\[1mm]
        $\indepq(S) = \cliqueq(S^\perp)$ &
        Quantum independence number.  One-shot quantum capacity of channel with
        confusability graph $S$ is $\log \indepc(S)$.
        \\[1mm]
        $\chromc(S) = \min\{ n : S \homm K_n \}$ &
        Chromatic number.  Source with characteristic graph $S$ can be transmitted using
        $\log \chromc(S)$ classical bits.
        \\[1mm]
        $\chromq(S) = \min\{ n : S \homm Q_n \}$ &
        Quantum chromatic number.  Source with characteristic graph $S$ can be transmitted using
        $\log \chromq(S)$ qubits.
        For classical graphs this equals the orthogonal rank.
        \\[1mm]
        $\cliquece$,
        $\cliqueqe$,
        $\indepce$,
        $\indepqe$,
        $\chromce$,
        $\chromqe$ &
        Entanglement assisted quantities.  Replace $\homm$ with $\home$ in above
        definitions.  Relevant when sender and receiver share an entanglement resource.
    \end{tabular}
    \caption{
        Basic definitions used in this paper, and their interpretations.
        See \cref{def:homq} for the full definition of $S \homm T$.
        See \cref{thm:fgame,thm:fgame_coh} for the definition of characteristic graph.
    }
    \label{tab:quantities_summary}
\end{table}

Lemma 2 of~\cite{arxiv:0906.2527} states that every non-commutative graph
containing the identity is the confusability graph of some channel (equivalently, every
trace-free non-commutative graph is the distinguishability graph of some channel).
A similar statement holds for sources.

\begin{theorem}
    \label{thm:any_is_char}
    Every non-commutative graph $S$ is the characteristic graph for discrete QSCC
    with only two inputs (i.e.\ $\rho_0$ and $\rho_1$).
\end{theorem}
\begin{proof}
    Let $S \in \linop{A}$ be a non-commutative graph and let $\{ S_x \}_{x \in X}$ be a
    basis of $S$, with each $S_x$ being Hermitian.
    That such a Hermitian basis always exists is shown in~\cite{arxiv:0906.2527}.
    Without loss of generality, assume that each $S_x$ is normalized under the Frobenius
    norm.
    Let $(S_x)_{ij}$ be the entries of matrix $S_x$ and define
    $\ket{S_x} = \sum_{ij} (S_x)_{ij} \ket{i}_A \ket{j}_B$.
    Also define $\ket{\Phi} = \dim(A)^{-1/2} \sum_i \ket{i}_A \ot \ket{i}_B$.
    Consider discrete QSCC with sources
    $\rho_i = \Tr_C\{\proj{\psi_i}\}$ for $i \in \{0,1\}$ with
    $\ket{\psi_i} \in A \ot B \ot B' \ot C$ defined by
    \begin{align*}
        \ket{\psi_0} &= \abs{X}^{-1/2} \sum_x \ket{\Phi} \ot \ket{x}_{B'} \ot \ket{x}_C
        \\
        \ket{\psi_1} &= \abs{X}^{-1/2} \sum_x \ket{S_x} \ot \ket{x}_{B'} \ot \ket{x}_C
    \end{align*}
    Alice receives subsystem $A$ and Bob receives subsystems $B \ot B'$.
    Subsystem $C$ goes to the environment.
    As per \cref{thm:fgame}, the characteristic graph is
    \begin{align*}
        S &= \Tr_{BB'C}\{ \linop{C} \ket{\psi_1} \bra{\psi_0} \} + \textrm{h.c.}
        \\ &= \linspan_x\{ \Tr_B( \ket{S_x}\bra{\Phi} ) \} + \textrm{h.c.}
        \\ &= \linspan_x\{ S_x \} + \textrm{h.c.} = S
    \end{align*}
    where ``$+\textrm{h.c.}$'' means that the adjoints of the operators are also included
    in the subspace.
\end{proof}

Note that we didn't require $S$ to be trace-free in \cref{thm:any_is_char}; however, if
$S$ is not trace-free then source-channel coding will be impossible:
$\rho_0$ and $\rho_1$ would be non-orthogonal and so would not be
distinguishable by any measurement.

\section{\texorpdfstring{$\qthperp$}{Lovasz theta} is a homomorphism monotone}
\label{sec:qthmon}

We will show that $\qthperp$ is monotone under entanglement assisted
homomorphisms of non-commutative graphs.
This leads to a Lov{\'a}sz sandwich theorem for non-commutative graphs, and a bound on
quantum source-channel coding.
We begin by showing $\qthperp$ to be insensitive to entanglement.
Recall that a source having non-commutative graph $S$, combined with an entanglement
resource $\ket{\lambda} \in A'' \ot B''$, yields a composite source with non-commutative graph
$S \ot \Lambda$ where $\Lambda = \Tr_{B''}\{ \proj{\lambda} \}$.

\begin{lemma}
    \label{thm:sperp_Lambda}
    Let $S$ be a trace-free non-commutative graph.  Let $\Lambda$ be a positive operator.
    Then $\qthperp(S) = \qthperp(S \ot \Lambda)$.
\end{lemma}
\begin{proof}
    Suppose $S \subseteq \linop{A}$ and $\Lambda \in \linop{B}$.
    By~\eqref{eq:qth_primal} we have
    \begin{align}
        \label{eq:qth_norho}
        \qthperp(S) &= \max \{ \opnorm{I+T} :
            T \in S \ot \linop{\mathbb{C}^m}, \notag
        \\ &\qquad I+T \succeq 0 \}
        \\
        \label{eq:qth_withrho}
        \qthperp(S \ot \Lambda) &= \max \{ \opnorm{I+T} :
            T \in S \ot \Lambda \ot \linop{\mathbb{C}^n}, \notag
        \\ &\qquad I+T \succeq 0 \}.
    \end{align}
    In~\cite{dsw2013} is it shown that the ancillary space can be
    taken to be any dimension at least as large as $\dim(A)$,
    so in~\eqref{eq:qth_norho}-\eqref{eq:qth_withrho} we may take any values
    $m \ge \dim(A)$ and $n \ge \dim(A \ot B)$.

    Take $n = \dim(A \ot B)$ and $m = n \dim(B)$.
    Any $T$ feasible for~\eqref{eq:qth_withrho} is also feasible
    for~\eqref{eq:qth_norho} since $\Lambda \ot \linop{\mathbb{C}^n} \subseteq
        \linop{\mathbb{C}^m}$.
    So $\qthperp(S) \ge \qthperp(S \ot \Lambda)$.

    Now take $m = n = \dim(A \ot B)$.
    Let $T$ be feasible for~\eqref{eq:qth_norho}.
    Without loss of generality, assume $\opnorm{\Lambda}=1$.
    Then $T' := T \ot \Lambda$ is feasible for~\eqref{eq:qth_withrho}.
    Indeed, $T \succeq -I \implies T' \succeq -I$ since $\Lambda \succ 0$ and
    $\opnorm{\Lambda} \le 1$.
    Also, $\opnorm{I+T'} \ge \opnorm{I+T}$ since $\Lambda \succ 0$ and $\opnorm{\Lambda} \ge 1$.
    So $\qthperp(S \ot \Lambda) \ge \qthperp(S)$.
\end{proof}

Before we prove the main theorem, we introduce some notation that will also be used in
\cref{sec:schrijver}.
For any (finite dimensional) Hilbert space $A$, define the state
\begin{align}
    \ket{\Phi}_A = \sum_i \ket{i}_A \ot \ket{i}_{A'}.
    \label{eq:transposer}
\end{align}
where $A'$ is another Hilbert space of the same dimension as $A$.
Note that this provides an isomorphism between $A$ and the dual space of $A'$
via the action $\ket{\psi}_A \to \bra{\Phi} (\ket{\psi}_A \ot I_{A'})$.
A bar over an operator denotes entrywise complex conjugation in the standard basis
(i.e.\ the basis used in~\eqref{eq:transposer}).  Additionally, the bar will be understood
to move an operator to the primed spaces (or from primed to unprimed).
For example, if $J : A \to B \ot C$ then $\conj{J} : A' \to B' \ot C'$ is equal to
\begin{align}
    \label{eq:Jconj}
    \conj{J} = \Tr_{BC}\{ \ket{\Phi}_{BC} \bra{\Phi}_A J^\dag \}.
\end{align}
We now prove the main theorem of this section: that $\qthperp$ is monotone under
entanglement assisted homomorphisms.
Such an inequality was already known for classical graphs~\cite{6880319}.

\begin{theorem}
    \label{thm:qthmon}
    Let $S$ and $T$ be trace-free non-commutative graphs.
    If $S \homm T$ or $S \home T$ then $\qthperp(S) \le \qthperp(T)$.
\end{theorem}
\begin{proof}
    If we prove monotonicity under $S \homm T$ then monotonicity under $S \home T$ follows.
    Indeed, if $S \home T$ then there is a $\Lambda \succ 0$ such that $S \ot \Lambda \to T$.
    Supposing that $\qthperp$ is monotone under (non-entanglement assisted) homomorphisms,
    we have $\qthperp(S \ot \Lambda) \le \qthperp(T)$.
    \Cref{thm:sperp_Lambda} then gives $\qthperp(S) \le \qthperp(T)$.

    We now show that $S \homm T$ implies $\qthperp(S) \le \qthperp(T)$.
    This can be seen as a consequence of corollaries from~\cite{dsw2013}.
    Let $S \subseteq \linop{A}$ and $T \subseteq \linop{B}$ be trace-free non-commutative
    graphs with $S \homm T$.
    By \cref{def:homq} there is a Hilbert space $C$ and an isometry $J : A \to B \ot C$
    such that $J^\dag (T^\perp \ot \linop{C}) J \subseteq S^\perp$.
    Then
    \begin{align*}
        \qth(T^\perp) &= \qth(T^\perp) \qth(\linop{C})
            && \mbox{(Since $\qth(\linop{C})=1$)}
        \\ &= \qth(T^\perp \ot \linop{C})
            && \mbox{(Corollary 10 of~\cite{dsw2013})}
        \\ &\ge \qth(J^\dag (T^\perp \ot \linop{C}) J)
            && \mbox{(Corollary 11 of~\cite{dsw2013})}
        \\ &\ge \qth(S^\perp).
            && \mbox{(Corollary 11 of~\cite{dsw2013})}
    \end{align*}

    We present also a more direct proof, since this can later be generalized for the
    $\qthmperp{\cone}$ and $\qthpperp{\cone}$ quantities of \cref{sec:schrijver}.
    Let $S \subseteq \linop{A}$ and $T \subseteq \linop{B}$ be trace-free
    non-commutative graphs with $S \homm T$.
    By definition there is a CPTP map $\chan{E} : \linop{A} \to \linop{B}$ with Kraus
    operators $\{E_i\}$ such that
    $E^\dag T^\perp E \subseteq S^\perp$ where $E = \linspan\{ E_i \}$.
    Let $J : A \to B \ot C$ be the Stinespring isometry for the channel $\chan{E}$, so that
    $J = \sum_i \ket{i}_C E_i$.
    Let $\conj{J} : A' \to B' \ot C'$ be the entrywise complex conjugate of $J$.
    Recall that $\conj{J}$ takes the form~\eqref{eq:Jconj}.

    Let $Y' \subseteq \linop{B} \ot \linop{B'}$ be an optimal solution for~\eqref{eq:qth_dual}
    for $\qthperp(T)$.
    Define $Y \subseteq \linop{A} \ot \linop{A'}$ as
    \begin{align}
        Y &= \sum_{ij} (E_i \ot \conj{E}_i)^\dag Y' (E_j \ot \conj{E}_j) \notag
        \\ &= (J \ot \conj{J})^\dag (\ket{\Phi}_C \ot Y' \ot \bra{\Phi}_C) (J \ot \conj{J}).
        \label{eq:Y_from_Yprime}
    \end{align}
    We have that
    \begin{align}
        Y' \in T^\perp \ot \linop{B'}
        &\implies Y \in E^\dag T^\perp E \ot \conj{E}^\dag \linop{B'} \conj{E} \notag
        \\ &\implies Y \in S^\perp \ot \linop{A'}.
        \label{eq:JJYJJ_in_Sperp}
    \end{align}

    Dirac notation becomes unworkable with so many Hilbert spaces, so we will make use of
    diagrams similar to those of~\cite{PhysRevA.73.052309} (such diagrams have been used
    for example in~\cite{PhysRevA.84.032316,PhysRevA.84.022333}).
    Operators and states are denoted by labeled boxes and multiplication (or traces or tensor
    contraction) by wires.  This is very much like standard quantum circuits, except that the
    diagram has no interpretation of time ordering.
    Wires have an arrow pointing away from the ket space and towards the bra space, and are
    labeled with the corresponding Hilbert space.
    Labeled circles represent the transposer operators~\eqref{eq:transposer}:
    \begin{center}
        \includegraphics{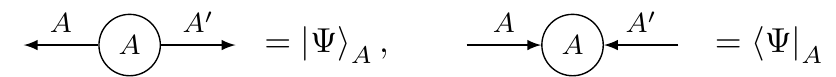}
    \end{center}
    With this notation,~\eqref{eq:Y_from_Yprime} becomes
    \begin{equation}
        \label{eq:Y_def_diagram}
        \begin{minipage}[c]{0.85\linewidth}
            \centering
            \includegraphics{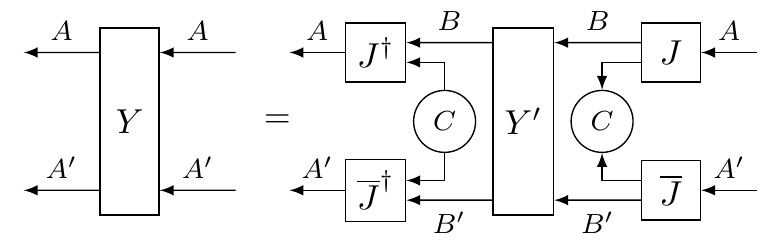}
        \end{minipage}
    \end{equation}
    The operator
    $(J \ot \conj{J})^\dag (\ket{\Phi}_C \ot I_{BB'})$
    turns the transposer $\ket{\Phi}_B$ into $\ket{\Phi}_A$:
    \begin{align}
        (J \ot \conj{J})^\dag (\ket{\Phi}_C \ot \ket{\Phi}_B)
        &= (J^\dag \ot \conj{J}^\dag) (\ket{\Phi}_C \ot \ket{\Phi}_B) \notag
        \\ &= (J^\dag J \ot I) \ket{\Phi}_A = \ket{\Phi}_A.
        \label{eq:JJcancel}
    \end{align}
    With diagrams, the same derivation is written as in \cref{fig:JJcancel}.
    \begin{figure*}
        \centering
        \includegraphics{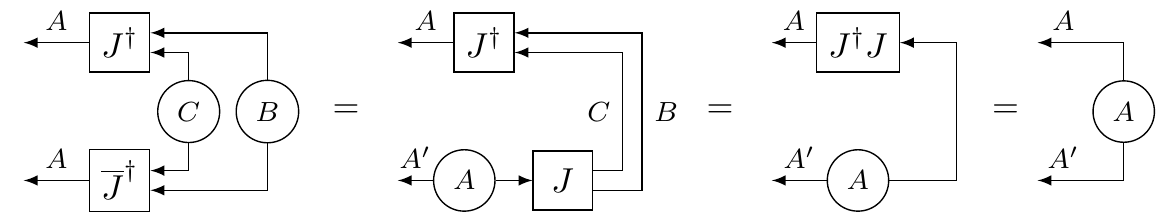}
        \caption{Graphical representation of \eqref{eq:JJcancel}.}
        \label{fig:JJcancel}
    \end{figure*}
    Consequently we have $Y' \succeq \proj{\Phi}_B \implies Y \succeq \proj{\Phi}_A$ so $Y$ is
    feasible for~\eqref{eq:qth_dual} for $\qthperp(S)$.
    To get $\qthperp(S) \le \qthperp(T)$ it remains only to show
    $\opnorm{\Tr_A Y} \le \opnorm{\Tr_B Y'}$.
    Let $\rho \in \linop{A'}$ be a density operator achieving
    $\Tr\{ (I_A \ot \rho) Y \} = \opnorm{\Tr_A Y}$.
    Plugging~\eqref{eq:Y_def_diagram} (the definition of $Y$) into $\Tr\{ (I_A \ot \rho) Y \}$ gives
    the diagram of \cref{fig:Yval}.
    \begin{figure*}
        \centering
        \includegraphics{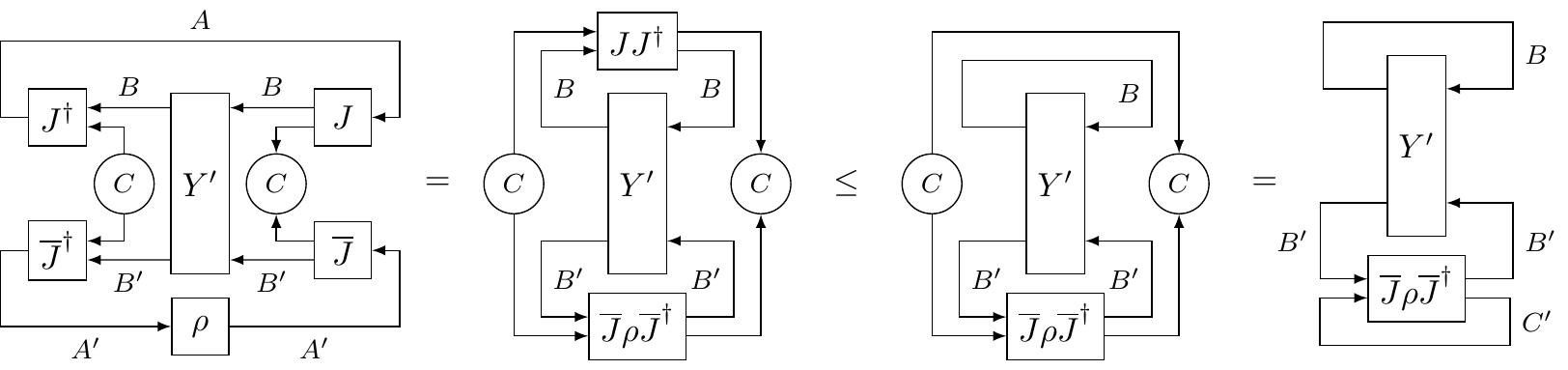}
        \caption{Graphical representation of \eqref{eq:mon_norm_bound}.}
        \label{fig:Yval}
    \end{figure*}
    The first equality involves only a rearrangement of the diagram.
    The inequality follows from the fact that $JJ^\dag \preceq I_{BC}$ (since $JJ^\dag$ is a
    projector) and the rest of the diagram represents a positive semidefinite
    operator.\footnote{
        If a portion of a diagram has reflection symmetry, with the operators located on the
        reflection axis being positive semidefinite, then that portion of the diagram is
        positive semidefinite~\cite{PhysRevA.73.052309}.
        This follows from the fact that $M \succeq 0 \implies N^\dag M N \succeq 0$.
    }
    The last equality uses $\Tr_C\{ \proj{\Phi}_C \} = I_{C'}$.
    The same derivation can be written in equations as
    \begin{align}
        &\Tr_{AA'} \{ (I_A \ot \rho) Y \} \notag
        \\ &= \Tr_{AA'} \{ (I_A \ot \rho) (J \ot \conj{J})^\dag
            (\ket{\Phi}_C \ot Y' \ot \bra{\Phi}_C) (J \ot \conj{J}) \} \notag
        \\ &= \Tr_{BB'CC'} \{ (J J^\dag \ot \conj{J} \rho \conj{J}^\dag)
            (\ket{\Phi}_C \ot Y' \ot \bra{\Phi}_C) \} \notag
        \\ &\le \Tr_{BB'CC'} \{ (I_{BC} \ot \conj{J} \rho \conj{J}^\dag)
            (\ket{\Phi}_C \ot Y' \ot \bra{\Phi}_C) \} \notag
        \\ &= \Tr_{BB'C'} \{ (I_B \ot \conj{J} \rho \conj{J}^\dag) (Y' \ot I_{C'}) \}.
        \label{eq:mon_norm_bound}
    \end{align}
    Since $\conj{J}$ is an isometry, $\Tr_{C'}\{ \conj{J} \rho \conj{J}^\dag \}$ is a density
    operator.  So~\eqref{eq:mon_norm_bound} is bounded from above by $\opnorm{\Tr_B Y'}$ and we
    have $\opnorm{\Tr_A Y} \le \opnorm{\Tr_B Y'}$ as desired.
\end{proof}

An immediate corollary is that $\qthperp$ satisfies a sandwich theorem analogous
to~\eqref{eq:sandwich}.

\begin{corollary}
    \label{thm:sandwich_nc}
    Let $S$ be a trace-free non-commutative graph.  Then
    \begin{enumerate}
        \item $\cliquece(S) \le       \qthperp(S)  \le \chromce(S)$.
        \item $\cliqueqe(S) \le \sqrt{\qthperp(S)} \le \chromqe(S)$.
    \end{enumerate}
\end{corollary}
\begin{proof}
    This follows from applying \cref{thm:qthmon} to \cref{def:homq_star} and using the
    fact that $\qthperp(K_n) = n$ and $\qthperp(Q_n) = n^2$.
    For example, $\cliquece(S)$ is the greatest $n$ such that $K_n \home S$.  By
    \cref{thm:qthmon} we have $\qthperp(K_n) \le \qthperp(S)$.  Since
    $\qthperp(K_n)=n$ this gives $n \le \qthperp(S)$.

    That $\qthperp(K_n)=n$ and $\qthperp(Q_n)=n^2$ is not hard to work
    out, and is proved in~\cite{dsw2013}.
\end{proof}

The bound on $\cliquece(S)$ and $\cliqueqe(S)$ was shown already in~\cite{dsw2013},
although it was stated in terms of $\indepce$ and $\indepqe$.
The bound on $\chromce(S)$ and $\chromqe(S)$ is new, although the bound on
$\chromce(S)$ was known already for classical graphs~\cite{6994835}.
Note that, since for example $\cliquec(S) \le \cliquece(S)$, we also have the weaker
sandwich theorem
$\cliquec(S) \le       \qthperp(S)  \le \chromc(S)$ and
$\cliqueq(S) \le \sqrt{\qthperp(S)} \le \chromq(S)$.

This bound $\sqrt{\qthperp(S)} \le \chromq(S)$ is not particularly tight when $S$
corresponds to a classical graph $G$, for the following reason.
In such cases $\chromq(S) = \orthrank(G)$, the orthogonal rank.
But it is known that $\qthperp(S) = \thbar(G) \le \orthrank(G)$~\cite{lovasz79},
so in this case the square root over $\qthperp$ is unnecessary.
The necessity of the square root arises from the possibility of dense
coding, since we are bounding the entanglement assisted quantities in
\cref{thm:sandwich_nc}.
Notice that $\cliqueqe(S) = \left\lfloor \sqrt{\cliquece(S)} \right\rfloor$ and
$\chromqe(S) = \left\lceil \sqrt{\chromce(S)} \right\rceil$ since
a quantum channel of dimension $n$ can simulate a classical channel of dimension $n^2$,
and teleportation can do the reverse.

The square root in \cref{thm:sandwich_nc} could be eliminated by defining a
different generalization of Lov{\'a}sz's $\thbar$ that is monotone under homomorphisms
and which takes the value $n$ on the graph $Q_n$.
Such a quantity would necessarily not be monotone under entanglement assisted
homomorphisms, since $Q_2 \home K_4$.
Finding such a quantity is left as an open question.

\Cref{thm:qthmon} can be applied to give bounds for all of the examples in \cref{sec:qsrcchan}.
Two are especially noteworthy: \cref{thm:qthmon} gives the well known bounds
$n \le m^2$ for dense coding and $n^2 \le m$ for teleportation (where $n$ is the dimension
of the source and $m$ is the dimension of the channel).

\section{Graph products and block coding}
\label{sec:repetitions}

Consider the problem of sending several parallel sources using several parallel channels.
In general these several sources (as well as the channels) could all be distinct, and we
will in fact consider this.
In the special case where the sources are identical, as well as the channels, one may ask
how many channel uses are required for each instance of the source.
This is known as the cost rate.
For classical sources and channels, we saw already (\cref{thm:costrate_classical}) that a
bound on cost rate is given in terms of the Lov{\'a}sz $\thbar$ number.
The goal of this section is to prove an analogous bound in the case of quantum sources and
channels.
To build this theory, we begin with an investigation of the classical case.

Consider two channels $\cN(v|s)$ and $\cN'(v'|s')$ having distinguishability graphs $H$ and
$H'$.  It is not hard to see that the composite channel
$\cN''(v,v'|s,s') = \cN(v|s) \cN'(v'|s')$ has a distinguishability graph with
vertices $V(H) \times V(H')$ and edges
\begin{align}
    \label{eq:disjunctive_classical}
    (x,x') \sim (y,y') \iff (x \sim y) \textrm{ or } (x' \sim y').
\end{align}
This is known as the \emph{disjunctive product}, denoted $H \djp H'$.
If $n$ identical copies of $\cN$ are used in parallel, the resulting composite
channel will have distinguishability graph $H^{\djp n} = H \djp H \djp \dots \djp H$.
Since the one-shot capacity of a channel is $\log \omega(H)$ bits, the capacity
(per-channel use) of $n$ parallel channels is $\frac{1}{n} \log \omega(H^{\djp n})$.
The capacity in the limit $n \to \infty$ is known as the
\emph{Shannon capacity} of the channel,
\begin{align}
    \label{eq:shannon_H}
    C_0(\Hc) = \lim_{n \to \infty} \frac{1}{n} \log
        \cliquec(H^{\djp n}).
\end{align}
Since $\cliquec(H^{\djp n})$ is super-multiplicative, by Fekete's lemma this limit exists
and is equal to the supremum.
The complement in the argument of $C_0(\Hc)$ is because we consider the
distinguishability graph rather than the confusability graph, in terms of
which $C_0$ is typically defined.
Since $\thbar(H^{\djp n}) = \thbar(H)^n$,
it holds that $C_0(\Hc) \le \log \thbar(H)$~\cite{lovasz79}.
In fact this was the original motivation for defining the $\vartheta$ number.

Now consider parallel sources.
Recall from \cref{sec:classical} that the sources $P(x,u|i)$ we consider are somewhat
generalized from what is traditionally considered.
The traditional definition is obtained by requiring $P(x,u|i) \ne 0$ only when $x=i$.
In this case, the characteristic graphs of parallel sources combine by the \emph{strong
product}~\cite{witsen76} which has vertices $V(G) \times V(G')$ and edges
\begin{align}
    (x,x') \sim (y,y')
    \iff &(x \sim y \textrm{ and } x' \sim y') \textrm{ or } \notag
    \\ &(x = y \textrm{ and } x' \sim y') \textrm{ or } \notag
    \\ &(x \sim y \textrm{ and } x' = y').  \label{eq:strong_prod}
\end{align}
Adapting this to non-commutative graphs is problematic because there is no clear analogue
of the condition $x=y$.
But already for our generalized sources, which can have $P(x,u|i) \ne 0$ when
$x \ne i$, the product rule needs modification.

Consider two parallel sources $P(x,u|i)$ and $P(x',u'|i')$ (these can
be over different alphabets) with characteristic graphs $G$ and $G'$.
Call the combined source $P''(x,x',u,u'|i,i') := P(x,u|i) P(x',u'|i')$.
This has characteristic graph $G''$ with vertex set
$V(G) \times V(G')$ and edges given by a generalization of~\eqref{eq:strong_prod}.
To this end, we introduce a graph $G_0$ having the same vertices as $G$ but with edges
\begin{align}
    \label{eq:G0}
    x \sim_{G_0} y \iff \exists u, \exists i &\textrm{ s.t.\ } P(x,u|i) P(y,u|i) \ne 0.
\end{align}
$G'_0$ is defined similarly.
If $P(x,u|i) \ne 0$ only when $x=i$ then $G_0$
has edges $x \sim_{G_0} y \iff x=y$.  In other words $G_0$
consists only of loops.  So~\eqref{eq:G0} can be regarded as a set of generalized loops.
We will call the pair $(G, G_0)$ a \emph{graph with generalized loops}.
We can now compute $G''$, the characteristic graph for the composite source:
\begin{align}
    \notag (x,x') \sim_{G''} &(y,y')
    \\ \notag \iff &\exists u, u', \exists (i,i') \ne (j,j') \textrm{ s.t.\ }
    \\ \notag &P''(x,x',u,u'|i,i') P''(y,y',u,u'|j,j') \ne 0
    \\ \notag \iff &(x \sim_G y \textrm{ and } x' \sim_{G'} y' ) \textrm{ or }
    \\ \notag &(x \sim_{G_0} y \textrm{ and } x' \sim_{G'} y' ) \textrm{ or }
    \\ &(x \sim_G y \textrm{ and } x' \sim_{G'_0} y' ).
     \label{eq:Gpp}
\end{align}
And the graph $G''_0$, defined analogously to~\eqref{eq:G0}, has edges
\begin{align}
    \notag (x,x') \sim_{G''_0} &(y,y')
    \\ \notag \iff &\exists u, u', \exists i,i' \textrm{ s.t.\ }
    \\ \notag &P''(x,x',u,u'|i,i') P''(y,y',u,u'|i,i') \ne 0
    \\ \iff &x \sim_{G_0} y \textrm{ and } x' \sim_{G'_0} y'.
    \label{eq:Gpp0}
\end{align}
We introduce the notation $(G'', G''_0) = (G, G_0) \boxtimes (G', G'_0)$ as shorthand
for~\eqref{eq:Gpp}-\eqref{eq:Gpp0}.
By induction, $m$ parallel instances of a source yields a characteristic graph
$(G, G_0)^{\boxtimes m} := (G, G_0) \boxtimes (G, G_0) \boxtimes \dots \boxtimes (G, G_0)$.

For convenience we will abuse notation by treating these ordered pairs as being graphs
themselves.  For instance, $(G, G_0)^{\boxtimes m} \homm H$ will be taken to mean
$G' \homm H$ where $(G', G'_0) = (G, G_0)^{\boxtimes m}$; similarly
$\thbar((G, G_0)^{\boxtimes m})$ will be taken to mean $\thbar(G')$ and
$(G, G_0)^{\boxtimes m} \supseteq G^{\boxtimes m}$ to mean
$G' \supseteq G^{\boxtimes m}$.

Now we can show that the condition $P(x,u|i) \ne 0$ only when $x=i$ can be dropped in
\cref{thm:costrate_classical}.
This is only a slight generalization of~\cite{1705019}.
We will later generalize this to quantum sources and quantum channels.
\begin{proposition}
    \label{thm:generalized_costrate}
    Let $P(x,u|i)$ be a classical source and $\cN(v|s)$ a classical channel.
    Let graphs $G$ and $H$ be given by~\eqref{eq:chargraph} and~\eqref{eq:distinguishgraph}.
    Then $m$ parallel instances of the source can be sent using $n$ parallel instances of
    the channel only if
    \begin{align*}
        \frac{n}{m} \ge \frac{\log \thbar(G)}{\log \thbar(H)}.
    \end{align*}
\end{proposition}
\begin{proof}
    Without loss of generality assume that each $x$ is possible.
    In other words, assume that for each $x$ there is an $i$ and $u$ such that
    $P(x,u|i) \ne 0$.  Generality is not lost because one can decrease the alphabet
    associated with $x$, removing values that can never occur.
    Reducing this alphabet only removes isolated vertices from $G$, and so doesn't affect
    the value of $\thbar(G)$.
    Let $G_0$ be defined as in~\eqref{eq:G0}.  Since each $x$ is possible, this graph has
    loops on all vertices: $x \sim_{G_0} x$ for all $x$.

    As per the above discussion, the composite source (consisting of $m$ parallel
    instances of $P(x,u|i)$) will have characteristic graph $(G, G_0)^{\boxtimes m}$.
    Since $G_0$ has loops on all vertices, our generalized strong product~\eqref{eq:Gpp}
    has at least as many edges as the standard strong product~\eqref{eq:strong_prod}.
    Since $\thbar$ is monotone increasing under addition of edges and is multiplicative
    under the strong product~\cite{knuth94} we have
    $\thbar((G, G_0)^{\boxtimes m}) \ge \thbar(G^{\boxtimes m}) = \thbar(G)^m$.

    The distinguishability graph of $n$ parallel instances of the channel $\cN(v|s)$ is
    $H^{\djp n}$.  Since $\thbar$ is multiplicative under the disjunctive
    product~\cite{lovasz79} we have $\thbar(H^{\djp n}) = \thbar(H)^n$.
    If $m$ parallel sources can be sent via $n$ parallel channels then
    $(G, G_0)^{\boxtimes m} \homm H^{\djp n}$.  Since $\thbar$ is monotone under
    homomorphisms,
    \begin{align*}
        (G, G_0)^{\boxtimes m} \homm H^{\djp n}
        &\implies \thbar((G, G_0)^{\boxtimes m}) \le \thbar(H^{\djp n})
        \\ &\implies \thbar(G)^m \le \thbar(H)^n
        \\ &\implies \frac{n}{m} \ge \frac{\log \thbar(G)}{\log \thbar(H)}.
    \end{align*}
\end{proof}


Similar arguments apply for quantum source-channel coding.
It is easy to see that the confusability graphs for parallel channels should combine by tensor
product since the Kraus operators combine by tensor product.  We have been using instead the
distinguishability graph, which then combines as $(S^\perp \ot T^\perp)^\perp$.  We take
this as the definition of disjunctive product:
\begin{definition}
    \label{def:disjunctive_nc}
    Let $S \subseteq \linop{A}$ and $T \subseteq \linop{B}$ be non-commutative graphs.
    Their \emph{disjunctive product} is
    $S \djp T = S \ot \linop{B} + \linop{A} \ot T
    = (S^\perp \ot T^\perp)^\perp$.
\end{definition}
When $S$ and $T$ derive from classical graphs this definition is equivalent
to~\eqref{eq:disjunctive_classical}.
We will use the notation $S^{\djp n} := S \djp S \djp \dots \djp S$.
Analogous to~\eqref{eq:shannon_H}, the Shannon capacity of a quantum channel with
distinguishability graph $T$ is
\begin{align*}
    C_0(T^\perp) = \lim_{n \to \infty} \frac{1}{n} \log \cliquec(T^{\djp n})
\end{align*}
It is known that $\qthperp(T)$ is an upper bound on $C_0(T^\perp)$, since
$\qthperp(T^{\djp n}) = \qth((T^\perp)^{\ot n}) = \qthperp(T)^n$~\cite{dsw2013}.

Consider now two parallel sources, with characteristic graphs $S$ and $S'$.
Analogous to~\eqref{eq:G0} we define $(S,S_0)$, a \emph{non-commutative graph
with generalized loops}.
For discrete QSCC, the subject of \cref{thm:fgame}, define
\begin{align}
    S &= \Tr_{BC}\{ \linop{C} J K_r J^\dag \} \notag
    \\ S_0 &= \Tr_{BC}\{ \linop{C} J K_r^\perp J^\dag \}
    \label{eq:S0_fgame}
\end{align}
and for coherent QSCC, the subject of \cref{thm:fgame_coh}, define
\begin{align}
    S &= \Tr_{BC}\{ \linop{C} J Q_r J^\dag \} \notag
    \\ S_0 &= \Tr_{BC}\{ \linop{C} J Q_r^\perp J^\dag \} \notag
    \\ &= \Tr_{BC}\{ \linop{C} J J^\dag \},
    \label{eq:S0_fgame_coh}
\end{align}
where $J : R \to A\ot B\ot C$ and $r = \dim(R)$.
Analogous to~\eqref{eq:Gpp}-\eqref{eq:Gpp0} define the strong product
$(S'',S''_0) = (S, S_0) \boxtimes (S', S'_0)$ as
\begin{align}
    S'' &= (S \ot S') + (S_0 \ot S') + (S \ot S'_0) \notag
    \\ S''_0 &= S_0 \ot S'_0.
    \label{eq:noncomm_strong}
\end{align}
If $S,S_0,S',S'_0$ correspond to classical graphs $G,G_0,G',G'_0$ then this product
corresponds to the classical graph $(G, G_0) \boxtimes (G', G'_0)$.
If $G_0$ and $G'_0$
consist of only loops on each vertex (i.e.\ $S_0 = \linspan\{ \proj{x} \}$
and similarly for $S'_0$) then this corresponds to $G \boxtimes G'$.
Define the graph power $(S, S_0)^{\boxtimes m}$ to be repeated application
of~\eqref{eq:noncomm_strong}.

Other graph products could be defined similarly.
For example, the Cartesian product of graphs, $G \cart G'$ is defined to have edges
$(x,x') \sim (y,y') \iff (x=y \wedge x' \sim y') \vee (x \sim y \wedge x'=y')$,
so for non-commutative graphs one could define
$(S'',S''_0) = (S, S_0) \cart (S', S'_0)$ with
$S'' = (S_0 \ot S') + (S \ot S'_0)$ and $S''_0 = S_0 \ot S'_0$.
The complement of a graph has edges $x \not\sim y \wedge x\ne y$, which would have
non-commutative analogue $\overline{(S,S_0)} = (S^\perp \setminus S_0, S_0)$,
assuming $S_0 \subseteq S^\perp$.
We will not have occasion to consider such constructions, but mention it as a starting
point for possible development of a richer theory of non-commutative graphs.
A similar concept was explored in~\cite{arxiv:1002.2514}; however, they suggested a
specific form of $S_0$ in terms of the multiplicative domain of a channel whereas we leave
the form of $S_0$ to be determined by the application at hand.

As before, we abuse notation and take $(S,S_0)^{\boxtimes m} \homm T$
to mean $S' \homm T$ where $(S',S'_0)=(S,S_0)^{\boxtimes m}$, and
$\thbar((S, S_0)^{\boxtimes m})$ to mean $\thbar(S')$.
The strong product~\eqref{eq:noncomm_strong} indeed corresponds to the
characteristic graph of parallel sources:

\begin{theorem}
    \label{thm:product_coding}
    Consider discrete QSCC with two sources
    $\{ \rho_i \}_{i \in \{1,\dots,r\}} \subseteq \linop{A} \ot \linop{B}$ and
    $\{ \rho'_i \}_{i \in \{1,\dots,r'\}} \subseteq \linop{A'} \ot \linop{B'}$.
    As in \cref{thm:fgame}, for each $i$ let
    $\ket{\psi_i} \in A \ot B \ot C$ be a purification of $\rho_i$ and define the
    isometry $J : R \to A\ot B\ot C$ with $R = \mathbb{C}^r$ as
    $J = \sum_{i=1}^r \ket{\psi_i}\bra{i}$.
    Define $\ket{\psi'_i}$ and $J'$ similarly.

    Let $(S,S_0)$ and $(S',S'_0)$ be the characteristic graphs (with generalized loops) for
    these two sources, as defined by~\eqref{eq:S0_fgame}, and similarly
    $(S'',S''_0)$ for the joint source $\{  \ket{\psi_i} \ot \ket{\psi'_{i'}} \}_{i,i'}$.
    Then it holds that $(S'',S''_0) = (S, S_0) \boxtimes (S', S'_0)$.
    These two sources can be sent using one copy of the channel $\chan{N}$ iff
    \begin{align}
        (S, S_0) \boxtimes (S', S'_0) \homm T
        \label{eq:homSST}
    \end{align}
    where $T = (N^\dag N)^\perp$.

    The analogous statement holds for coherent QSCC, where now
    $(S,S_0)$, $(S',S'_0)$, and $(S'',S''_0)$ are defined using~\eqref{eq:S0_fgame_coh}
    rather than~\eqref{eq:S0_fgame}.

    In either case (discrete or coherent QSCC),
    it is possible to send $m$ copies of a source using $n$ copies of a channel iff
    \begin{align}
        (S, S_0)^{\boxtimes m} \homm T^{\djp n}.
        \label{eq:homSmTn}
    \end{align}
\end{theorem}
\begin{proof}
    We give the proof only for discrete QSCC; the proof for coherent QSCC follows from
    similar arguments.
    A state from the joint source will be of the form
    $\ket{\psi''_{ii'}} = \ket{\psi_i} \ot \ket{\psi'_{i'}}$ and the corresponding
    isometry will be $J'' = J \ot J'$, so we have (according to~\eqref{eq:S0_fgame})
    \begin{align*}
        S'' &= \Tr_{BB'CC'}\{ \linop{C \ot C'} J'' K_{r''} J''^\dag \}
        \\
        S''_0 &= \Tr_{BB'CC'}\{ \linop{C \ot C'} J'' K_{r''}^\perp J''^\dag \}
    \end{align*}
    where $r'' = r r' = \dim(R) \dim(R')$.
    It is readily verified that $(S'', S''_0) = (S, S_0) \boxtimes (S', S'_0)$, since
    \begin{align*}
        K_{r''} &= K_r \ot K_{r'} + K_r^\perp \ot K_{r'} + K_r \ot K_{r'}^\perp,
        \\ K_{r''}^\perp &= K_r^\perp \ot K_{r'}^\perp.
    \end{align*}
    By \cref{thm:fgame}, the joint source can be sent using a single use of channel $\cN$
    iff $S'' \homm T$, that is to say iff condition~\eqref{eq:homSST} holds.

    By induction, $m$ instances of a source can be sent with a single channel use iff
    $(S, S_0)^{\boxtimes m} \homm T$.
    Since the distinguishability graph of $n$ copies of the channel is $T^{\djp n}$, it is possible
    to send $m$ instances of the source using $n$ instances of the channel iff
    $(S, S_0)^{\boxtimes m} \homm T^{\djp n}$.
\end{proof}

For classical source-channel coding the amount of communication needed to transmit a joint
source is at least as much as is needed for each individual source, since the second
source can always be simulated: Alice and Bob can just agree ahead of time on some $x'$
and $u'$ that can be emitted be the second source.
Somewhat surprisingly, this does not necessarily hold for quantum source-channel coding.
For example, consider the following two sources.
The first source is some classical source for which an entanglement resource
$\proj{\lambda}$ would allow for more efficient transmission.
In other words, $\chromc(S)$ is large and $\chromce(S)$ is small.
Examples of such graphs are given in, e.g.~\cite{david2006quantum}.
The second source consists of only a single possible input: $\proj{\lambda}$.
So $S'=\varnothing$ and $S'_0=\mathbb{C} \Lambda$ where
$\Lambda = \Tr_B\{\proj{\lambda}\}$.
Then the first source requires an amount of communication $\chromc(S)$, the second requires no
communication (i.e.\ $\chromc(S')=1$), but the joint source requires
communication $\chromce(S) < \max\{ \chromc(S), \chromc(S') \}$.

Entanglement assisted chromatic number does not exhibit this same anomaly.
Indeed, the joint source can never be easier to transmit than either of the individual
sources since Alice and Bob can always simulate (some particular input from)
the second source, by choosing said state ahead of time and adding this to their
entanglement resource.
For a similar reason, even without entanglement assistance a joint source is not easier to
transmit than the individual sources in the case where the individual sources are each
capable of producing a product state: Alice and Bob can simulate any of these sources by
producing the product state themselves, in order to turn a single source into (a subset
of) the joint source.

For classical source-channel coding, we defined the cost rate as the infimum of $n/m$ such
that $m$ instances of the source can be transmitted using $n$ instances of the channel.
As per the above discussion, this can be achieved iff
$(G, G_0)^{\boxtimes m} \homm H^{\djp n}$,
so the cost rate is
\begin{align*}
    \lim_{m \to \infty} \frac{1}{m} \min\left\{
        n : (G, G_0)^{\boxtimes m} \homm H^{\djp n}
        \right\}.
\end{align*}
Cost rate for quantum source-channel coding can be defined similarly,
\begin{align}
    \label{eq:costrate_nc}
    \lim_{m \to \infty} \frac{1}{m} \min\left\{
        n : (S, S_0)^{\boxtimes m} \homm T^{\djp n}
        \right\},
\end{align}
where $(S,S_0)$ is the characteristic graph of the source
(as per~\eqref{eq:S0_fgame} or~\eqref{eq:S0_fgame_coh})
and $T$ is the distinguishability graph of the channel.
Similarly, the entanglement assisted cost rate is
\begin{align}
    \label{eq:costrate_nc_ent}
    \lim_{m \to \infty} \frac{1}{m} \min\left\{
        n : (S, S_0)^{\boxtimes m} \home T^{\djp n}
        \right\}.
\end{align}
Clearly $\eqref{eq:costrate_nc_ent} \le \eqref{eq:costrate_nc}$.

In all three cases, by Fekete's lemma the limit exists and is equal to the infimum, since
the $\min\{n : \dots\}$ quantities are sub-multiplicative in $m$.  This can be seen as
follows: if it is possible to transmit $m_1$ instances of the source using $n_1$ instances
of the channel via one protocol, and $m_2$ instances of the source using $n_2$ instances
of the channel via another protocol, one can transmit $m_1+m_2$ instances of the source
using $n_1+n_2$ instances of the channel by simply running the two protocols in parallel.

For the classical case, the Lov{\'a}sz $\thbar$ number is multiplicative under the
relevant graph products and is a homomorphism monotone, so it leads to a lower bound on the
cost rate, \cref{thm:generalized_costrate}.
A similar bound applies for quantum source-channel coding, with a caveat.
The $\qthperp$ quantity is not multiplicative under strong product in general; however, it is
when $S_0$ and $S'_0$ contain the identity.
So our generalization of \cref{thm:generalized_costrate} will require $I \in S_0$.
This happens for example when the states emitted by the source include a maximally entangled state,
or product states with Alice's shares forming a complete orthonormal basis (such as is the case
with classical source-channel coding).
We have then the following bound on cost rate.

\begin{theorem}
    \label{thm:costrate_nc}
    Consider a source with characteristic graph $(S,S_0)$,
    defined as in~\eqref{eq:S0_fgame} for discrete QSCC
    or as in~\eqref{eq:S0_fgame_coh} for coherent QSCC\@.
    Consider a noisy quantum channel $\cN$ with distinguishability graph
    $T = (N^\dag N)^\perp$.
    If $I \in S_0$ then the entanglement assisted cost rate~\eqref{eq:costrate_nc_ent} is
    lower bounded by $\log \qthperp(S) / \log \qthperp(T)$.
\end{theorem}
\begin{proof}
    Since $I \in S_0$, the $\qthperp$ quantity is multiplicative under both strong and
    disjunctive graph powers, by \cref{thm:qthperp_mult}.
    Using this fact, and the fact that $\qthperp$ is monotone under entanglement assisted
    homomorphisms, we have
    \begin{align*}
        \eqref{eq:costrate_nc_ent}
        &\ge \inf_{m \to \infty} \frac{1}{m} \min\left\{
            n : \qthperp((S, S_0)^{\boxtimes m}) \le \qthperp(T^{\djp n})
            \right\}
        \\ &= \inf_{m \to \infty} \frac{1}{m} \min\left\{
            n : \qthperp(S)^m \le \qthperp(T)^n
            \right\}
        \\ &= \inf_{m \to \infty} \frac{1}{m} \min\left\{
            n : \log \qthperp(S) / \log \qthperp(T) \le n/m
            \right\}
        \\ &= \log \qthperp(S) / \log \qthperp(T).
    \end{align*}
\end{proof}

We now prove the lemma used in the preceding proof.

\begin{lemma}
    \label{thm:qthperp_mult}
    Let $S$ and $S'$ be trace-free non-commutative graphs.
    Then,
    \begin{itemize}
        \item $\qthperp(S \djp S') = \qthperp(S) \qthperp(S')$
        \item $\qthperp((S, S_0) \boxtimes (S', S'_0)) = \qthperp(S) \qthperp(S')$
            if $I \in S_0$ and $I \in S'_0$
    \end{itemize}
\end{lemma}
\begin{proof}
    From~\cite{dsw2013} we have $\qth(S^\perp \ot S'^\perp) = \qth(S^\perp) \qth(S'^\perp)$.
    But $(S^\perp \ot S'^\perp)^\perp = S \djp S'$, so $\qthperp(S \djp S') = \qthperp(S)
    \qthperp(S')$.
    Since $(S, S_0) \boxtimes (S', S'_0) \subseteq S \djp S'$ and since
    $\qthperp$ is monotone decreasing under subsets, we have
    \begin{align*}
        \qthperp((S, S_0) \boxtimes (S', S'_0)) \le \qthperp(S \djp S') = \qthperp(S)
        \qthperp(S').
    \end{align*}
    Let $X$ be an optimal solution to~\eqref{eq:qth_primal} for $\qthperp(S)$, from
    \cref{def:qth}.
    Then $X \in S \ot \linop{B}$ (for some Hilbert space $B$), $I+X \succeq 0$,
    and $\opnorm{I+X} = \qthperp(S)$.
    Similarly, there is an $X' \in S' \ot \linop{B'}$, $I+X' \succeq 0$,
    and $\opnorm{I+X'} = \qthperp(S')$.
    Define
    \begin{align*}
        X'' &= (I_{AB}+X) \ot (I_{A'B'}+X') - I_{AA'BB'}.
    \end{align*}
    Clearly $I+X'' \succeq 0$.
    Also,
    \begin{align*}
        X'' &= X \ot X' + I_{AB} \ot X' + X \ot I_{A'B'}
        \\ &\in [S \ot S' + I_A \ot S' + S \ot I_{A'}] \ot \linop{BB'}
        \\ &\subseteq [S \ot S' + S_0 \ot S' + S \ot S'_0] \ot \linop{BB'}
        \\ &= [(S, S_0) \boxtimes (S', S'_0)] \ot \linop{BB'}.
    \end{align*}
    So $X''$ is feasible for~\eqref{eq:qth_primal} for $\qthperp((S, S_0) \boxtimes (S', S'_0))$.
    Therefore
    \begin{align*}
        \qthperp((S, S_0) \boxtimes (S', S'_0)) &\ge \opnorm{I+X''}
        \\ &= \opnorm{(I+X) \ot (I+X')}
        \\ &= \qthperp(S) \qthperp(S').
    \end{align*}
\end{proof}

\section{Schrijver and Szegedy numbers}
\label{sec:schrijver}

In this section we will provide a generalization to non-commutative graphs for two
quantities related to Lov{\'a}sz's $\vartheta$: Schrijver's $\thm$ and Szegedy's $\thp$.
Schrijver's number comes from adding extra constraints to the maximization program for
$\vartheta$, yielding a smaller value; Szegedy's number comes from adding extra constraints
to the minimization (dual) program for $\vartheta$, yielding a larger value.
We will consider the complimentary quantities $\thmbar(G)=\thm(\Gc)$ and
$\thpbar(G)=\thp(\Gc)$.
These are homomorphism monotones in the same sense that $\thbar$
is~\cite{de2013optimization}; therefore they satisfy the sandwich theorem
\begin{align*}
    \omega(G) \le \thmbar(G) \le \thbar(G) \le \thpbar(G) \le \chi(G).
\end{align*}
These quantities are not suitable for bounding asymptotic channel capacity or cost rate
for source-channel coding because they are not multiplicative under the appropriate graph
products~\cite{6880319}.

For classical graphs these quantities have been shown to be monotone under entanglement
assisted homomorphisms~\cite{6880319}.
Strangely enough, our generalization to non-commutative graphs will yield quantities
monotone under homomorphisms but not under entanglement assisted homomorphisms.
For classical graphs the gap between $\thmbar(G)$ and $\thpbar(G)$ tends to be small or,
often, zero.
For non-commutative graphs the gap tends to be much more extreme, sometimes infinite, even
for random graphs of small dimension.
After developing basic properties of these quantities we will show how they can be
used to reproduce some results from the literature regarding entanglement assisted
activation of one-shot channel capacity and impossibility of one-way LOCC measurement of
entangled states.
Also we will provide a channel for which maximally entangled states are not sufficient for
achieving the entanglement assisted one-shot capacity.

The classical quantities are defined as
follows~\cite{lovasz79,knuth94,1056072,mceliece1978lovasz,365707}.
\begin{definition}
    The Lov{\'a}sz, Schrijver, and Szegedy numbers of the complement of a graph,
    $\thbar(G)$, $\thmbar(G)$, and $\thpbar(G)$, are defined
    by the following dual (and equivalent) semidefinite programs.
    All matrices are either real or complex (it doesn't matter), $J$ is the all-ones
    matrix, and $\nonneg$ is the cone of symmetric entrywise non-negative matrices.
    Take $S = \linspan\{ \ket{x}\bra{y} : x \sim y \}$ and
    $S_0 = \linspan\{ \ket{x}\bra{x} : x \in V(G) \}$ (the diagonal matrices).
    \begin{align}
        \thbar  (G) &= \max \ip{B, J} & \textrm{ s.t. } &B \succeq 0, \Tr B = 1,
            \notag \\ &\, &&B \in S + S_0
            \label{eq:lovasz_max}
        \\ \thmbar(G) &= \max \ip{B, J} & \textrm{ s.t. } &B \succeq 0, \Tr B = 1,
            \notag \\ &\, &&B \in S + S_0, B \in \nonneg
            \label{eq:schrij_max}
        \\ \thpbar(G) &= \max \ip{B, J} & \textrm{ s.t. } &B \succeq 0, \Tr B = 1,
            \notag \\ &\, &&B + L \in S + S_0, L \in \nonneg
            \label{eq:szegedy_max}
        \\ \thbar  (G) &= \min \lambda & \textrm{ s.t. } &Z \succeq J, (Z_{ii} = \lambda \textrm{ for all } i),
            \notag \\ &\, &&Z \in S^\perp
            \label{eq:lovasz_min}
        \\ \thmbar(G) &= \min \lambda & \textrm{ s.t. } &Z \succeq J, (Z_{ii} = \lambda \textrm{ for all } i),
            \notag \\ &\, &&Z + L \in S^\perp, L \in \nonneg
            \label{eq:schrij_min}
        \\ \thpbar(G) &= \min \lambda & \textrm{ s.t. } &Z \succeq J, (Z_{ii} = \lambda \textrm{ for all } i),
            \notag \\ &\, &&Z \in S^\perp, Z \in \nonneg
            \label{eq:szegedy_min}
    \end{align}
\end{definition}
The constraint $B \in \nonneg$ that is added to~\eqref{eq:lovasz_max}
to yield~\eqref{eq:schrij_max} has the following justification.
Suppose that $W \subseteq V(G)$ is a clique.  Then the matrix
\begin{align*}
    B_{ij} =
        \begin{cases}
            1/\abs{W} & \textrm{if } i,j \in W \\
            0 & \textrm{otherwise}
        \end{cases}
\end{align*}
is a feasible solution to~\eqref{eq:lovasz_max} with value $\abs{W}$.
So $\omega(G) \le \thbar(G)$.  But $B \in \nonneg$, so this condition can be added to the
maximization program to yield a potentially smaller quantity $\thmbar(G)$ that still upper
bounds $\omega(G)$.\footnote{
    An even tighter constraint, requiring $B$ to be completely positive, yields $\omega(G)$
    exactly~\cite{doi:10.1137/S1052623401383248}.
}
Similarly, if $f : V(G) \to \{1,\dots,m\}$ is a proper coloring of $G$ then
\begin{align*}
    Z_{ij} =
        \begin{cases}
            m & \textrm{if } f(i)=f(j) \\
            0 & \textrm{otherwise}
        \end{cases}
\end{align*}
is feasible for~\eqref{eq:lovasz_min} with value $\chi(G)$, so
$\thbar(G) \ge \chi(G)$.  Since this satisfies $Z \in \nonneg$, adding this condition
to the minimization program gives a quantity $\thpbar(G)$ still lower bounding $\chi(G)$.
We will follow this sort of strategy to create analogues of $\thmbar$ and $\thpbar$
for non-commutative graphs.

The primal program for $\qthperp$ can be written~\cite{dsw2013}
\begin{align}
    \qthperp(S) = \max &\; \braopket{\Phi}{ I \ot \rho + T }{\Phi} \notag
    \\ \textrm{s.t. } &\; \rho \succeq 0, \Tr \rho = 1, \notag
    \\ &\; I \ot \rho + T \succeq 0, \notag
    \\ &\; T \in S \ot \linop{A'},
    \label{eq:qth_rho_T}
\end{align}
where $A'$ is an ancillary system of the same dimension as $A$ and
$\ket{\Phi} = \sum_i \ket{i}_A \ot \ket{i}_{A'}$.
With this definition it is easy to see that $\cliquec(S) \le \qthperp(S)$:
since $\cliquec(S)$ is the classical communication capacity of the distinguishability graph $S$,
there are $m = \cliquec(S)$ vectors $\ket{\psi_1},\dots,\ket{\psi_m} \in A$
such that $\ket{\psi_i}\bra{\psi_j} \in S$ for $i \ne j$.
Define
\begin{align}
    \notag
    T &= \frac{1}{m} \sum_{i \ne j}
        \ket{\psi_i} \bra{\psi_j}_A \ot
        \ket{\conj{\psi_i}} \bra{\conj{\psi_j}}_{A'} \textrm{ and}
    \\ \rho &= \frac{1}{m} \sum_i \ket{\psi_i} \bra{\psi_i}_{A'},
    \label{eq:T_feas_c}
\end{align}
where a bar over a vector represents complex conjugation in the computational basis.
This is readily verified to be a feasible solution to~\eqref{eq:qth_rho_T} with value $m$.
A tighter upper bound on $\cliquec(S)$ can be obtained by adding constraints
to~\eqref{eq:qth_rho_T}.  As long as~\eqref{eq:T_feas_c} remains feasible under these new
constraints, the modified program will remain an upper bound on $\cliquec(S)$.

To this end, consider the ``rotated transpose'' linear superoperator
$\rot : \linop{A}\ot\linop{A'} \to \linop{A}\ot\linop{A'}$ with action
\begin{align*}
    \rot\left(
        \ket{i} \bra{j}_A  \ot \ket{k} \bra{l}_{A'}
    \right) &=
        \ket{i} \bra{k}_A  \ot \ket{j} \bra{l}_{A'}
\end{align*}
on standard basis states and
\begin{align*}
    \rot\left(
        \ket{\psi} \bra{\phi}_A  \ot \ket{\chi} \bra{\xi}_{A'}
    \right) &=
        \ket{\psi} \bra{\conj{\chi}}_A  \ot \ket{\conj{\phi}} \bra{\xi}_{A'}
\end{align*}
in general.
Note that $\rot$ is an involution (it is its own inverse).
Define the double-dagger operation
\begin{align*}
    X^\ddag = \rot(\rot(X)^\dag).
\end{align*}
We have $\rot(I_{A} \ot I_{A'}) = \proj{\Phi}$.
The $T$ from~\eqref{eq:T_feas_c} transforms as
\begin{align*}
    \rot(T) &= \frac{1}{m} \sum_{i \ne j} \proj{\psi_i} \ot \proj{\conj{\psi_j}}.
\end{align*}
Since $\rot(T)$ is a separable operator, we may add this as an extra constraint
in~\eqref{eq:qth_rho_T} to get a tighter bound on $\cliquec(S)$.

In general, consider some closed convex cone $\cone \subseteq \linop{A}\ot\linop{A'}$ and
a trace-free
non-commutative graph $S$.
We consider only cones over the real inner product space of Hermitian matrices.
For $S \in \linop{A}$, we use the notation
$\conj{S} := \{ \conj{M} : M \in S \} \subseteq \linop{A'}$,
where a bar over an operator denotes entrywise complex conjugate, with the conjugated
operator moved into the primed space (as discussed in \cref{sec:qthmon}).
Define the semidefinite program
\begin{align}
    \qthmperp{\cone}(S) = \max &\; \braopket{\Phi}{ I \ot \rho + T }{\Phi} \notag
    \\ \textrm{s.t. } &\; \rho \succeq 0, \Tr \rho = 1, \notag
    \\ &\; I \ot \rho + T \succeq 0, \notag
    \\ &\; T \in S \ot \conj{S}, \notag
    \\ &\; \rot(T) \in \cone.
    \label{eq:qthm_C_primal}
\end{align}
Note that $T \in S \ot \linop{A'}$ and $\rot(T) \in \cone$ implies $T \in S \ot \conj{S}$,
since $\cone$ contains only Hermitian operators.
We choose to explicitly state the condition $T \in S \ot \conj{S}$
in~\eqref{eq:qthm_C_primal}.

Since linear programming duality turns constraints into variables, the dual of this program is
similar to~\eqref{eq:qth_dual} but with an extra variable that runs over the dual cone
$\cone^*$.  In \cref{sec:duality} we show that strong duality holds, so that primal and dual
have equal value.  The dual program is
\begin{align}
    \qthmperp{\cone}(S) = \min &\; \opnorm{\Tr_A Y} \notag
    \\ \textrm{s.t. } &\; Y + (L + L^\dag) \in
        S^\perp \djp \conj{S}^\perp = (S \ot \conj{S})^\perp, \notag
    \\ &\; \rot(L) + \rot(L)^\dag \in \cone^*, \notag
    \\ &\; Y \succeq \proj{\Phi}, \notag
    \\ &\; L \in \linop{A} \ot \linop{A'}.
    \label{eq:qthm_C_dual}
\end{align}
Recall that ``$\djp$'' denotes the disjunctive product from \cref{def:disjunctive_nc}.
The point $\rho=I/\dim(A)$, $T=0$ is feasible for~\eqref{eq:qthm_C_primal},
giving $\qthmperp{\cone}(S) \ge 1$.
In \cref{sec:duality} we provide a feasible point for~\eqref{eq:qthm_C_dual},
giving $\qthmperp{\cone}(S) < \infty$.

Denote by $\sep$ the cone of separable operators in $\linop{A}\ot\linop{A'}$.
Since~\eqref{eq:T_feas_c} satisfies $\rot(T) \in \sep$,
it is feasible for~\eqref{eq:qthm_C_primal} for $\qthmperp{\sep}$.  Therefore
$\cliquec(S) \le \qthmperp{\sep}(S)$.
One can also show
$\cliqueq(S)^2 \le \qthmperp{\sep}(S)$ by similar means, but we will eventually obtain this
result by showing $\qthmperp{\sep}$ to be a homomorphism monotone in the same sense that
$\qthperp$ is.

From a computational perspective $\qthmperp{\sep}(S)$ is not the most convenient because there is
no efficient way to determine whether an operator is in $\sep$.
Fortunately there are closed convex cones containing $\sep$ that are efficiently optimized over and
that give good bounds on $\cliquec(S)$ and $\cliqueq(S)$.
Namely, consider $\psdcone$, the cone of positive semidefinite matrices,
$\ppt$, the cone of matrices with positive semidefinite partial transpose, or
even $\psdcone \cap \ppt$.
Note that $\psdcone$ and $\ppt$ are self-dual and the dual of $\psdcone \cap \ppt$ is
$\psdcone + \ppt$.
The dual of $\sep$ is
$\sep^* = \{ W : \ip{W, M} \ge 0 \textrm{ for all } M \in \sep \}$.
We have
\begin{align}
    \label{eq:qthm_chain}
    \cliquec(S) \le \qthmperp{\sep}(S) \le \qthmperp{\psdcone \cap \ppt}(S)
        \le \qthmperp{\psdcone}(S) \le \qthperp(S).
\end{align}
This sequence of refinements is reminiscent of the approximations to the copositive cone
that yield the Lov{\'a}sz and Schrijver numbers for classical
graphs~\cite{doi:10.1137/S1052623401383248,bomze2010gap}.
In fact the middle three values in the above chain of inequalities
collapse to Schrijver's number when $S$ derives from a classical graph.

\begin{theorem}
    \label{thm:qthm_classical}
    Let $G$ be a classical loop-free graph and $S = \linspan\{ \ket{i}\bra{j} : i \sim j \}$.
    Then for any closed convex cone $\cone$ satisfying $\sep \subseteq \cone \subseteq \sep^*$,
    it holds that $\qthmperp{\cone}(S) = \thmbar(G)$.
\end{theorem}
\begin{proof}
    Define the isometry $V = \sum_i \ket{ii} \bra{i}$.
    Let $T$ and $\rho$ be an optimal solution
    for~\eqref{eq:qthm_C_primal} for $\qthmperp{\sep^*}(S)$.
    We will show that $B = V^\dag (I \ot \rho + T) V$ is feasible for~\eqref{eq:schrij_max}.
    This matrix has coefficients
    \begin{align*}
        B_{ij} &= \rho_{ii} \delta_{ij} + \braopket{ii}{T}{jj}
        \\ &= \rho_{ii} \delta_{ij} + \braopket{ij}{\rot(T)}{ij}.
    \end{align*}
    Since $\rho_{ii} \ge 0$, $\rot(T) \in \sep^*$, and $\ket{ij}\bra{ij} \in \sep$,
    it holds that $B_{ij} \ge 0$ for all $i,j$.
    So $B \in \nonneg$.
    We have $I \ot \rho + T \succeq 0 \implies B \succeq 0$.
    Since $T \in S \ot \conj{S}$ we have $\braopket{ii}{T}{jj} = 0$
    when $i \not\sim j$.
    In particular, $B_{ii} = \rho_{ii}$ and $B_{ij}=0$ when $i \not\sim j$, $i \ne j$.
    Since $\Tr \rho = 1$, also $\Tr B = 1$.
    So $B$ is feasible for~\eqref{eq:schrij_max}.
    Its value is
    \begin{align*}
        \ip{B, J} &= \sum_{ij} B_{ij}
        = \sum_{ij} \braopket{ii}{I \ot \rho + T}{jj}
        \\ &= \braopket{\Phi}{I \ot \rho + T}{\Phi}
        = \qthmperp{\sep^*}(S).
    \end{align*}
    Therefore $\thmbar(G) \ge \qthmperp{\sep^*}(S)$.

    Now let $B$ be an optimal solution for~\eqref{eq:schrij_max}.
    Decompose this into diagonal and off-diagonal components: $B = \rho + T'$.
    Define $T=VT'V^\dag$.  We will show these to be feasible
    for~\eqref{eq:qthm_C_primal}
    for $\qthmperp{\sep}(S)$.
    For any vector $\ket{\psi} \in A \ot A'$ we have
    \begin{align*}
        \braopket{\psi}{I \ot \rho + T}{\psi}
        &= \sum_{ij} \abs{\psi_{ij}^2} \rho_{jj}
            + \sum_{i \ne j} \psi_{ii}^* T'_{ij} \psi_{jj}
        \\ &= \sum_{i \ne j} \abs{\psi_{ij}^2} \rho_{jj}
            + \sum_{ij} \psi_{ii}^* B_{ij} \psi_{jj}
        \ge 0,
    \end{align*}
    where the last inequality follows from $\rho_{jj} \ge 0$ and $B \succeq 0$.
    Therefore $I \ot \rho + T \succeq 0$.
    We have
    \begin{align*}
        \rot(T) &= \sum_{ij} T'_{ij} \rot(\ket{ii}\bra{jj})
        \\ &= \sum_{ij} T'_{ij} \ket{ij}\bra{ij}
        \in \sep,
    \end{align*}
    where the last relation requires $T'_{ij} \ge 0$.
    Clearly $\rho \succeq 0$ and $\Tr \rho = \Tr B = 1$.
    For $i \not\sim j$ we have $T'_{ij} = 0 \implies (\bra{i} \ot I) T (\ket{j} \ot I) = 0$,
    giving $T \in S \ot \linop{A'}$.
    Similarly, $T \in \linop{A} \ot \conj{S}$.
    So in fact $T \in (S \ot \linop{A'}) \cap (\linop{A} \ot \conj{S})
        = S \ot \conj{S}$.
    Therefore $\rho$ and $T$ are feasible for~\eqref{eq:qthm_C_primal} for
    $\qthmperp{\sep}(S)$.  This solution has value
    \begin{align*}
        \braopket{\Phi}{ I \ot \rho + T }{\Phi}
        &= \sum_{ij} \braopket{ii}{I \ot \rho + T}{jj}
        \\ &= \sum_i \rho_{ii} + \sum_{ij} T'_{ij}
        = \ip{B, J} = \thmbar(G),
    \end{align*}
    giving $\qthmperp{\sep}(S) \ge \thmbar(G)$.

    Clearly $\sep \subseteq \cone \subseteq \sep^* \implies
    \qthmperp{\sep}(S) \le \qthmperp{\cone}(S) \le \qthmperp{\sep^*}(S)$ since maximization
    programs have nondecreasing value as constraints are loosened.
    Combining this with the above two inequalities gives the desired equality result.
\end{proof}

A generalization of Szegedy's number to non-commutative graphs follows similarly, now adding
extra constraints to the dual program~\eqref{eq:qth_dual}.  Extra constraints on the dual
become extra variables in the primal.
For a closed convex cone $\cone$ of operators in $\linop{A}\ot\linop{A'}$ and for
a trace-free non-commutative graph $S$, the primal and dual take the form
\begin{align}
    \qthpperp{\cone}(S) = \max &\; \braopket{\Phi}{ I \ot \rho + T }{\Phi} \notag
    \\ \textrm{s.t. } &\; \rho \succeq 0, \Tr \rho = 1, \notag
    \\ &\; I \ot \rho + T \succeq 0, \notag
    \\ &\; T + (L + L^\dag) \in
        S \djp \conj{S} = (S^\perp \ot \conj{S}^\perp)^\perp, \notag
    \\ &\; \rot(L) + \rot(L)^\dag \in \cone^*, \notag
    \\ &\; L \in \linop{A} \ot \linop{A'},
    \label{eq:qthp_C_primal}
    \\
    \qthpperp{\cone}(S) = \min &\; \opnorm{\Tr_A Y} \notag
    \\ \textrm{s.t. } &\; Y \in S^\perp \ot \conj{S}^\perp, \notag
    \\ &\; \rot(Y) \in \cone, \notag
    \\ &\; Y \succeq \ket{\Phi}\bra{\Phi}.
    \label{eq:qthp_C_dual}
\end{align}
That these two programs take the same value is shown in \cref{sec:duality}.
The point $\rho=I/\dim(A)$, $T=0, L=0$ is feasible for~\eqref{eq:qthp_C_primal},
giving $\qthpperp{\cone}(S) \ge 1$.
Although~\eqref{eq:qthm_C_dual} is always feasible,
in some cases~\eqref{eq:qthp_C_dual} is not feasible so
$\qthpperp{\cone}(S)$ can be infinite; see \cref{thm:qthm_qthp_complete} for an example.

Similar to~\eqref{eq:qthm_chain}, we have the chain of inequalities
\begin{align}
    \label{eq:qthp_chain}
    \qthperp(S) \le
    \qthpperp{\psdcone}(S) \le
    \qthpperp{\psdcone \cap \ppt}(S) \le
    \qthpperp{\sep}(S) \le
    \chromc(S).
\end{align}
The last inequality will be proved in \cref{thm:qthm_qthp_sandwich}, and the others follow from
the fact that~\eqref{eq:qthp_C_dual} has nondecreasing value as constraints are tightened.
Note, however, that the last two values can be $\infty$ and, unlike $\qthperp(S)$,
don't provide a bound on $\chromq(S)^2$.
As was the case with our Schrijver generalization, this generalized Szegedy quantity
matches the classical value when $S$ derives from a classical graph.

\begin{theorem}
    \label{thm:qthp_classical}
    Let $G$ be a classical loop-free graph and $S = \linspan\{ \ket{i}\bra{j} : i \sim j \}$.
    Then for any closed convex cone $\cone$ satisfying $\sep \subseteq \cone \subseteq \sep^*$,
    it holds that $\qthpperp{\cone}(S) = \thpbar(G)$.
\end{theorem}
\begin{proof}
    Define the isometry $V = \sum_i \ket{ii} \bra{i}$.
    Let $Z$ be an optimal solution for~\eqref{eq:szegedy_min}.
    Define $Y=VZV^\dag$.
    We have $Z \succeq J \implies Y \succeq VJV^\dag = \proj{\Phi}$.
    Since $Z \in S^\perp$ we have $Y=\sum_{i \not\sim j} Z_{ij} \ket{ii}\bra{jj}$; this
    is an element of $S^\perp \ot \conj{S}^\perp$.
    $Z$ being entrywise nonnegative ensures that
    $\rot(Y)=\sum_{ij} Z_{ij} \proj{i} \ot \proj{j} \in \sep$.
    So $Y$ is feasible for~\eqref{eq:qthp_C_dual} for $\qthpperp{\sep}(S)$.
    Its value is $\opnorm{\Tr_A Y} = \opnorm{\sum_i Z_{ii} \proj{i}} = \thpbar(G)$.
    Therefore $\qthpperp{\sep}(S) \le \thpbar(G)$.

    Now let $B$, $L'$ be an optimal solution for~\eqref{eq:szegedy_max} for $\thpbar(G)$.
    Without loss of generality, assume that $L'$ is Hermitian (any feasible solution
    for~\eqref{eq:szegedy_max} can be averaged with its adjoint).
    Also, assume that $L'$ vanishes on the diagonal since zeroing the diagonal entries
    of $L'$ doesn't affect feasibility for~\eqref{eq:szegedy_max}.
    Decompose $B$ into diagonal and off-diagonal components: $B = \rho + T'$.
    Define $T = VT'V^\dag$ and $L = VL'V^\dag/2$.
    We will show this to be feasible for~\eqref{eq:qthp_C_primal} for $\qthpperp{\sep^*}(S)$.
    By the arguments of \cref{thm:qthm_classical}, $\rho \succeq 0$, $\Tr \rho = 1$,
    and $I \ot \rho + T \succeq 0$.
    For $i \not\sim j$ we have $(T'+L')_{ij} = 0$ since $B+L' \in S+S_0$ and $T'+L'$ vanishes
    on the diagonal.
    So $T+L+L^\dag = V(T'+L')V^\dag = \sum_{i \sim j} (T'+L')_{ij} \ket{ii}\bra{jj}$,
    which is an element of $S \ot \conj{S}$.
    We have
    \begin{align*}
        \rot(L)
        &= \sum_{ij} \frac{1}{2} L'_{ij} \rot(\ket{ii}\bra{jj})
        \\ &= \sum_{ij} \frac{1}{2} L'_{ij} \proj{ij} \in \sep,
    \end{align*}
    where the last line relies on $L'_{ij} \ge 0$.
    Similarly $\rot(L)^\dag \in \sep$; therefore
    $\rot(L) + \rot(L)^\dag \in \sep = \sep^{**}$.
    So $\rho$, $T$, and $L$ are feasible for~\eqref{eq:qthp_C_primal} for
    $\qthpperp{\sep^*}(S)$.
    By the arguments of \cref{thm:qthm_classical}, the value of this solution
    is $\thpbar(G)$; therefore $\qthpperp{\sep^*}(S) \ge \thpbar(G)$.

    Clearly $\sep \subseteq \cone \subseteq \sep^* \implies
    \qthpperp{\sep}(S) \ge \qthpperp{\cone}(S) \ge \qthpperp{\sep^*}(S)$ since maximization
    programs have nonincreasing values as constraints are tightened.
    Combining this with the above two inequalities gives the desired equality result.
\end{proof}

\begin{theorem}
    \label{thm:qthm_qthp_mon}
    Suppose a closed convex cone $\cone$ is closed under the action
    of maps of the form $\chan{E} \ot \chanconj{E}$
    where $\chan{E}$ is a completely positive trace preserving map and
    $\chanconj{E}$ is the entrywise complex conjugate of $\chan{E}$.\footnote{
        Note that $(\chan{E} \ot \chanconj{E})(X)$ can be on a different Hilbert space than $X$.
        So, technically, one must consider a collection of cones, one for each Hilbert
        space.  For example, $\sep$ is such a collection.
    }
    In particular, the cones
    $\{ \sep, \psdcone, \ppt, \psdcone \cap \ppt, \sep^* \}$
    satisfy this requirement.
    Then $\qthmperp{\cone}$ and $\qthpperp{\cone}$ are homomorphism monotones in the
    sense that for trace-free non-commutative graphs $S$ and $T$ we have
    \begin{align}
        S \homm T \implies
           &\qthmperp{\cone}(S) \le \qthmperp{\cone}(T), \label{eq:qthm_mon}
        \\ &\qthpperp{\cone}(S) \le \qthpperp{\cone}(T). \label{eq:qthp_mon}
    \end{align}
\end{theorem}
\begin{proof}
    The proof is similar to that of \cref{thm:qthmon}, so we only describe the needed
    modifications.
    To prove~\eqref{eq:qthm_mon}, let $Y',L' \subseteq \linop{B} \ot \linop{B'}$ be a feasible
    solution for~\eqref{eq:qthm_C_dual} for $\qthmperp{\cone}(T)$.
    As was done in \cref{thm:qthmon}, define $Y \subseteq \linop{A} \ot \linop{A'}$ as
    $Y = \sum_{ij} (E_i \ot \conj{E}_i)^\dag Y' (E_j \ot \conj{E}_j)$ where
    the Kraus operators $\{E_i\}$ are a homomorphism $S \homm T$.
    Similarly, define
    $L = \sum_{ij} (E_i \ot \conj{E}_i)^\dag L' (E_j \ot \conj{E}_j)$.
    We will show this to be a feasible solution for~\eqref{eq:qthm_C_dual} for
    $\qthmperp{\cone}(S)$ with value at most $\qthmperp{\cone}(T)$.
    The arguments in the proof of \cref{thm:qthmon} apply directly to show
    $Y \succeq \proj{\Phi}$ and $\opnorm{\Tr_A Y} \le \opnorm{\Tr_B Y'}$.

    Since $Y',L'$ are feasible for~\eqref{eq:qthm_C_dual} for $\qthmperp{\cone}(T)$
    we have that
    $Y'+L'+L'^\dag \in T^\perp \ot \linop{B'} + \linop{B} \ot \conj{T}^\perp$, giving
    \begin{align*}
        Y+L+L^\dag &= \sum_{ij} (E_i \ot \conj{E}_i)^\dag (Y'+L'+L'^\dag) (E_j \ot \conj{E}_j)
        \\ &\in E^\dag T^\perp E \ot \conj{E}^\dag \linop{B'} \conj{E}
        \\ &\qquad + E^\dag \linop{B} E \ot \conj{E}^\dag \conj{T}^\perp \conj{E}
        \\ &\subseteq S^\perp \ot \linop{A'} + \linop{A} \ot \conj{S}^\perp.
    \end{align*}
    All that remains is to show $\rot(L) + \rot(L)^\dag \in \cone^*$.
    We have
    \begin{align}
        \rot(L) &= \sum_{ij} \rot( (E_i \ot \conj{E}_i)^\dag L' (E_j \ot \conj{E}_j) ) \notag
        \\ &= \sum_{ij} (E_i \ot \conj{E}_j)^\dag \rot(L') (E_i \ot \conj{E}_j) \notag
        \\ &= (\chanadj{E} \ot \chanadjconj{E})(\rot(L')).
        \label{eq:ERL}
    \end{align}
    Since completely positive maps commute with the taking of adjoints we also have
    $\rot(L)^\dag = (\chanadj{E} \ot \chanadjconj{E})(\rot(L')^\dag)$.
    Consequently,
    $\rot(L) + \rot(L)^\dag = (\chanadj{E} \ot \chanadjconj{E})(\rot(L') + \rot(L')^\dag)$.
    But $\rot(L')+\rot(L')^\dag \in \cone$ and this cone is assumed to be closed under such
    product maps, so $\rot(L)+\rot(L)^\dag \in \cone$.

    To prove~\eqref{eq:qthp_mon}, let $Y' \subseteq \linop{B} \ot \linop{B'}$ be a feasible
    solution for~\eqref{eq:qthp_C_dual} for $\qthpperp{\cone}(T)$ and define $Y$ as in the
    previous paragraph.
    We will show this to be a feasible solution for~\eqref{eq:qthp_C_dual} for
    $\qthmperp{\cone}(S)$ with value at most $\qthmperp{\cone}(T)$.
    Again the arguments in the proof of \cref{thm:qthmon} apply directly to show
    $Y \succeq \proj{\Phi}$ and $\opnorm{\Tr_A Y} \le \opnorm{\Tr_B Y'}$.
    A straightforward modification of~\eqref{eq:JJYJJ_in_Sperp} yields
    $Y \in S^\perp \ot \conj{S}^\perp$.
    All that remains is to show that $\rot(Y) \in \cone$.
    Similar to~\eqref{eq:ERL}, we have
    $\rot(Y) = (\chanadj{E} \ot \chanadjconj{E})(\rot(Y'))$.
    But $\rot(Y') \in \cone$ and this cone is assumed to be closed under such
    product maps, so $\rot(Y)^\dag \in \cone$.
\end{proof}

\begin{lemma}
    Let $\cone$ be a closed convex cone. Then,
    \begin{enumerate}
        \item $\qthmperp{\cone}(K_n) = \qthpperp{\cone}(K_n) = n$ if $\cone \supseteq \sep$
        \item $\qthmperp{\cone}(Q_n) = n^2$ if $\cone \supseteq \sep$
        \item $\qthpperp{\cone}(Q_n) = n^2$ if $\proj{\Phi} \in \cone$ (e.g.\ if
            $\cone \supseteq \psdcone$)
        \item $\qthpperp{\cone}(Q_n) = \infty$ if $\proj{\Phi} \not\in \cone$
            (e.g.\ if $\cone \subseteq \ppt$)
    \end{enumerate}
    \label{thm:qthm_qthp_complete}
\end{lemma}
\begin{proof}
    For $\cone \supseteq \sep$ we have
    $\qthmperp{\sep}(K_n) \le \qthmperp{\cone}(K_n) \le \qthperp(K_n)
    \le \qthpperp{\cone}(K_n) \le \qthpperp{\sep}(K_n)$.
    By \cref{thm:qthm_classical,thm:qthp_classical}
    $\qthmperp{\sep}(K_n) = \qthpperp{\sep}(K_n) = n$,
    since
    $\thmbar(K_n)=\thpbar(K_n)=n$.

    A feasible solution for~\eqref{eq:qthm_C_primal} for $\qthmperp{\sep}(Q_n)$
    is given by $\rho=I/n$ and $T = \proj{\Phi} - I \ot I/n$.
    The operator $\rot(T) = I \ot I - \proj{\Phi}/n$ is separable~\cite{PhysRevA.66.062311}.
    The value of this solution is $n^2$, so $\qthmperp{\sep}(Q_n) \ge n^2$.
    For $\cone \supseteq \sep$ we have
    $\qthmperp{\sep}(Q_n) \le \qthmperp{\cone}(Q_n) \le \qthperp(Q_n)=n^2$,
    so in fact $\qthmperp{\cone}(Q_n) = n^2$.

    Suppose $\proj{\Phi} \in \cone$.
    Then $\rot(I \ot I) \in \cone$ so
    a feasible solution for~\eqref{eq:qthp_C_dual} for $\qthpperp{\cone}(Q_n)$
    with value $n^2$ is given by $Y = n I \ot I$; therefore $\qthpperp{\cone}(Q_n) \le n^2$.
    But also $\qthpperp{\cone}(Q_n) \ge \qthperp(Q_n) = n^2$, so in
    fact $\qthpperp{\cone}(Q_n) = n^2$.

    Suppose $\proj{\Phi} \not\in \cone$.
    Any feasible solution for~\eqref{eq:qthp_C_dual} for $\qthpperp{\cone}(Q_n)$
    requires $Y \in Q_n^\perp \ot \conj{Q}_n^\perp = \linspan\{ I \ot I \}$.
    In other words, $Y = cI \ot I$ for some $c > 0$.
    But then $\rot(Y) = c \proj{\Phi}$.
    Since $\proj{\Phi} \not\in \cone$, there can be no feasible solution.
    So $\qthpperp{\cone}(Q_n) = \infty$.
\end{proof}

\begin{corollary}
    \label{thm:qthm_qthp_sandwich}
    Let $S$ be a trace-free non-commutative graph.
    For $\cone \in \{ \sep, \psdcone, \ppt, \psdcone \cap \ppt, \sep^* \}$,
    it holds that
        $\cliquec(S) \le
        \qthmperp{\cone}(S) \le
        \qthperp(S) \le
        \qthpperp{\cone}(S) \le
        \chromc(S)$ and
        $[\cliqueq(S)]^2 \le \qthmperp{\cone}(S)$.
        For $\cone \in \{\psdcone, \sep^*\}$,
        $\qthpperp{\cone}(S) \le [\chromq(S)]^2$.
\end{corollary}
\begin{proof}
    The corollary follows from application of \cref{thm:qthm_qthp_mon} to the definition
    of $\cliquec(S)$, $\cliqueq(S)$, $\chromc(S)$, and $\chromq(S)$, and using the values
    from \cref{thm:qthm_qthp_complete}.
    Note that for $\cone \subseteq \ppt$, in particular, the bound
    $\qthpperp{\cone}(S) \le [\chromq(S)]^2$ does not hold since
    $\chromq(Q_n) = n$ but $\qthpperp{\cone}(Q_n) = \infty$.
\end{proof}

Having developed the basic theory of Schrijver and Szegedy numbers for non-commutative graphs,
we turn now to commentary and applications.
It is interesting to note that a gap between $\thmbar$, $\thbar$, and $\thpbar$ for classical
graphs is somewhat difficult to find and the gaps are often small.
The smallest classical graph displaying a gap between any of these three quantities has 8
vertices.\footnote{
    Verified numerically.  The graph with graph6 code ``\texttt{GRddY\{}'' has
    $\thmbar=3.236$, $\thbar=3.302$, $\thpbar=3.338$.
}
The gap is much more pronounced for non-commutative graphs, showing up already for
graphs in $\linop{\mathbb{C}^2}$.
Indeed, by \cref{thm:qthm_qthp_complete}, $\qthperp(Q_2)=4$ but $\qthpperp{\ppt}(Q_2)=\infty$.
Numerical results on 10000 random graphs $S \in \linop{\mathbb{C}^3}$ with $\dim(S)=4$
yielded $\qthmperp{\ppt}(S)=1$ for all test cases and $\qthpperp{\ppt}(S)=\infty$ for 93\%
of test cases (with the solver failing to converge in one case).

An extreme gap between $\qthperp$ and $\qthmperp{\ppt}$
appears for $S=\mathbb{C}\Delta$ with
$\Delta = \textrm{diag}\{d-1,-1,\dots,-1\} \subseteq \linop{\mathbb{C}^d}$.
In this case,
$\qthperp(S)=d$~\cite{arxiv:1002.2514},
but $\qthmperp{\ppt}(S)=1$.
This can be seen as follows.
For $\qthperp(S)$, the feasible solution $T = \Delta \ot \proj{0}$, $\rho = \proj{0}$
allows $\qthperp(S)=d$.
For $\qthmperp{\ppt}(S)$ it is required first of all that $T \in S \ot \conj{S}$.
The only feasible solutions are then of the form $T = c \Delta \ot \conj{\Delta}$
for some constant $c$.
But $\rot(c \Delta \ot \conj{\Delta}) \in \ppt$ requires $c=0$.
Therefore the only feasible solution
for $\qthmperp{\ppt}(S)$ is $T=0$, giving $\qthmperp{\ppt}(S)=1$.
So in this case $\qthmperp{\ppt}(S)=1$ exactly matches the clique number
$\cliquec(S)$, since $1 \le \cliquec(S) \le \qthmperp{\ppt}(S)=1$.

Note, however, that the entanglement assisted clique number of $S=\mathbb{C}\Delta$ is
$\cliquece(S)=2$~\cite{arxiv:1002.2514}.
So, in this case, $\qthmperp{\ppt}(S)$ is \emph{not} an upper bound on one-shot
entanglement assisted capacity.
This is a bit of a surprise, since for classical graphs and for any cone
$\sep \subseteq \cone \subseteq \sep^*$
our $\qthmperp{\cone}$ and $\qthpperp{\cone}$ reduce to
$\thmbar$ and $\thpbar$ (by \cref{thm:qthm_classical,thm:qthp_classical}),
and these are known to be monotone under entanglement assisted
homomorphisms~\cite{6880319}.
In particular, for classical graphs, $\thmbar(G)$ upper bounds
one-shot entanglement assisted capacity.

The failure of $\qthmperp{\ppt}(S)$ to bound entanglement assisted one-shot capacity
$\cliquece$ can be understood as follows.
This capacity is the largest $n$ such that $K_n \home S$.
By \cref{def:homq_star} this means there is some $\Lambda \succ 0$ such that
$K_n \ot \Lambda \homm S$.
By \cref{thm:qthmon} we have $\qthperp(K_n \ot \Lambda) \le \qthperp(S)$
and by \cref{thm:sperp_Lambda} $\qthperp(K_n \ot \Lambda) = n$, so $n \le \qthperp(S)$.
Thus $\cliquece(S) \le \qthperp(S)$.
It is this last step that breaks down for $\qthmperp{\ppt}$.
By \cref{thm:qthm_qthp_mon} we have
$\qthmperp{\ppt}(K_n \ot \Lambda) \le \qthmperp{\ppt}(S)$.
But, as we will show in \cref{thm:qthm_ppt_Phi},
$\qthmperp{\ppt}(K_n \ot \Lambda) = 1$, so this is a trivial bound
that says nothing about $n$.

Although $\cone=\ppt$ is therefore unsuitable for bounding entanglement assisted clique
number, all is not lost.
In \cref{thm:qthm_psd_lambda} we will show
$\qthmperp{\psdcone}(S \ot I) = \qthmperp{\psdcone}(S)$.
So $\qthmperp{\psdcone}$ indeed provides a bound on entanglement assisted one-shot
capacity, when sender and receiver share a maximally entangled state (i.e.\ $\Lambda=I$).
For general $\Lambda$ this does not hold: $\qthmperp{\psdcone}(S \ot \Lambda)$ can
be smaller than $\qthmperp{\psdcone}(S)$.

\begin{lemma}
    \label{thm:qthm_ppt_Phi}
    Let $S$ be a trace-free non-commutative graph and $\Lambda \succeq 0$ with
    $\rank(\Lambda) > 1$.
    Then $\qthmperp{\ppt}(S \ot \Lambda) = 1$.
\end{lemma}
\begin{proof}
    We will show that the only possible feasible solutions for~\eqref{eq:qthm_C_primal}
    are those with $T=0$.
    Indeed, suppose that $T \ne 0$.
    It is required that $T \in (S \ot \Lambda) \ot (\conj{S} \ot \conj{\Lambda})$,
    so $T$ must be of the form $T=T' \ot \Lambda \ot \conj{\Lambda}$ where
    $T' \in S \ot \conj{S}$.
    Then $\rot(T) = \rot(T') \ot \rot(\Lambda \ot \conj{\Lambda}) \in \ppt$ requires that
    $\rot(\Lambda \ot \conj{\Lambda}) \in \ppt$.
    But $\rot(\Lambda \ot \conj{\Lambda}) = \proj{\psi}$,
    where $\ket{\psi} = \sum_{ij} \Lambda_{ij} \ket{ij}$, is an entangled state since
    $\rank(\Lambda) > 1$.
    Entangled pure states are not in $\ppt$.
\end{proof}

\begin{lemma}
    \label{thm:qthm_psd_lambda}
    Let $S$ be a trace-free non-commutative graph and let $\Lambda \succeq 0$,
    $\Lambda \ne 0$.  Then
    \begin{align}
        \label{eq:schrij_psd_lambda_eq}
        \frac{\qthmperp{\psdcone}(S)-1}{\qthmperp{\psdcone}(S \ot \Lambda)-1} =
            \frac{\opnorm{\Lambda} \Tr(\Lambda)}{\Tr(\Lambda^2)}.
    \end{align}
    In particular, $\qthmperp{\psdcone}(S \ot I) = \qthmperp{\psdcone}(S)$.
\end{lemma}
\begin{proof}
    Work in a basis in which $\Lambda$ is diagonal:
    $\Lambda=\diag(\lambda_1, \dots, \lambda_n)$ with
    $\opnorm{\Lambda} = \lambda_1 \ge \lambda_2 \ge \dots \ge \lambda_n \ge 0$.

    ($\ge$):
    Say $S \subseteq \linop{A}$ and $\Lambda \in \linop{B}$.
    Let $T \in \linop{A\ot B\ot A'\ot B'}$
    and $\rho \in \linop{A'\ot B'}$ be an optimal solution for~\eqref{eq:qthm_C_primal}
    for $\qthmperp{\psdcone}(S \ot \Lambda)$.
    Since $T \in (S \ot \Lambda) \ot (\conj{S} \ot \conj{\Lambda})$
    it must be that $T=T' \ot (\Lambda \ot \conj{\Lambda})$ for some $T' \in S \ot \conj{S}$.
    So $T$ is block diagonal:
    \begin{align*}
        T = \sum_{ij} \lambda_i \lambda_j T' \ot \proj{i}_B \ot \proj{j}_{B'}.
    \end{align*}
    Without loss of generality $\rho$ is also block diagonal:
    $\rho = \sum_j \rho_j \ot \proj{j}_{B'}$.
    This can be assumed since the off diagonal components of $\rho$ can be zeroed out
    without affecting its trace or the relation $I_{AB} \ot \rho + T \succeq 0$.
    Since $I_{AB} \ot \rho + T$ is block diagonal and positive semidefinite, each block
    must be positive semidefinite:
    $I_A \ot \rho_j + \lambda_i \lambda_j T' \succeq 0$ or, equivalently,
    \begin{align*}
        I_A \ot \frac{\rho_j}{\lambda_i \lambda_j} + T' \succeq 0.
    \end{align*}
    Let $\sigma$ be the member of $\{ \rho_j / \lambda_1 \lambda_j \}_j$ with the least
    trace.
    We have
    \begin{align*}
        \Tr(\Lambda) \Tr(\sigma)
        = \sum_j \lambda_j \Tr(\sigma)
        \le \sum_j \Tr(\rho_j)/\lambda_1
        = \Tr(\rho) / \opnorm{\Lambda}.
    \end{align*}
    But $\Tr(\rho)=1$ so $c := \Tr(\sigma)^{-1} \ge \Tr(\Lambda) \opnorm{\Lambda}$.
    We have $\Tr(c \sigma)=1$ and
    $I_A \ot \sigma + T' \succeq 0 \implies I_A \ot c\sigma + cT' \succeq 0$.
    Also
    \begin{align}
        \rot(&T) \succeq 0
        \notag \\ &\implies \rot(T') \ot \rot(\sum_{ij} \lambda_i \lambda_j \proj{i}_B \ot
            \proj{j}_{B'}) \succeq 0
        \notag \\ &\implies \rot(T') \ot (\sum_i \lambda_i \ket{ii}_{BB'})
            (\sum_j \lambda_j \bra{jj}_{BB'}) \succeq 0
        \notag \\ &\implies \rot(T') \succeq 0.
        \label{eq:rot_Tprime_pos}
    \end{align}
    So $c\sigma$ and $cT'$ are feasible for~\eqref{eq:qthm_C_primal} for
    $\qthmperp{\psdcone}(S)$ with value
    \begin{align}
        \braopketnolr{\Phi_A}{I_A &\ot c\sigma + cT'}{\Phi_A}
        \notag \\ &= \Tr(c\sigma) + c \braopket{\Phi_A}{T'}{\Phi_A}
        \notag \\ &= 1 + \frac{c}{\Tr(\Lambda^2)} \braopket{\Phi_A}{T'}{\Phi_A}
            \braopket{\Phi_B}{\Lambda \ot \conj{\Lambda}}{\Phi_B}
        \notag \\ &= 1 + \frac{c}{\Tr(\Lambda^2)} \braopket{\Phi_{AB}}{T}{\Phi_{AB}}
        \notag \\ &\ge 1 + \frac{\Tr(\Lambda) \opnorm{\Lambda}}{\Tr(\Lambda^2)}
            (\qthmperp{\psdcone}(S \ot \Lambda)-1).
        \label{eq:sch_psd_feas_from_lambda}
    \end{align}
    Therefore $\qthmperp{\psdcone}(S) \ge \eqref{eq:sch_psd_feas_from_lambda}$
    and the left side of~\eqref{eq:schrij_psd_lambda_eq} is at least as great as the right side.

    ($\le$):
    Let $\rho'$ and $T'$ be an optimal solution for~\eqref{eq:qthm_C_primal}
    for $\qthmperp{\psdcone}(S)$.
    Define $\rho = \rho' \ot \conj{\Lambda} / \Tr(\Lambda)$
    and $T = T' \ot \Lambda \ot \conj{\Lambda} / (\opnorm{\Lambda} \Tr(\Lambda))$.
    Then $\Tr(\rho)=1$, $T \in (S \ot \Lambda) \ot (\conj{S} \ot \conj{\Lambda})$, and
    \begin{align*}
        I_{AB} \ot \rho + T
        &= \frac{1}{\Tr(\Lambda)} \left(
                I_A \ot \rho' \ot I_B + T' \ot \frac{\Lambda}{\opnorm{\Lambda}}
            \right) \ot \conj{\Lambda}
        \\ &\succeq \frac{1}{\Tr(\Lambda)} \left(
                I_A \ot \rho' \ot \frac{\Lambda}{\opnorm{\Lambda}} +
                T' \ot \frac{\Lambda}{\opnorm{\Lambda}}
            \right) \ot \conj{\Lambda}
        \\ &= \frac{1}{\Tr(\Lambda) \opnorm{\Lambda}} \left(
                I_A \ot \rho' + T'
            \right) \ot \Lambda \ot \conj{\Lambda}
        \succeq 0.
    \end{align*}
    And $\rot(T) \succeq 0$ by following the logic of~\eqref{eq:rot_Tprime_pos} in
    reverse.
    So $\rho$ and $T$ are feasible for~\eqref{eq:qthm_C_primal}
    for $\qthmperp{\psdcone}(S \ot \Lambda)$.
    The objective value is
    \begin{align}
        \braopketnolr{\Phi_{AB}}{I_{AB} &\ot \rho + T}{\Phi_{AB}}
        \notag \\ &= 1 + \braopket{\Phi_{AB}}{T}{\Phi_{AB}}
        \notag \\ &= 1 + \frac{\braopket{\Phi_A}{T'}{\Phi_A}
            \braopket{\Phi_B}{\Lambda \ot \conj{\Lambda}}{\Phi_B}}
            {\opnorm{\Lambda} \Tr(\Lambda)}
        \notag \\ &= 1 + (\qthmperp{\psdcone}(S)-1)
            \frac{\Tr(\Lambda^2)}{\opnorm{\Lambda} \Tr(\Lambda)}.
        \label{eq:sch_psd_lam_feas_from_S}
    \end{align}
    So $\qthmperp{\psdcone}(S \ot \Lambda) \ge \eqref{eq:sch_psd_lam_feas_from_S}$
    and the left side of~\eqref{eq:schrij_psd_lambda_eq} is no greater than the right side.
\end{proof}

An extreme example of the difference between unassisted capacity and entanglement assisted
capacity is given in theorem 3 of~\cite{arxiv:0906.2527}:
a channel is defined having distinguishability graph $S = Q_n \ot I_2$, where
$I_2$ is the $2 \times 2$ identity operator.
In~\cite{arxiv:0906.2527} it is shown that this channel has no unassisted zero-error
classical capacity (even with many uses of the channel) but has one-shot entanglement
assisted quantum capacity $\log n$.  In other words, $\cliqueqe(S) = n$.
Our techniques show this result to be ``obvious in retrospect''.
Indeed, trivially $Q_n \home S$ since $Q_n \ot I_2 \homm Q_n \ot I_2$.
So $\cliqueqe(S) \ge n$; the channel has one-shot entanglement assisted capacity of
at least $\log n$ qubits.
And by \cref{thm:qthm_ppt_Phi}, $\qthmperp{\ppt}(S)=1$ so $\cliquec(S)=1$; the channel has
no one-shot capacity in the absence of entanglement.
Unfortunately we cannot use these techniques to bound the asymptotic capacity
$\lim_{m \to \infty} \frac{1}{m} \log \cliquec(S^{*m})$
since $\qthmperp{\ppt}$ is not in general multiplicative under powers $S^{\djp m}$
(even for classical graphs~\cite{6880319}).
We conjecture, however, that $\qthmperp{\cone}$ (for certain cones $\cone$)
is multiplicative when $\qthmperp{\cone}(S)=1$.

Inspired by this $S = Q_n \ot I_2$ example, we construct a channel that has no one-shot
capacity when assisted by a maximally entangled state of arbitrary dimension, but does
have one-shot capacity when assisted by a non-maximally entangled state.
To our knowledge this is a new result.
We note that the possibility of such behavior for a classical channel is still an open
problem~\cite{arxiv:1212.1724,6880319}.
This example nicely illustrates the utility of these semidefinite programming bounds
which, at least for small dimensions, are very computationally tractable.
The following example was found and
verified numerically before \cref{thm:qthm_psd_lambda} was discovered; the latter was
inspired by the former.

\begin{theorem}
    \label{thm:need_non_mes}
    There is a channel that can transmit an error-free quantum state of dimension $n$
    (i.e.\ $\log n$ qubits) using entanglement between sender and receiver, but that
    cannot transmit even a single error-free classical bit if the sender and receiver only
    share a maximally entangled state.
\end{theorem}
\begin{proof}
    Let $T = Q_n \ot \Lambda$ where $\Lambda$ satisfies
    $c := \frac{\opnorm{\Lambda} \Tr(\Lambda)}{\Tr(\Lambda^2)} > n^2-1$.
    For instance, take
    $\Lambda = \diag(1,\alpha,\dots,\alpha) \in \linop{\mathbb{C}^m}$ where
    $\alpha = (\sqrt{m}-1)/(m-1)$.  This maximizes $c$ for a given $m$, achieving
    $c=(m-1)/2(\sqrt{m}-1)$.
    So if $n=2$ we can take $m=26$ to get $c > 3$.

    By lemma 2 of~\cite{arxiv:0906.2527}, $T$ is the distinguishability graph of some
    quantum channel.
    $Q_n \ot \Lambda \homm T$ (there is always a homomorphism from a graph to itself), so
    a quantum state of dimension $n$ can be sent using an entanglement resource
    $\ket{\lambda}$ with reduced density operator $\Lambda$.
    In fact, the encoding is trivial: Alice simply puts her state to be transmitted, along
    with her half of the entanglement resource, directly into the channel.

    On the other hand, by \cref{thm:qthm_psd_lambda},
    $\qthmperp{\psdcone}(K_2 \ot I) = \qthmperp{\psdcone}(K_2) = 2$
    (with $I$ being identity on a space of arbitrary finite dimension) whereas
    $\qthmperp{\psdcone}(T) = 1 + (\qthmperp{\psdcone}(Q_n) - 1)/c = 1+(n^2-1)/c < 2$.
    Since $\qthmperp{\psdcone}$ is a homomorphism monotone, $K_2 \ot I \not\homm T$;
    it is not possible to transmit an error-free classical bit using a maximally entangled
    resource.
\end{proof}

As mentioned above, we conjecture that $\qthmperp{\cone}$ (for certain cones $\cone$)
is multiplicative when $\qthmperp{\cone}(S)=1$.
If this were the case, then $\qthmperp{\cone}(S)=1$ would be enough to guarantee
that a channel has no zero-error asymptotic capacity without entanglement assistance.
We might as well focus on $\cone=\sep$ since this is the smallest of the cones we have
considered, and so gives the strongest bound.
When is $\qthmperp{\sep}(S) = 1$?
Below we present a characterization, but leave the interpretation open.

\begin{theorem}
    Let $S$ be a trace-free non-commutative graph.
    $\qthmperp{\sep}(S) = 1$ iff there is an $M \in (S \ot \conj{S})^\perp$
    such that $\rot(M)-I \in \sep^*$.
    Such channels have no unassisted one-shot capacity.
\end{theorem}
\begin{proof}
    ($\implies$):
    Let $S \subseteq \linop{A}$ be a trace-free non-commutative graph with
    $\qthmperp{\cone}(S) = 1$
    Let $Y, L$ be an optimal solution for~\eqref{eq:qthm_C_dual} for
    $\qthmperp{\cone}(S)$.
    We have
    \begin{align*}
        \opnorm{\Tr_A Y} = \qthmperp{\cone}(S) = 1
        &\implies \Tr_A Y \preceq I = \Tr_A(\proj{\Phi})
        \\ &\implies \Tr_A (Y - \proj{\Phi}) \preceq 0
        \\ &\implies \Tr (Y - \proj{\Phi}) \le 0
    \end{align*}
    But $Y - \proj{\Phi} \succeq 0$ so in fact $Y = \proj{\Phi}$.

    Notice that $Y = \proj{\Phi}$ is symmetric under $\dag$ and
    $\ddag$ (i.e.\ $Y = Y^\dag = Y^\ddag$).
    The subspace $(S \ot \conj{S})^\perp$ is also symmetric under these operations,
    as is the cone $\sep^*$.
    So we can assume without loss of generality that $L$ is invariant under
    $\dag$ and $\ddag$.  Indeed, any general $L$ could be replaced with
    $(L + L^\dag + L^\ddag + L^{\dag\ddag})/4$.
    Then $Y+2L \in (S \ot \conj{S})^\perp$ and $\rot(L) \in \sep^*$.
    Define $M = Y+2L$.  Then $M \in (S \ot \conj{S})^\perp$ and
    $\rot(M) - I = \rot(\proj{\Phi}) + 2\rot(L) - I = I + 2 \rot(L) - I =
    2 \rot(L) \in \sep^*$.

    ($\impliedby$):
    Suppose $M \in (S \ot \conj{S})^\perp$ and $\rot(M)-I \in \sep^*$.
    By the same logic as the first part of the proof, we can assume that $M$
    is invariant under $\dag$ and $\ddag$, so that $M=M^\dag$ and
    $\rot(M) = \rot(M)^\dag$.
    Define $Y=\proj{\Phi}$ and $L=(M-Y)/2$.
    Then $Y+L+L^\dag = M \in (S \ot \conj{S})^\perp$ and
    $\rot(L)+\rot(L)^\dag = \rot(M) - \rot(Y) = \rot(M) - I \in \sep^*$, so
    this is a feasible solution for~\eqref{eq:qthm_C_dual} for $\qthmperp{\cone}(S)$.
    Its value is $\opnorm{\Tr_A Y} = \opnorm{I_{A'}} = 1$, so
    $\qthmperp{\cone}(S) \le 1$.
    But any feasible solution has $Y \ge \proj{\Phi}$ and so must have value at least
    $\opnorm{\Tr_A\{\proj{\Phi}\}} = 1$.
    Therefore also $\qthmperp{\cone}(S) \ge 1$.
\end{proof}

We now turn our attention to $\qthpperp{\cone}$.
Whereas $\qthmperp{\cone}(S)=1$, for any cone $\cone \supseteq \sep$,
certifies that a channel has no one-shot capacity (without entanglement assistance),
$\qthpperp{\cone}(S)=\infty$ certifies that a source cannot be transmitted using local
operations and one-way classical communication (LOCC-1).
This is because
\begin{align*}
    \cone \supseteq \sep \implies \qthpperp{\cone}(S) \le \qthpperp{\sep}(S) \le \chi(S).
\end{align*}
So if $\qthpperp{\cone}(S)=\infty$ then $\chi(S)=\infty$ and no amount of classical
communication from Alice to Bob can transmit the source.

As an example,~\cite{PhysRevA.88.062316} provides a set of three maximally entangled states
that are LOCC-1 indistinguishable:
\begin{align*}
    \ket{\psi_0} &= \frac{1}{2} (\ket{00} + \ket{11})_{A_1B_1} \ot (\ket{00} + \ket{11})_{A_2B_2}
    \\
    \ket{\psi_1} &= \frac{1}{2} (\omega \ket{00} + \ket{11})_{A_1B_1} \ot (\ket{01} + \ket{10})_{A_2B_2}
    \\
    \ket{\psi_2} &= \frac{1}{2} (\gamma \ket{00} + \ket{11})_{A_1B_1} \ot (\ket{00} - \ket{11})_{A_2B_2}
\end{align*}
where $\omega$ and $\gamma$ are phases in general position.
The characteristic graph for this source is $\linspan\{ I, Z \} \ot Q_2$.
The quantity $\qthpperp{\ppt}(S)$ is efficiently computable numerically (at least for
spaces this small), and immediately provides a certificate that these states are LOCC-1
indistinguishable, with no manual computation needed.
In the case of this example there is in fact an alternate proof of this result.
If the three states defined above were LOCC-1 distinguishable then there would be an $n$
such that $\linspan\{ I, Z \} \ot Q_2 \homm K_n$.  But then
\begin{align*}
    Q_2 \homm \diag(1,0) \ot Q_2 \homm \linspan\{ I, Z \} \ot Q_2 \homm K_n
\end{align*}
where the second follows from $\diag(1,0) \ot Q_2 \subseteq \linspan\{ I, Z \} \ot Q_2$.
By transitivity of homomorphisms this yields $Q_2 \homm K_n$.
But a qubit cannot be transmitted through a classical channel so $Q_2 \not\homm K_n$.

\section{Conclusion}
\label{sec:conclusion}

We have defined and investigated the problem of quantum zero-error source-channel coding.
This broad class of problems includes dense coding, teleportation, channel capacity,
and one-way LOCC state measurement.
Whereas classical zero-error source-channel coding relies on graphs, the quantum version
relies on non-commutative graphs.  Central to this theory is a generalization of the
notion of graph homomorphism to non-commutative graphs.

For classical graphs, it is known that the Lov{\'a}sz number is monotone under
homomorphisms (and in fact even entanglement assisted homomorphisms).
The Lov{\'a}sz number has been generalized to non-commutative graphs
by~\cite{dsw2013}; we showed this quantity to be monotone under entanglement assisted
homomorphisms on non-commutative graphs.

We investigated the problem of sending many parallel source instances using many parallel
channels (block coding) and found that the Lov{\'a}sz number provides a bound on the cost rate, but
only if the source satisfies a particular condition.  Classical sources, as well as
sources that can produce a maximally entangled state, both satisfy this condition.

We defined Schrijver and Szegedy quantities for non-commutative graphs.  These are
monotone under non-commutative graph homomorphisms, but not entanglement assisted
homomorphisms.
In fact, we derived a sequence of such quantities that are all equal to the traditional
Schrijver and Szegedy quantities for classical graphs but can take different values
on general non-commutative graphs.
These results were used to investigate some known examples from the literature regarding
entanglement assisted communication over a noisy channel and one-way LOCC measurements.
Strangely, one of the Schrijver variants, $\qthmperp{\psdcone}$, scores non-maximally
entangled states as more valuable a resource than maximally entangled states (which are
not even visible to $\qthmperp{\psdcone}$).
Exploiting this oddity we constructed a channel that can transmit several
zero-error qubits if sender and receiver can share an arbitrary entangled state, but
cannot transmit even a single classical bit if only a maximally entangled resource is
allowed.
It is still an open question whether such behavior is possible for a classical channel.

Most of all, and more importantly than any specific bounds provided for the quantum
source-channel coding problem, we have furthered the program of non-commutative graph
theory set forth in~\cite{dsw2013}.  It is a curiosity that a field as discrete as graph
theory can be ``quantized'' by replacing sets with Hilbert spaces and binary relations
with operator subspaces.  Non-commutative graphs offer the promise that some of the
wealth of graph theory may be imported into the theory of operator subspaces.
But actually this promise is more of a tease, as even the most basic facts from graph theory
lead only to (interesting!) open questions in the theory of non-commutative graphs.
We close by outlining some of these questions.

\begin{itemize}
    \item For classical graphs, $\chi(G) \omega(\Gc) \ge \abs{V(G)}$.
        Does this hold also for non-commutative graphs, with an appropriate definition of
        graph complement?  We propose the complement (for trace-free graphs)
        $S^c = (S+\mathbb{C}I)^\perp$, and conjecture that
        $\chromc(S) \cliquec(S^c) \ge n$ where $S \subseteq \linop{\mathbb{C}^n}$.
        Note that $\chromc(S)$ and $\cliquec(S^c)$ are only defined
        when $S$ and $S^c$ are both trace free.
        Similarly, does it hold that $\qthperp(S)\qthperp(S^c) \ge n$?

    \item What is the analogue of vertex transitive for non-commutative
        graphs, and what are the properties of these graphs?
        We propose to define the automorphism group as
        $\textrm{Aut}(S) = \{ U \textrm{ unitary} : USU^\dag = S \}$
        and to call such a group vertex transitive if the only operators
        satisfying $U \rho U^\dag = \rho$ for all $U \in \textrm{Aut}(S)$
        are those proportional to identity.

    \item A Hamiltonian path for a trace-free non-commutative graph $S \in \linop{\mathbb{C}^n}$
        can be taken to be a set of nonzero vectors such that $\ket{\psi_i}\bra{\psi_{i+1}} \in S$
        for $i \in \{1,\dots,n-1\}$.  Does the Lov{\'a}sz conjecture generalize?  That is
        to say, does every connected trace-free vertex transitive non-commutative graph
        have a Hamiltonian path?

    \item Let $S$ be a non-commutative graph associated with the classical graph $G$.
        We saw that $\chi(G)$ is the smallest $n$ such that $S \to K_n$ and
        orthogonal rank $\orthrank(G)$ is the smallest $n$ such that $S \to Q_n$.
        Projective rank $\xi_f$~\cite{arxiv:1212.1724} is to
        $\orthrank$ as fractional chromatic number $\chi_f$ is to $\chi$.
        Since $\chi_f(G) = \min\{ p/q : G \to K_{p:q} \}$ where $K_{p:q}$ is the Kneser
        graph~\cite{GodsilRoyle200105}, is it the case that
        $\xi_f(G) = \min\{ p/q : S \to K'_{p:q} \}$ for some class of non-commutative
        graphs $K'_{p:q}$?

    \item How is the distinguishability graph of a channel related to that of the
        complementary channel?  The same question can be asked for the source: swapping
        Alice and Bob's inputs defines a complementary source.

    \item For classical graphs, $\thmbar$ and $\thpbar$ are monotone under entanglement assisted
        homomorphisms.  For non-commutative graphs this does not always hold.  Is there some
        insight here?  Or does this mean there is some better generalization of
        $\thmbar$ and $\thpbar$?

    \item Is it the case that $\qthmperp{\cone}(S) = 1$ implies
        $\qthmperp{\cone}(S^{\djp n}) = 1$, for some suitable choice of $\cone$?
        If so, $\qthmperp{\cone}(S) = 1$ would certify that a channel had no
        asymptotic zero-error capacity.

    \item Any trace-free non-commutative graph is both the characteristic graph of some
        source and the distinguishability graph of some channel.
        Is there something to be learned from this relation between sources and channels?

    \item It is known that two channels with no one-shot capacity, when put in parallel,
        may have positive one-shot capacity~\cite{6094278,arxiv:0906.2527,6157069}.
        Is there a similar effect with sources?  Are there two sources that are both one-way
        LOCC (LOCC-1) indistinguishable but in parallel are LOCC-1 distinguishable?

    \item The quantity $\opnorm{\Lambda} \Tr(\Lambda) / \Tr(\Lambda^2)$, which shows up
        in \cref{thm:qthm_psd_lambda}, is only greater than $1$ for the reduced density
        operator of a non-maximally entangled state.  Is this an ad hoc quantity, or is it
        a meaningful measure of entanglement?
\end{itemize}

\section*{Acknowledgments}

The author would like to thank Simone Severini, David Roberson, and Vern Paulsen for many
helpful discussions.

This research received financial support from the National
Science Foundation through Grant PHY-1068331.

\appendices
\section{Duality Proofs}
\label[secinapp]{sec:duality}

We will derive the dual of~\eqref{eq:qthm_C_primal}, which we rewrite here for reference.
\begin{align}
    \qthmperp{\cone}(S) = \max &\; \braopket{\Phi}{ I \ot \rho + T }{\Phi} \notag
    \\ \textrm{s.t. } &\; \rho \succeq 0, \Tr \rho = 1, \notag
    \\ &\; I \ot \rho + T \succeq 0, \notag
    \\ &\; T \in S \ot \conj{S}, \notag
    \\ &\; \rot(T) \in \cone,
    \label{eq:qthm_C_appendix}
\end{align}
Section 4.7 of~\cite{GartnerMatousek201201} gives the following duality recipe for conic
programming over real vectors, where $\mathcal{G}$ and $\mathcal{H}$ are closed convex cones:
\begin{align}
    \notag
    \textrm{(Primal)} \quad \max &\ip{\cvec, \xvec}
        \\ \notag \textrm{ s.t. } &\bvec - A(\xvec) \in \mathcal{G},
        \\ &\xvec \in \mathcal{H}
    \label{eq:gartner_primal}
    \\ \notag
    \textrm{(Dual)} \quad \min &\ip{\bvec, \yvec}
        \\ \notag \textrm{ s.t. } &A^T(\yvec) - \cvec \in \mathcal{H}^*,
        \\ &\yvec \in \mathcal{G}^*.
    \label{eq:gartner_duality}
\end{align}
This nearly suffices for our purposes, since~\eqref{eq:qthm_C_appendix} can be viewed as a program
over real vectors by considering the real inner product space of Hermitian matrices with
the Hilbert--Schmidt inner product
(cf.~\cite{watroussdplecture} for the special case where the cones are $\psdcone$).
The difficulty is that the superoperator
$\rot$ is not Hermiticity-preserving, and so cannot be considered as a
linear map on the space of Hermitian matrices.
This is not hard to fix, as the condition $\rot(T) \in \cone$ requires $\rot(T)$ to be
Hermitian and so is equivalent to the pair
of conditions $\rot(T)-\rot(T)^\dag=0$ and $\rot(T)+\rot(T)^\dag \in \cone$.
The first of these can also be written $T - T^\ddag = 0$ (recall that
we define $X^\ddag = \rot(\rot(X)^\dag)$).
Note that the left-hand sides of these relations, seen as superoperators (e.g.\ $T \to T -
T^\ddag$), are not linear in the space $\linop{A} \ot \linop{A'}$ since they each contain an
anti-linear term.
They are, however, linear in the real inner product space of Hermitian matrices.
Within this space, the map $T \to \rot(T) + \rot(T)^\dag$ is self-adjoint.
Indeed, for Hermitian $L,T$ we have
\begin{align*}
    \ip{L, \rot(T) + \rot(T)^\dag}
    & = \ip{L, \rot(T)} +\ip{L, \rot(T)^\dag}
    \\ &= \ip{L, \rot(T)} + \ip{\rot(T), L}^*
    \\ &= \ip{\rot(L), T} + \ip{T, \rot(L)}^*
    \\ &= \ip{\rot(L), T} + \ip{\rot(L)^\dag, T}
    \\ &= \ip{\rot(L) + \rot(L)^\dag, T}.
\end{align*}
The map $T \to T-T^\ddag$ is also self-adjoint within the space of Hermitian matrices.
The primal becomes
\begin{align}
    \qthmperp{\cone}(S) = \max &\; \braopket{\Phi}{ I \ot \rho + T }{\Phi} \notag
    \\ \textrm{s.t. } &\; 1 - \Tr \rho = 0, \label{eq:qthm_primal_trrho}
    \\ &\; I \ot \rho + T \succeq 0, \label{eq:qthm_primal_rho_T_sum}
    \\ &\; T - T^\ddag = 0, \label{eq:qthm_primal_antidh}
    \\ &\; \rot(T)+\rot(T)^\dag \in \cone, \label{eq:qthm_primal_RC}
    \\ &\; \rho \succeq 0, T \in S \ot \conj{S}, \label{eq:qthm_primal_notcon}
\end{align}
Applying the recipe~\eqref{eq:gartner_duality} gives a dual formulation with a variable
for each constraint in the primal:
$\lambda$ for~\eqref{eq:qthm_primal_trrho},
$W$ for~\eqref{eq:qthm_primal_rho_T_sum}, $X$ for~\eqref{eq:qthm_primal_antidh},
and $L'$ for~\eqref{eq:qthm_primal_RC}.
In other words, $\yvec = \lambda \oplus W \oplus X \oplus L'$ (with these thought of as vectors
in the inner product space of Hermitian matrices).
The~\eqref{eq:qthm_primal_notcon} constraints correspond to the
$\xvec \in K$ constraint in~\eqref{eq:gartner_primal}, taking $\xvec = \rho \oplus T$.
The dual will have a constraint for each variable of the primal:~\eqref{eq:qthm_dual_rho} for
$\rho$ and~\eqref{eq:qthm_dual_T} for $T$.
The dual is then
\begin{align}
    \min &\; \lambda \notag
    \\ \textrm{s.t. } &\; \lambda I - \Tr_A W - I \succeq 0, \label{eq:qthm_dual_rho}
    \\ &\; -(W +X-X^\ddag +\rot(L')+\rot(L')^\dag)
    \notag \\ &\qquad -\proj{\Phi} \in (S \ot \conj{S})^\perp,
     \label{eq:qthm_dual_T}
    \\ &\; \lambda \in \mathbb{R}, W \succeq 0, X \textrm{ Hermitian}, L' \in \cone^*.
    \label{eq:qthm_dual_WX}
\end{align}
Define $Y = W + \proj{\Phi}$ and $L = \rot(L') + (X-X^\ddag)/2$.
Note that $L$ is not necessarily Hermitian, but $L'$ is since $L' \in \cone^*$.
We have $\rot(L) + \rot(L)^\dag = L' + L'^\dag + (\rot(X) - \rot(X^\ddag) +
\rot(X)^\dag - \rot(X^\ddag)^\dag)/2 = 2L' \in \cone^*$
since $\rot(X^\ddag) = \rot(X)^\dag$.
So these give a solution to
\begin{align}
    \min &\; \lambda\notag
    \\ \textrm{s.t. } &\; \lambda I - \Tr_A Y \succeq 0 \notag
    \\ &\; Y + (L + L^\dag) \in (S \ot \conj{S})^\perp, \notag
    \\ &\; \rot(L) + \rot(L)^\dag \in \cone^*, \notag
    \\ &\; Y \succeq \proj{\Phi}, \notag
    \\ &\; \lambda \in \mathbb{R}, L \in \linop{A} \ot \linop{A'}. \label{eq:qthm_dual_YL}
\end{align}
Conversely, a solution to~\eqref{eq:qthm_dual_YL} gives a solution
to~\eqref{eq:qthm_dual_rho}-\eqref{eq:qthm_dual_WX}
via $W = Y - \proj{\Phi}$,
$L' = (\rot(L) + \rot(L)^\dag)/2$,
$X = [(L - L^\ddag) + (L - L^\ddag)^\dag]/4$.
The program~\eqref{eq:qthm_dual_YL} is equivalent to~\eqref{eq:qthm_C_dual}.

We now show the primal and dual to have equal and finite optimum values.
Let $L'$ be in the relative interior\footnote{
    See~\cite{boyd2004convex} for the definition of relative interior.
} of $\cone^*$, and let $X=0$.
There is a $W \succ 0$ such that the left hand side of~\eqref{eq:qthm_dual_T} is
proportional to negative identity, and so is in $(S \ot \conj{S})^\perp$.
For a large enough $\lambda$,~\eqref{eq:qthm_dual_rho} is satisfied with strict
inequality.
Thus the dual program~\eqref{eq:qthm_dual_rho}-\eqref{eq:qthm_dual_WX} is strictly
feasible.
The dual has finite value because $W \succeq 0$ requires $\lambda \ge 1$
in~\eqref{eq:qthm_dual_rho}.
Therefore strong duality holds: the
primal~\eqref{eq:qthm_primal_trrho}-\eqref{eq:qthm_primal_notcon} and
dual~\eqref{eq:qthm_dual_rho}-\eqref{eq:qthm_dual_WX} are both feasible and take the same
optimal value.
Since these are equivalent to~\eqref{eq:qthm_C_appendix} and~\eqref{eq:qthm_dual_YL},
these two are also feasible and take the same value.


We now compute the primal for~\eqref{eq:qthp_C_dual}, which we rewrite here for reference.
\begin{align}
    \qthpperp{\cone}(S) = \min &\; \opnorm{\Tr_A Y} \notag
    \\ \textrm{s.t. } &\; Y \in S^\perp \ot \conj{S}^\perp, \notag
    \\ &\; \rot(Y) \in \cone, \notag
    \\ &\; Y \succeq \ket{\Phi}\bra{\Phi}.
    \label{eq:qthp_C_appendix}
\end{align}
As with $\qthmperp{\cone}$, we can rewrite this using only Hermiticity preserving
maps:
\begin{align}
    \qthpperp{\cone}(S) = \min &\; \lambda \notag
    \\ \textrm{s.t. } &\; \lambda I - \Tr_A Y \succeq 0, \label{eq:qthp_dual_trY}
    \\ &\; Y - \ket{\Phi}\bra{\Phi} \succeq 0, \label{eq:qthp_dual_YJ}
    \\ &\; Y - Y^\ddag = 0, \label{eq:qthp_dual_antidh}
    \\ &\; \rot(Y) + \rot(Y)^\dag \in \cone, \label{eq:qthp_dual_RC}
    \\ &\; \lambda \in \mathbb{R},
        Y \in S^\perp \ot \conj{S}^\perp. \label{eq:qthp_dual_notcon}
\end{align}
The primal will have a variable for each constraint in the dual:
$\rho$ for~\eqref{eq:qthp_dual_trY},
$T'$ for~\eqref{eq:qthp_dual_YJ},
$X$ for~\eqref{eq:qthp_dual_antidh}, and
$L'$ for~\eqref{eq:qthp_dual_RC}.
In other words, $\xvec = \rho \oplus T' \oplus X \oplus L'$.
The~\eqref{eq:qthp_dual_notcon} constraints correspond to $\yvec \in \mathcal{G}^*$
in~\eqref{eq:gartner_duality}.
The dual will have a constraint for each variable of the primal:~\eqref{eq:qthp_primal_lambda}
for $\lambda$ and~\eqref{eq:qthp_primal_Y} for $Y$.
The primal is then
\begin{align}
    \max &\; \braopket{\Phi}{T'}{\Phi} \notag
    \\ \textrm{s.t. } &\; 1 - \Tr \rho = 0, \label{eq:qthp_primal_lambda}
    \\ &\; I \ot \rho - T' - X + X^\ddag - \rot(L')
        \notag \\ &\qquad - \rot(L')^\dag \in
        (S^\perp \ot \conj{S}^\perp)^\perp, \label{eq:qthp_primal_Y}
    \\ &\; \rho \succeq 0, T' \succeq 0, X \textrm{ Hermitian},
        L' \in \cone^*. \label{eq:qthp_primal_TX}
\end{align}
Define $T = T' - I \ot \rho$
and $L = \rot(L') + (X-X^\ddag)/2$.
Note that $L$ is not necessarily Hermitian, but $L'$ is since $L' \in \cone^*$.
As before, we have $\rot(L) + \rot(L)^\dag \in \cone^*$.
These give a solution to
\begin{align}
    \max &\; \braopket{\Phi}{I \ot \rho + T}{\Phi} \notag
    \\ \textrm{s.t. } &\; \Tr \rho = 1, \notag
    \\ &\; T + (L + L^\dag) \in
        (S^\perp \ot \conj{S}^\perp)^\perp \notag
    \\ &\; \rot(L) + \rot(L)^\dag \in \cone^*, \notag
    \\ &\; \rho \succeq 0, I \ot \rho + T \succeq 0, \notag
    \\ &\; L \in \linop{A} \ot \linop{A'}. \label{eq:qthp_primal_TL}
\end{align}
Conversely, a solution to~\eqref{eq:qthp_primal_TL} gives a solution
to~\eqref{eq:qthp_primal_lambda}-\eqref{eq:qthp_primal_TX}
via $T' = T + I \ot \rho$,
$L' = (\rot(L) + \rot(L)^\dag)/2$,
$X = [(L - L^\ddag) + (L - L^\ddag)^\dag]/4$.
The program~\eqref{eq:qthp_primal_TL} is equivalent to~\eqref{eq:qthp_C_primal}.

We now show the primal and dual to have equal, but not necessarily finite, optimum values.
Let $L''$ be in the relative interior of $\cone^*$; then for any $c>0$, $L'=cL''$ is also
in the relative interior of $\cone^*$.  Let $X=0$ and $\rho=I/\dim(A)$.
For sufficiently small $c$, there is a $T' \succ 0$ such that the left hand side
of~\eqref{eq:qthp_primal_Y} vanishes.
Thus the primal~\eqref{eq:qthp_primal_lambda}-\eqref{eq:qthp_primal_TX} is strictly feasible.
If the primal takes finite optimum value, then by strong duality the dual is feasible and
takes the same value.
On the other hand, if the primal is unbounded (has infinite optimal value) then by weak
duality the dual is infeasible and so also has infinite value.
See \cref{thm:qthm_qthp_complete} for an example of such a case.



\begin{thebibliography}{10}
\providecommand{\url}[1]{#1}
\csname url@samestyle\endcsname
\providecommand{\newblock}{\relax}
\providecommand{\bibinfo}[2]{#2}
\providecommand{\BIBentrySTDinterwordspacing}{\spaceskip=0pt\relax}
\providecommand{\BIBentryALTinterwordstretchfactor}{4}
\providecommand{\BIBentryALTinterwordspacing}{\spaceskip=\fontdimen2\font plus
\BIBentryALTinterwordstretchfactor\fontdimen3\font minus
  \fontdimen4\font\relax}
\providecommand{\BIBforeignlanguage}[2]{{%
\expandafter\ifx\csname l@#1\endcsname\relax
\typeout{** WARNING: IEEEtran.bst: No hyphenation pattern has been}%
\typeout{** loaded for the language `#1'. Using the pattern for}%
\typeout{** the default language instead.}%
\else
\language=\csname l@#1\endcsname
\fi
#2}}
\providecommand{\BIBdecl}{\relax}
\BIBdecl

\bibitem{dsw2013}
R.~Duan, S.~Severini, and A.~Winter, ``Zero-error communication via quantum
  channels, noncommutative graphs, and a quantum {L}ov{\'a}sz number,''
  \emph{IEEE Transactions on Information Theory}, vol.~59, no.~2, pp.
  1164--1174, 2013.

\bibitem{arxiv:0906.2527}
\BIBentryALTinterwordspacing
R.~Duan, ``Super-activation of zero-error capacity of noisy quantum channels,''
  2009. [Online]. Available: \url{http://www.arxiv.org/abs/0906.2527}
\BIBentrySTDinterwordspacing

\bibitem{arxiv:1301.1166}
\BIBentryALTinterwordspacing
T.~Feng and S.~Severini, ``Quantum channels from association schemes,'' 2013.
  [Online]. Available: \url{http://www.arxiv.org/abs/1301.1166}
\BIBentrySTDinterwordspacing

\bibitem{arxiv:0709.2090}
\BIBentryALTinterwordspacing
S.~Beigi and P.~W. Shor, ``On the complexity of computing zero-error and
  {H}olevo capacity of quantum channels,'' 2008. [Online]. Available:
  \url{http://www.arxiv.org/abs/0709.2090}
\BIBentrySTDinterwordspacing

\bibitem{PhysRevLett.69.2881}
\BIBentryALTinterwordspacing
C.~H. Bennett and S.~J. Wiesner, ``Communication via one- and two-particle
  operators on {E}instein-{P}odolsky-{R}osen states,'' \emph{Phys. Rev. Lett.},
  vol.~69, pp. 2881--2884, Nov 1992. [Online]. Available:
  \url{http://link.aps.org/doi/10.1103/PhysRevLett.69.2881}
\BIBentrySTDinterwordspacing

\bibitem{PhysRevLett.70.1895}
\BIBentryALTinterwordspacing
C.~H. Bennett, G.~Brassard, C.~Cr\'epeau, R.~Jozsa, A.~Peres, and W.~K.
  Wootters, ``Teleporting an unknown quantum state via dual classical and
  {E}instein-{P}odolsky-{R}osen channels,'' \emph{Phys. Rev. Lett.}, vol.~70,
  pp. 1895--1899, Mar 1993. [Online]. Available:
  \url{http://link.aps.org/doi/10.1103/PhysRevLett.70.1895}
\BIBentrySTDinterwordspacing

\bibitem{witsen76}
H.~Witsenhausen, ``The zero-error side information problem and chromatic
  numbers (corresp.),'' \emph{Information Theory, IEEE Transactions on},
  vol.~22, no.~5, pp. 592--593, 1976.

\bibitem{dewolfphd}
\BIBentryALTinterwordspacing
R.~M. de~Wolf, ``Quantum computing and communication complexity,'' Ph.D.
  dissertation, University of Amsterdam, 2001. [Online]. Available:
  \url{http://homepages.cwi.nl/~rdewolf/publ/qc/phd.pdf}
\BIBentrySTDinterwordspacing

\bibitem{PhysRevA.88.062316}
\BIBentryALTinterwordspacing
M.~Nathanson, ``Three maximally entangled states can require two-way local
  operations and classical communication for local discrimination,''
  \emph{Phys. Rev. A}, vol.~88, p. 062316, Dec 2013. [Online]. Available:
  \url{http://link.aps.org/doi/10.1103/PhysRevA.88.062316}
\BIBentrySTDinterwordspacing

\bibitem{1705019}
J.~Nayak, E.~Tun\c{c}el, and K.~Rose, ``Zero-error source-channel coding with
  side information,'' \emph{Information Theory, IEEE Transactions on}, vol.~52,
  no.~10, pp. 4626--4629, 2006.

\bibitem{6994835}
J.~Bri\"{e}t, H.~Buhrman, M.~Laurent, T.~Piovesan, and G.~Scarpa,
  ``Entanglement-assisted zero-error source-channel coding,'' \emph{Information
  Theory, IEEE Transactions on}, vol.~61, no.~2, pp. 1124--1138, Feb 2015.

\bibitem{6880319}
T.~Cubitt, L.~Man\v{c}inska, D.~Roberson, S.~Severini, D.~Stahlke, and
  A.~Winter, ``Bounds on entanglement-assisted source-channel coding via the
  {L}ov{\'a}sz $\vartheta$ number and its variants,'' \emph{Information Theory,
  IEEE Transactions on}, vol.~60, no.~11, pp. 7330--7344, Nov 2014.

\bibitem{arxiv:1411.7666}
\BIBentryALTinterwordspacing
S.~Lu, ``No quantum brooks' theorem,'' 2014. [Online]. Available:
  \url{http://www.arxiv.org/abs/1411.7666}
\BIBentrySTDinterwordspacing

\bibitem{de2013optimization}
M.~K. de~Carli~Silva and L.~Tun{\c{c}}el, ``Optimization problems over
  unit-distance representations of graphs,'' \emph{The Electronic Journal of
  Combinatorics}, vol.~20, no.~1, p. P43, 2013.

\bibitem{david2006quantum}
D.~Avis, J.~Hasegawa, Y.~Kikuchi, and Y.~Sasaki, ``A quantum protocol to win
  the graph colouring game on all {H}adamard graphs,'' \emph{IEICE Transactions
  on Fundamentals of Electronics, Communications and Computer Sciences},
  vol.~89, no.~5, pp. 1378--1381, 2006.

\bibitem{cameron2007quantum}
P.~J. Cameron, A.~Montanaro, M.~W. Newman, S.~Severini, and A.~Winter, ``On the
  quantum chromatic number of a graph,'' \emph{Electron. J. Combin}, vol.~14,
  no.~1, 2007.

\bibitem{arxiv:1212.1724}
\BIBentryALTinterwordspacing
D.~E. Roberson and L.~Man\v{c}inska, ``Graph homomorphisms for quantum
  players,'' 2012. [Online]. Available:
  \url{http://www.arxiv.org/abs/1212.1724}
\BIBentrySTDinterwordspacing

\bibitem{hahn1997graph}
G.~Hahn and C.~Tardif, ``Graph homomorphisms: structure and symmetry,'' in
  \emph{Graph symmetry}.\hskip 1em plus 0.5em minus 0.4em\relax Springer, 1997,
  pp. 107--166.

\bibitem{HellNesetril200409}
P.~Hell and J.~Nesetril, \emph{Graphs and Homomorphisms (Oxford Lecture Series
  in Mathematics and Its Applications)}.\hskip 1em plus 0.5em minus 0.4em\relax
  Oxford {U}niversity {P}ress, USA, 9 2004.

\bibitem{lovasz79}
\BIBentryALTinterwordspacing
L.~Lov{\'a}sz, ``On the {S}hannon capacity of a graph,'' \emph{Information
  Theory, IEEE Transactions on}, vol.~25, no.~1, pp. 1 -- 7, jan 1979.
  [Online]. Available:
  \url{http://ieeexplore.ieee.org/xpls/abs_all.jsp?arnumber=1055985}
\BIBentrySTDinterwordspacing

\bibitem{lovaszsemidef}
\BIBentryALTinterwordspacing
L.~Lov\'{a}sz, ``\BIBforeignlanguage{English}{Semidefinite programs and
  combinatorial optimization},'' in \emph{\BIBforeignlanguage{English}{Recent
  Advances in Algorithms and Combinatorics}}, ser. CMS Books in Mathematics /
  Ouvrages de math\'{e}matiques de la SMC, B.~A. Reed and C.~L. Sales,
  Eds.\hskip 1em plus 0.5em minus 0.4em\relax Springer New York, 2003, pp.
  137--194. [Online]. Available:
  \url{http://dx.doi.org/10.1007/0-387-22444-0_6}
\BIBentrySTDinterwordspacing

\bibitem{PhysRevA.55.900}
\BIBentryALTinterwordspacing
E.~Knill and R.~Laflamme, ``Theory of quantum error-correcting codes,''
  \emph{Phys. Rev. A}, vol.~55, pp. 900--911, Feb 1997. [Online]. Available:
  \url{http://link.aps.org/doi/10.1103/PhysRevA.55.900}
\BIBentrySTDinterwordspacing

\bibitem{PhysRevLett.98.100502}
\BIBentryALTinterwordspacing
C.~B\'eny, A.~Kempf, and D.~W. Kribs, ``Generalization of quantum error
  correction via the {H}eisenberg picture,'' \emph{Phys. Rev. Lett.}, vol.~98,
  p. 100502, Mar 2007. [Online]. Available:
  \url{http://link.aps.org/doi/10.1103/PhysRevLett.98.100502}
\BIBentrySTDinterwordspacing

\bibitem{PhysRevA.73.052309}
R.~B. Griffiths, S.~Wu, L.~Yu, and S.~M. Cohen, ``Atemporal diagrams for
  quantum circuits,'' \emph{Phys. Rev. A}, vol.~73, no.~5, p. 052309, May 2006.

\bibitem{PhysRevA.84.032316}
\BIBentryALTinterwordspacing
D.~Stahlke and R.~B. Griffiths, ``Entanglement requirements for implementing
  bipartite unitary operations,'' \emph{Phys. Rev. A}, vol.~84, p. 032316, Sep
  2011. [Online]. Available:
  \url{http://link.aps.org/doi/10.1103/PhysRevA.84.032316}
\BIBentrySTDinterwordspacing

\bibitem{PhysRevA.84.022333}
\BIBentryALTinterwordspacing
C.~B\'eny and O.~Oreshkov, ``Approximate simulation of quantum channels,''
  \emph{Phys. Rev. A}, vol.~84, p. 022333, Aug 2011. [Online]. Available:
  \url{http://link.aps.org/doi/10.1103/PhysRevA.84.022333}
\BIBentrySTDinterwordspacing

\bibitem{knuth94}
\BIBentryALTinterwordspacing
D.~E. Knuth, ``The sandwich theorem,'' \emph{Electron. J. Combin.}, vol.~1,
  1994. [Online]. Available:
  \url{http://www.combinatorics.org/Volume_1/Abstracts/v1i1a1.html}
\BIBentrySTDinterwordspacing

\bibitem{arxiv:1002.2514}
\BIBentryALTinterwordspacing
R.~Duan, S.~Severini, and A.~Winter, ``Zero-error communication via quantum
  channels, non-commutative graphs and a quantum {L}ov\'{a}sz $\vartheta$
  function,'' 2010. [Online]. Available:
  \url{http://www.arxiv.org/abs/1002.2514}
\BIBentrySTDinterwordspacing

\bibitem{1056072}
A.~Schrijver, ``A comparison of the {D}elsarte and {L}ov{\'a}sz bounds,''
  \emph{Information Theory, IEEE Transactions on}, vol.~25, no.~4, pp.
  425--429, 1979.

\bibitem{mceliece1978lovasz}
R.~J. McEliece, E.~R. Rodemich, and H.~C. Rumsey~Jr., ``The {L}ov{\'a}sz bound
  and some generalizations,'' \emph{J. Comb. Inf. Syst. Sci}, vol.~3, no.~3,
  pp. 134--152, 1978.

\bibitem{365707}
M.~Szegedy, ``A note on the $\vartheta$ number of {L}ov{\'a}sz and the
  generalized {D}elsarte bound,'' in \emph{Foundations of Computer Science,
  1994 Proceedings., 35th Annual Symposium on}, 1994, pp. 36--39.

\bibitem{doi:10.1137/S1052623401383248}
\BIBentryALTinterwordspacing
E.~de~Klerk and D.~Pasechnik, ``Approximation of the stability number of a
  graph via copositive programming,'' \emph{SIAM Journal on Optimization},
  vol.~12, no.~4, pp. 875--892, 2002. [Online]. Available:
  \url{http://epubs.siam.org/doi/abs/10.1137/S1052623401383248}
\BIBentrySTDinterwordspacing

\bibitem{bomze2010gap}
I.~Bomze, F.~Frommlet, and M.~Locatelli, ``Gap, cosum and product properties of
  the $\theta'$ bound on the clique number,'' \emph{Optimization}, vol.~59,
  no.~7, pp. 1041--1051, 2010.

\bibitem{PhysRevA.66.062311}
\BIBentryALTinterwordspacing
L.~Gurvits and H.~Barnum, ``Largest separable balls around the maximally mixed
  bipartite quantum state,'' \emph{Phys. Rev. A}, vol.~66, p. 062311, Dec 2002.
  [Online]. Available: \url{http://link.aps.org/doi/10.1103/PhysRevA.66.062311}
\BIBentrySTDinterwordspacing

\bibitem{GodsilRoyle200105}
C.~Godsil and G.~F. Royle, \emph{Algebraic Graph Theory (Graduate Texts in
  Mathematics)}, 2001st~ed.\hskip 1em plus 0.5em minus 0.4em\relax Springer, 5
  2001.

\bibitem{6094278}
T.~Cubitt, J.~Chen, and A.~Harrow, ``Superactivation of the asymptotic
  zero-error classical capacity of a quantum channel,'' \emph{Information
  Theory, IEEE Transactions on}, vol.~57, no.~12, pp. 8114--8126, Dec 2011.

\bibitem{6157069}
T.~Cubitt and G.~Smith, ``An extreme form of superactivation for quantum
  zero-error capacities,'' \emph{Information Theory, IEEE Transactions on},
  vol.~58, no.~3, pp. 1953--1961, March 2012.

\bibitem{GartnerMatousek201201}
B.~G\"{a}rtner and J.~Matousek, \emph{Approximation Algorithms and Semidefinite
  Programming}, 2012th~ed.\hskip 1em plus 0.5em minus 0.4em\relax Springer, 1
  2012.

\bibitem{watroussdplecture}
\BIBentryALTinterwordspacing
J.~Watrous, ``Lecture 7: Semidefinite programming,'' 2011. [Online]. Available:
  \url{https://cs.uwaterloo.ca/~watrous/CS766/LectureNotes/07.pdf}
\BIBentrySTDinterwordspacing

\bibitem{boyd2004convex}
S.~P. Boyd and L.~Vandenberghe, \emph{Convex optimization}.\hskip 1em plus
  0.5em minus 0.4em\relax Cambridge university press, 2004.

\end{thebibliography}




\begin{IEEEbiographynophoto}{Dan~Stahlke}
    received an MS in Physics from the University of Alaska Fairbanks, in 2010, studying
    nonlinear dynamics under the supervision of Renate Wackerbauer, and a PhD in Physics
    from Carnegie Mellon University, in 2014, studying quantum information under the
    supervision of Robert B. Griffiths.
    He currently writes software at Intel Corporation.
\end{IEEEbiographynophoto}

\end{document}